\documentclass[phd]{nuthesis}

\usepackage{matlab-prettifier}

\usepackage{amsmath,amssymb,amscd,amsthm}
\usepackage{amsfonts, tikz-cd}
\usepackage{xypic}
\usepackage{breqn}
 \usepackage{listings}

\usepackage{microtype}

\usepackage{booktabs}

\usepackage{paralist}

\usepackage{graphicx}
\usepackage{color}
\definecolor{dark-red}{rgb}{0.6,0,0}
\definecolor{dark-green}{rgb}{0,0.6,0}
\definecolor{dark-blue}{rgb}{0,0,0.6}

\usepackage{rotating}

\allowdisplaybreaks

\usepackage[%
pdfauthor={NAME}, pdftitle={TITLE}, pdfsubject={Thesis}, pdfkeywords={LaTeX, Thesis,
University of Nebrska, Test}, linkcolor=dark-blue,
pagecolor=dark-green, citecolor=dark-blue, urlcolor=dark-red,
colorlinks=false, backref,
plainpages=false,
pdfpagelabels
]{hyperref}

\usepackage{memhfixc}

\newtheorem{theorm}{Theorem}[chapter]
\newtheorem*{thm*}{Theorem}

\newtheorem{lemma}[theorm]{Lemma}

\newtheorem{proposition}[theorm]{Proposition}

\newtheorem{conjecture}[theorm]{Conjecture}

\theoremstyle{definition}

\newtheorem{defint}[theorm]{Definition}

\newcommand{\R}{\mathbb{R}}

\definecolor{blue}{rgb}{0,.0, .81}

\begin{document}
\frontmatter

\title{Applications of Discrete Mathematics for Understanding Dynamics of Synapses and Networks in Neuroscience}
\author{Caitlyn M. Parmelee}
\adviser{Professor Carina Curto} \adviserAbstract{Professor Carina Curto}
\major{Mathematics} \degreemonth{August} \degreeyear{2016} 

\maketitle
\begin{abstract}

Mathematical modeling has broad applications in neuroscience whether we are modeling the dynamics of a single synapse or the dynamics of an entire network of neurons. In Part I, we model vesicle replenishment and release at the photoreceptor synapse to better understand how visual information is processed. In Part II, we explore a simple model of neural networks with the goal of discovering how network structure shapes the behavior of the network. 

Vision plays an important role in how we interact with our environments.
 To fully understand how visual information is processed requires an understanding of the way signals are transformed at the very first synapse: the ribbon synapse of photoreceptor neurons (rods and cones). These synapses possess a ribbon-like structure on which approximately 100 synaptic vesicles can be stored, allowing graded responses through the release of different numbers of vesicles in response to visual input. These responses depend critically on the ability of the ribbon to replenish itself as ribbon sites empty upon release. The rate of vesicle replenishment is thus an important factor in shaping neural coding in the retina. In collaboration with experimental neuroscientists we developed a mathematical model to describe the dynamics of vesicle release and replenishment at the ribbon synapse. 
 
To learn more about how network architecture shapes the dynamics of the network, we study a specific type of threshold-linear network that is constructed from a simple directed graph. These networks are particularly well suited for our study because the network construction guarantees that differences in dynamics arise solely from differences in the connectivity of the underlying graph.  By design, the activity of these networks is bounded and there are no stable fixed points. Computational experiments show that most of these networks yield limit cycles where the neurons fire in sequence. Can we predict the order in which the neurons fire? To this end, we devised an algorithm to predict the sequence of firing using the structure of the underlying graph. Using the algorithm we classify all the networks of this type on five or fewer nodes.

\end{abstract}

\begin{copyrightpage}
\end{copyrightpage}

\begin{dedication}
\centering
To my parents, to whom I owe everything.

\vspace{0.25in}
And to Matt, to whom I owe this.
\end{dedication}

\begin{acknowledgments}

I would like to first thank my advisor, Dr.\ Carina Curto, for her support and encouragement over the last few years. The amazing opportunities she has provided me have been invaluable. 

Another thank you goes out to the other members of my committee, Dr.\ Bo Deng, Dr.\ Vladimir Itskov, Dr.\ Jamie Radcliffe, and Dr.\ Wallace Thoreson, for their conversations and support. I would also like to thank our exceptional collaborators, Dr.\ Wallace Thoreson, Dr.\ Matthew Van Hook, and Dr.\ Katherine Morrison for their discussions and insights.

I would also like to thank the University of Nebraska--Lincoln Mathematics Department for providing the resources to pursue my research and develop my teaching. 
  
  There are so many people I am grateful to have met along my journey and owe my deepest thanks:
  my labmates, for helping me navigate the world of mathematical neuroscience,
my fellow ``first-years,'' for being an incredible support system and filling my life with fun facts, and my officemates, for creating a safe space and making me smile even on the most difficult days. 
 
A special thanks goes to my undergraduate advisor, Dr.\ Matt Koetz, and the other math faculty at Nazareth College for setting me on this journey and continuing to believe in me every step of the way.

I could not have done this without my partner in crime, Nate Clayburn. He has been a source of strength and love for the last five years. 
 I am strong if you are strong.

 Lastly, my eternal thanks to my family, especially my parents and sister, whose love fills my life and to whom I owe so much.

\end{acknowledgments}

\begin{grantinfo}
This work was supported in part by the National Science Foundation grants  DMS-1516881, DMS-1225666/DMS-1537228 and P200A120068. This work was also supported in part by a sub-award from the NIH grant R01-EY010542-19.
\end{grantinfo}


\tableofcontents*

\mainmatter


\part{Dynamics of Ribbon Synapses}\label{Ribbon_part}
	\chapter{Introduction to Part \ref{Ribbon_part}}

Vision plays an important role in how we interact with our environments. In fact, half of our cerebral cortex is dedicated to processing the visual world \cite{neurotext}.  
Part \ref{Ribbon_part} explores how visual information is processed at the very first synapse of the visual pathway, the photoreceptor synapse.  We specifically look at the role of a structure called the synaptic ribbon.

In this introductory section we will discuss some background information about the structure and function of the visual system.  Visual processing begins when light enters the retina and is absorbed by photoreceptor neurons. 
Photoreceptor neurons are the principal light-sensitive cells in the retina. When light enters the retina, it passes through several layers of cells\footnote{The other cells in the retina are relatively transparent, so when light passes through them there is very little image distortion.\cite{neurotext}} before being absorbed by the outer segments of photoreceptors (see Figure \ref{retina}).   The absorption of light initiates the process of phototransduction which ultimately triggers changes in membrane potential. These signals are passed to a layer of bipolar cells\footnote{Photoreceptors also synapse onto horizontal cells, which modify the bipolar cells laterally.}, and then to a layer of ganglion cells. The ganglion cells send axons to the optic nerve and are the only source of outputs from the retina. 

\begin{figure}[h]
\begin{center}
\includegraphics[width=50mm]{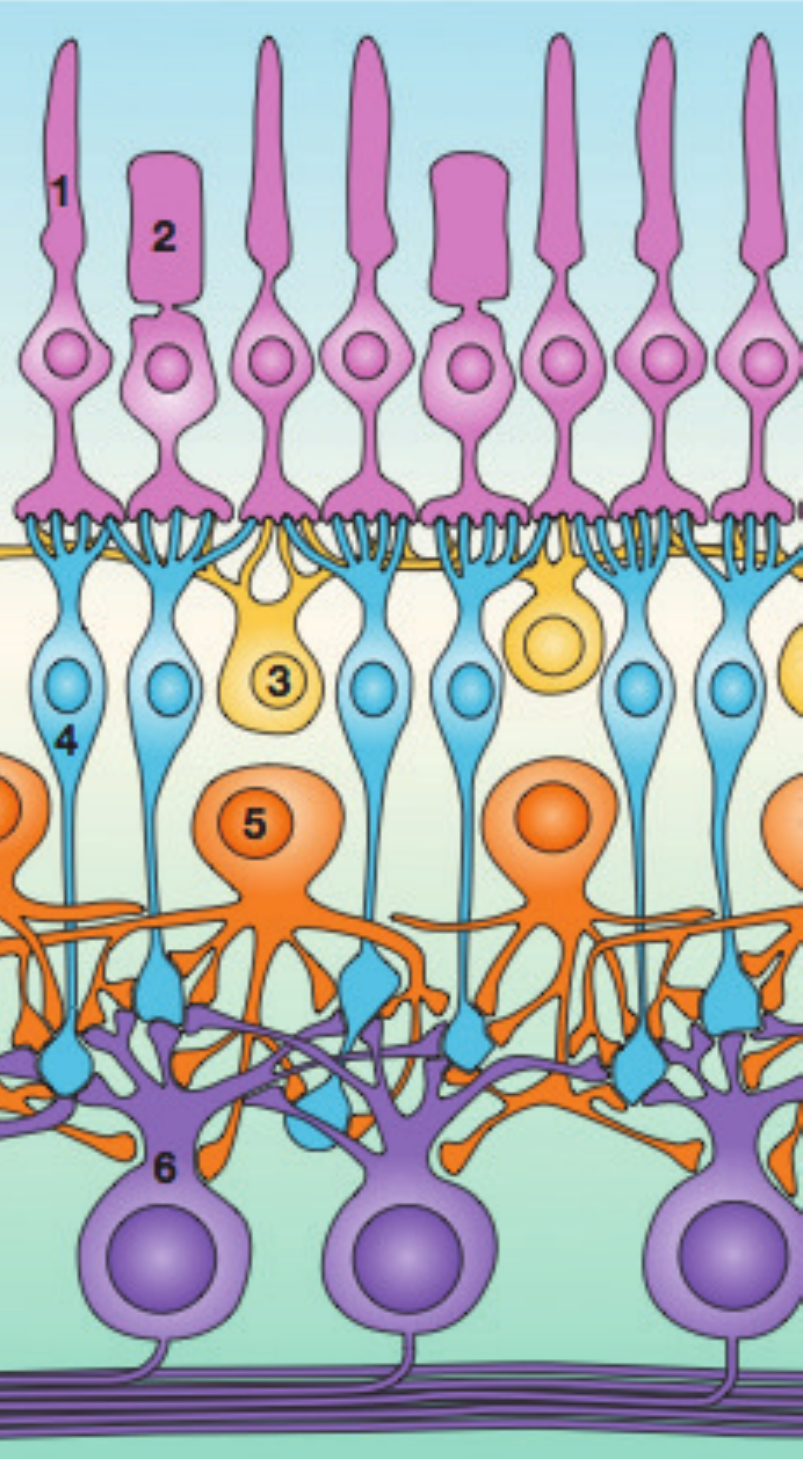}
\caption{Diagram of retinal pathway: (1) photoreceptor rods, (2) photoreceptors cones, (3) horizontal cells, (4) bipolar cells, (5) amacrine cells, (6) retinal ganglion cells. Adapted from \cite{retina}.}\label{retina}\end{center}
\end{figure}

To better understand vision, we wish to first understand how visual information is processed at the photoreceptor synapse.  To do this we need to set up some necessary background information about the biology of these synapses. We will start by discussing neurons and how they communicate with each other.  Then we will look more specifically at how photoreceptor neurons communicate.  Finally we will describe the synaptic ribbon, a specialized structure in photoreceptor synapses, and discuss what is currently known about its role in the vesicle cycle and information processing.

Following the introduction, Chapters \ref{releasereplenishmodel}-\ref{compmodel} discuss our work involving the ribbon synapse.  These results were obtained in collaboration with experimental neuroscientists in the Thoreson Lab at the University of Nebraska Medical Center. Our contribution has been to develop theoretical models describing the dynamics of release and replenishment in the ribbon synapse. Results from Sections \ref{dynamics}-\ref{ourmodel} were published in \cite{mine2} and results from Sections \ref{replenishment}-\ref{calcium} were published in \cite{mine1}.

\section{Synaptic transmission}\label{neurons}

Neurons are cells involved in the transmission of information in the nervous system. The neurons receive inputs from other neurons at the dendrites and once a threshold is reached the neuron can send a signal, often in the form of an action potential, down its axon to other cells. The pattern of action potentials codes the information being sent. This information transfer between the two cells takes place at the synapse. Figure \ref{neuron} shows a diagram of a conventional synapse.    
  The information is passed between neurons through the release of vesicles, which are small spheres made of membrane and packed with neurotransmitters. 
  
  \begin{figure}[h]
\begin{center}
\includegraphics[width=105mm]{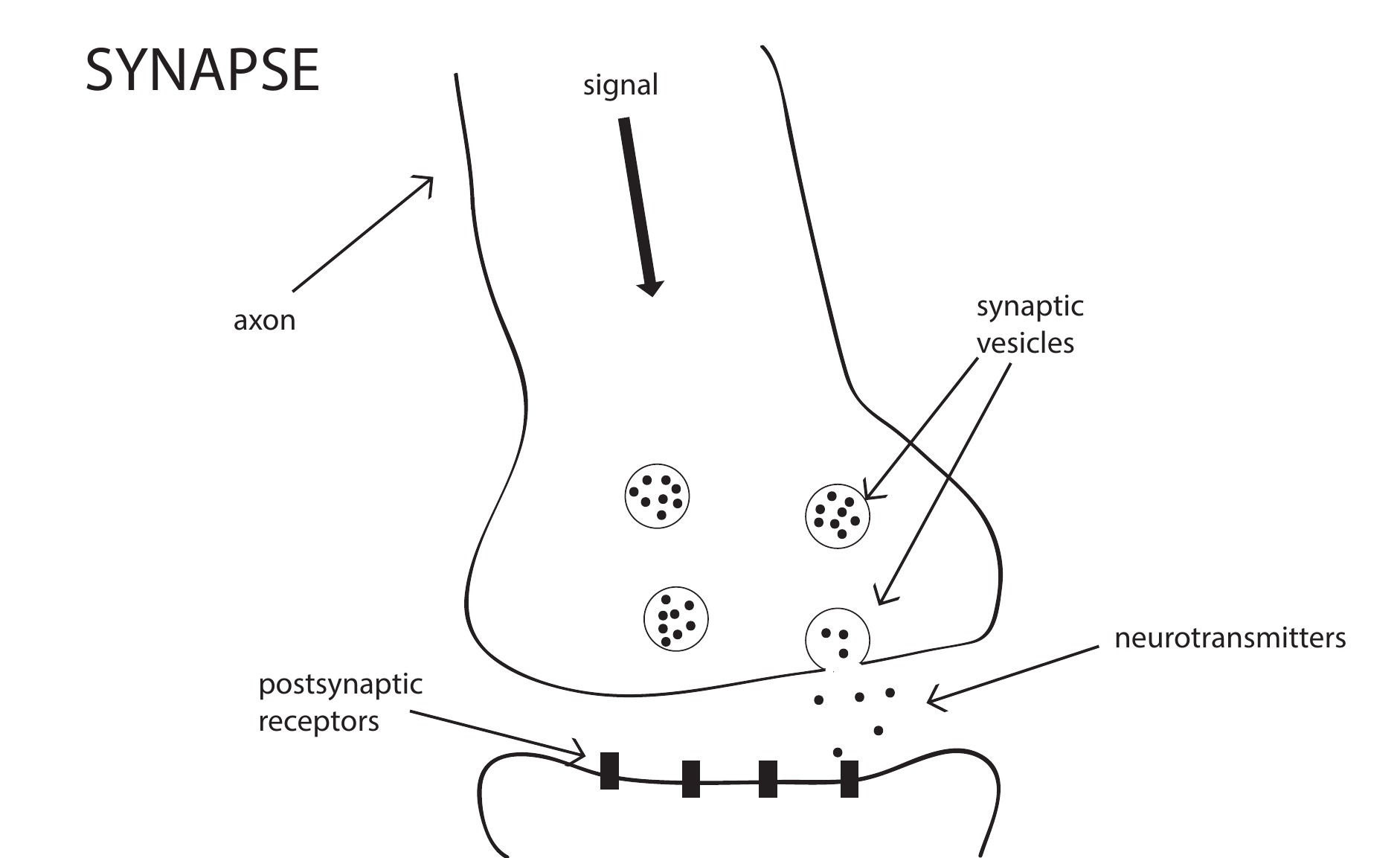}
\caption{Diagram of a synapse: When a signal reaches the terminal, it triggers vesicles to dock and fuse with the cell membrane and release neurotransmitters. These neurotransmitters travel across the synaptic cleft and  bind with the receptors on the postsynaptic cell.}\label{neuron}\end{center}
\end{figure}

When an action potential reaches the cell terminal, it triggers vesicle exocytosis. Vesicle exocytosis is the process in which the vesicles dock and fuse with the cell membrane and release their neurotransmitters into the synaptic cleft, the space between the two cells. The area of the cell membrane where this occurs is referred to as the ``active zone.''  The neurotransmitters then bind with the receptors on the postsynaptic cell, passing the information. For example, these receptors may open or close ion channels or activate second messenger systems.

Vesicles are recycled through endocytosis, which is the process by which vesicles are reformed using parts of the cell membrane and refilled with neurotransmitters. These recycled vesicles then become part of the mobile vesicle pool inside the cell.

\section{Photoreceptor neurons}\label{photoreceptors}

We are particularly interested in studying synaptic transmission at photoreceptor neurons. Photoreceptor neurons are the first cells of the visual system. Photoreceptors are located in the retina and their function is to convert light into changes in membrane potential. Light is absorbed by membranous disks, located on the outer segments, containing photopigment. There are two main types of photoreceptors: rods and cones. Rods are involved in night vision, motion detection, and peripheral vision and they are dense everywhere but the center of the eye. Cones are located in the center of the retina and are involved in color vision and detecting finer detail.  Unless otherwise specified, all the experimental data in Part \ref{Ribbon_part} refers only to cone photoreceptors, specifically in the aquatic tiger salamander.

In a conventional synapse the neuron responds to action potentials with discrete vesicle events.  Photoreceptor cells instead respond directly to the absorption of photons by releasing vesicles constantly in darkness and slowing release as light increases, i.e. the cell is depolarized in darkness and an increase in light causes the cell to hyperpolarize.   
The graded responses given by photoreceptor cells allow for a quicker processing of information as well as a larger range of responses \cite{diverseroles}.  This graded release is facilitated by a structure called the synaptic ribbon, described in the next section.
 
\section{Synaptic ribbon}\label{synapticribbon}

The synaptic ribbons present in cone photoreceptors are plate-like rectangular\footnote{Synaptic ribbons in different cells may have different shapes. For example, ribbons in the auditory system can be spherical or ellipsoidal rather than rectangular (\cite{ribbonshape},\cite{auditoryribbon}).} proteinaceous structures anchored to the inside of the cell membrane close to the Ca$^{2+}$ channels \cite{structure}. Daily and seasonal changes in the size, shape, number, and location of synaptic ribbons can occur based on light conditions \cite{diurnal}.  Vesicle release at the active zone is controlled by the opening and closing of the calcium channels. In cones, less than three channel openings are required to cause the fusion of a single vesicle, which allows for precise timing of release to accurately reflect changes in light intensity \cite{channels}. The increase in intracellular Ca$^{2+}$ also speeds the replenishment of vesicles, allowing for sustained release \cite{cone2}.

\begin{figure}[h]
\begin{center}
\includegraphics[width=150mm]{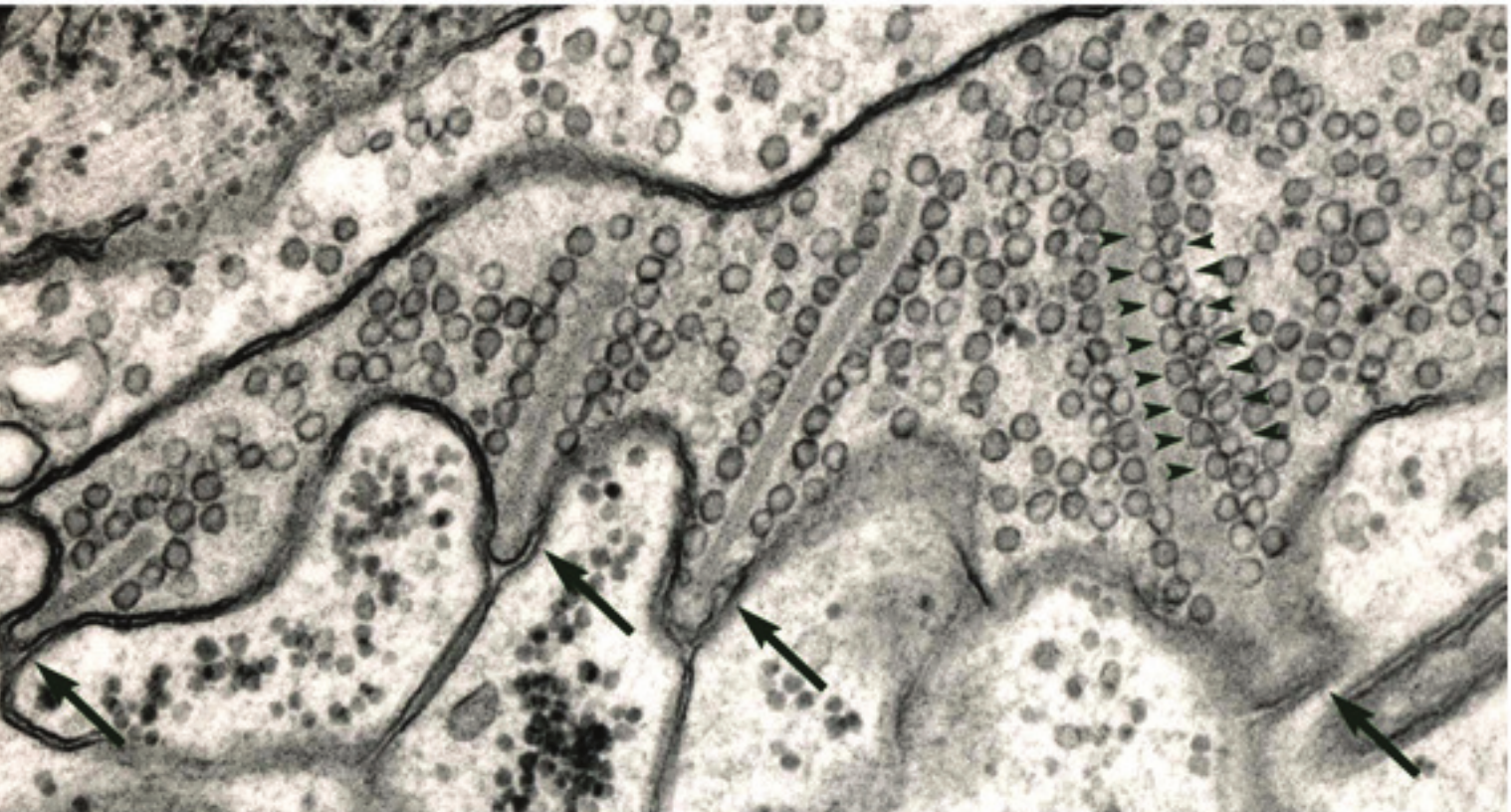}
\caption{Electron micrographs of synaptic ribbons in rod terminals: The larger arrows indicate the active zone at the bottom of each ribbon and the smaller arrows indicate the hexagonally packed vesicles tethered to the ribbon. Adapted from \cite{thoreson2004}.}\label{em}
\end{center}
\end{figure}

Recall that in conventional synapses vesicles dock and fuse directly with the cell membrane. The vesicles in ribbon synapses are instead first collected on the synaptic ribbon. In the cone photoreceptors of the aquatic tiger salamander there are approximately 11 rows of 5 vesicles stacked on each side of the ribbon, for a total of 110 vesicles \cite{poolsize}.  The vesicles become tethered to the ribbon via tiny filaments and then move along the ribbon towards the active zone. Not much is known about how the vesicles move down the ribbon to the active zone, but recent research posits that vesicles passively diffuse along the ribbon without an active transport mechanism \cite{diamond}. Once the vesicles reach the bottom two rows of the ribbon they are considered part of the rapidly releasable pool (RRP).  Experiments have shown that the ribbon may play a role in priming the vesicles for release \cite{ribbondestruction}. As a result, the RRP can be released almost immediately following the opening of calcium channels. Once the RRP is depleted, additional vesicles from the reserve pool on the ribbon take their place. Empty sites on the ribbon are refilled by the mobile vesicles in the cell terminal. See Figure \ref{fig:vesiclecycle} for a cartoon of the vesicle cycle in a ribbon synapse.

There are many theories regarding the function of the synaptic ribbon. The ribbon appears to support high rates of sustained vesicle release \cite{structure}, but how the ribbon achieves this is still an open question.  One theory is that the ribbon acts as a ``conveyor belt'' shuttling vesicles toward the active zone \cite{belt}. Another theory posits that it serves to hold the vesicles in contact with each other to facilitate multivesicular release via compound fusion \cite{belt}. Another theory asserts that the ribbon slows the delivery of vesicles, regulating the timing of release \cite{cone3}. Yet another proposes that the ribbon functions to store the vesicles close to the active zone \cite{singlevesiclemotion}.

\begin{figure}[h]
\begin{center}
                \includegraphics[width=130mm]{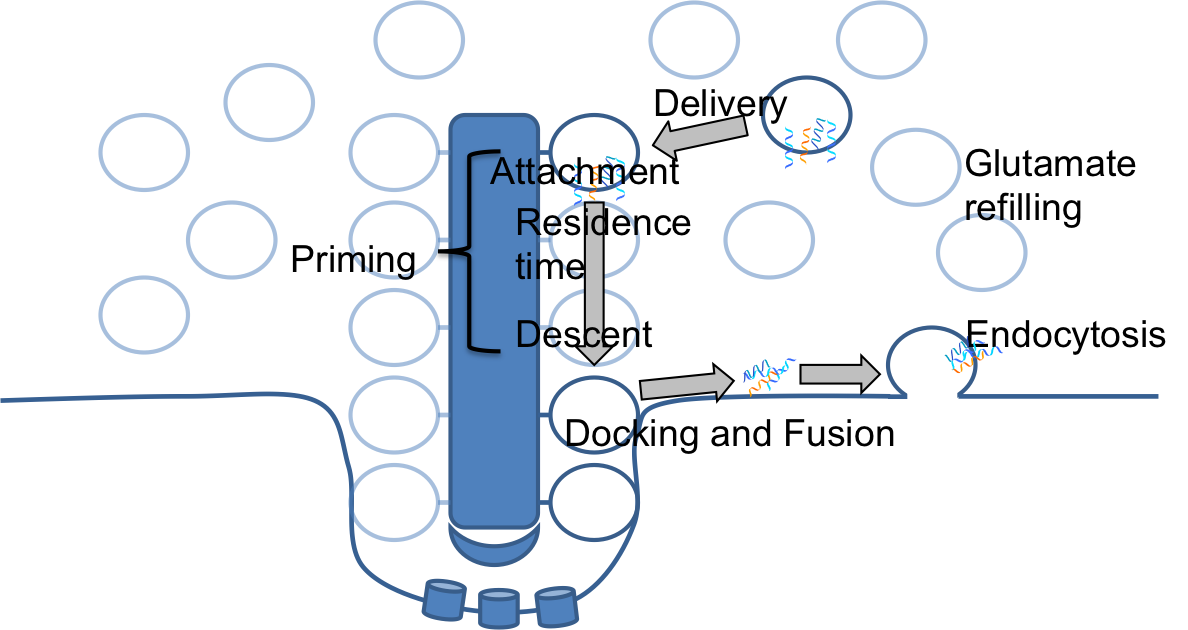}
                \caption{A cartoon depicting the vesicle cycle in the ribbon synapse \cite{ribbonfig}.}
                \label{fig:vesiclecycle}
\end{center}
\end{figure}

\vspace{-5mm}
\section{Questions about release and replenishment at the synaptic ribbon}

With the goal of better understanding how visual information is processed at the photoreceptor synapse in mind, we ask some questions about the role of the synaptic ribbon in the vesicle cycle of this synapse.

The number of vesicles released is stimulus-dependent, with stronger stimuli resulting in more vesicles released.  What causes this stimulus-dependence in the vesicle release? Does it depend solely on the probability of release or does it also depend the number of vesicles currently available on the ribbon?  To answer this question we created a model of release and replenishment using experimentally measured quantities to predict the unknown quantities of pool size and release probability. This allowed us to independently predict the pool size and the release probability to determine which changes with the stimulus strength. See Chapter \ref{releasereplenishmodel}.

There is an upper limit on the rate of sustained vesicle release. What is the rate-limiting factor for release? Studies indicate that vesicle replenishment is the rate-limiting step in sustained release \cite{cone3}, so we take a closer look at replenishment.  Vesicles move randomly in the cell terminal without a directed movement toward the ribbon or active zone \cite{cones}. This may be due to the fact that ribbon synapses lack synapsins, proteins that help bind vesicles to the actin cytoskeleton, allowing the vesicles to diffuse freely \cite{holt}.
Is this random motion of vesicles the rate-limiting step for replenishment?	To answer this question we designed a three-dimensional random walk model of vesicle replenishment to calculate the replenishment timescale.  We conclude that the random motion is not rate-limiting for replenishment.  See Section \ref{replenishment}.

It is known that Ca$^{2+}$ speeds the replenishment process \cite{cone2}. By what mechanism does Ca$^{2+}$ speed replenishment? To answer this question we modified our random walk model to test whether Ca$^{2+}$ acted on the ribbon or on the vesicles. We conclude that Ca$^{2+}$ affects the probability of attachment at sites on the ribbon rather than directly affecting vesicles. See Section \ref{calcium}.

To further study replenishment we ask two additional replenishment-related questions: (1) How many vesicles collide with the ribbon per second? and (2) How long does it take for the ribbon to fill? We can use our random walk model of replenishment to answer both. See Section \ref{otherquantities}.

Our random walk model does not take into account the geometry of the synaptic ribbon, so we designed a computational model to test the effects of ribbon geometry.  The computational model indicates that ribbon geometry does play a role in replenishment, so we also explore changes in local concentration near the ribbon and the effect of attachment probability on replenishment in an effort to explain the effects of geometry. The results appear in Section \ref{compmodel}.

\chapter{Model of vesicle release and replenishment}\label{releasereplenishmodel}

		Photoreceptors respond to changes in light by releasing vesicles from the ribbon. The amount of release depends on several key quantities: available pool size, release probability, and quantal amplitude.  The available pool size, $N$, is the number of vesicles on the ribbon that are primed and ready for release. The release probability, $P$, is the probability that a vesicle on the ribbon will be released. The quantal amplitude, $Q$, is the postsynaptic influence of a single vesicle.  We experimentally measure the response of photoreceptors to a given stimulus by measuring the postsynaptic currents\footnote{PSCs are generally measured in units of picoamps (pA).}  (PSCs) evoked in the postsynaptic cells onto which the photoreceptor synapses. Since the postsynaptic current is a linear sum of mini-EPSCs\footnote{ Mini-EPSCs (mEPSCs) are the change in current resulting from a single vesicle releasing its neurotransmitters.} \cite{mEPSC}, we can then estimate the number of vesicles released from the postsynaptic current using the quantal amplitude.
 
 The amount of release depends on the stimulus, so which of $N$, $P$, and/or $Q$ contribute to this stimulus-dependence? 
Changes in quantal amplitude occur on a longer timescale than our experiments \cite{quantal,quantal2, quantal3, quantal4}. Thus $Q$ cannot change quickly enough to be the cause of the stimulus-dependent changes in postsynaptic response.  Stimulus-dependent changes in postsynaptic response are often due to Ca$^{2+}$-dependent changes in $P$ \cite{mine2}, but it is also possible that stronger stimuli allow Ca$^{2+}$ to spread further up the ribbon, effectively increasing $N$. With these possibilities in mind, are the stimulus-dependent changes then due only to changes in the release probability, $P$, or are they a result of changes in $N$ as well?
 
In this section we will discuss a paradox that arises when asking this question. We then provide a model that  estimates $N$ and $P$ independently, based on experimental data, allowing us to resolve this paradox. We will also describe a generalization of the model.  The results in Sections \ref{dynamics}-\ref{ourmodel} are published in \cite{mine2}. The generalized model results in Section \ref{genmodel} are an unpublished extension of this work.

\section{Dynamics of release and replenishment}\label{dynamics}

Vesicles on the ribbon are released when the photoreceptor is stimulated and the amount released depends on the stimulus. As the vesicles are released, the empty sites on the ribbon are replenished by vesicles freely diffusing in the cell terminal. To study release and replenishment we model the available pool size during alternating periods of release and replenishment.  Figure \ref{fig:simplemodel} shows a cartoon of the model.  The variable $A(t)$ tracks the number of vesicles on the ribbon at time $t$. When the stimulus is on, $A(t)$ decreases due to vesicle release and when the stimulus is off, $A(t)$ increases due to vesicle replenishment. When designing the model, we make several key assumptions. We assume that no replenishment occurs during the periods of release. We also assume identical stimuli for the basic model, but the generalized model allows for multiple stimulus types in the same trial. By patching together the release and replenishment dynamics, using the ending value of one period as the initial value of the next, we can create a model of available pool size.

\begin{figure}
\begin{center}
		\includegraphics[width=120mm]{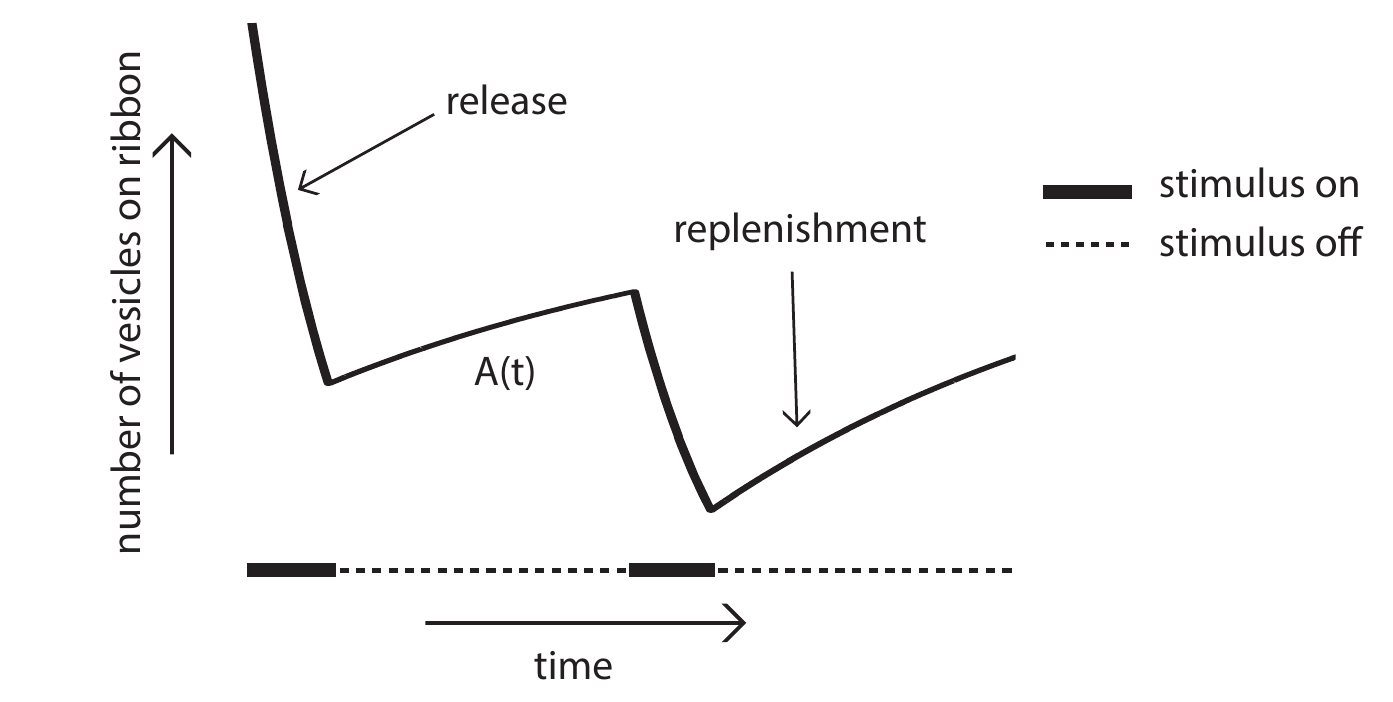}        
                \caption{Cartoon of release/replenishment model}
                \label{fig:simplemodel}
\end{center}
\end{figure}

\paragraph{Release dynamics.} The cumulative release at time $t$, $c(t)$, is governed by the differential equation:

\begin{equation}\dfrac{dc}{dt}=\dfrac{p_s(A_i-c)}{\tau_r}\end{equation}
where $p_s$ is the stimulus-specific probability of release, $A_i$ is the pool size at the beginning of the $i$th pulse, and $\tau_r$ is the time constant of release.  In salamander cones, the cumulative release curve can be fit by a two term exponential, one of the release time constants, $\tau_r$, is around $5$ ms and the other is too long to be accurately measured in our experimental setup, so we omit it from the model.  The timescale $\tau_r$ regulates the release for strong stimuli (e.g. steps to -19 mV), and for weaker stimuli (e.g. steps to -39 mV) the time constant is made effectively slower by the release factor $p_s$. Hence for strong stimuli we have $p_s=1$. 
Solving the differential equation yields \begin{equation} c(t)=A_i(1-e^{-p_st/\tau_r}).\end{equation}  Thus the release during the $i$th pulse is given by 
\begin{equation}R_i=c(\Delta t)=P_sA_i\end{equation}
where $P_s = 1-e^{-p_s\Delta t/\tau_r}$. Note that for 25 ms steps to -19 mV, $P_{-19} = 1-e^{-5} \approx 0.9933$. This is consistent with previous work showing that steps to -19 mV are strong enough to stimulate the release of nearly the entire pool of vesicles \cite{poolsize}. 

\paragraph{Replenishment dynamics} The cumulative replenishment at time $t$, $a(t)$, is governed by the differential equation:

\begin{equation}\dfrac{da}{dt}=\dfrac{n-a}{\tau_a}\end{equation}

\noindent where $n$ is number of sites unoccupied at the end of a pulse and $\tau_a$ is the time constant of replenishment. In salamander cones, replenishment is modeled with a two term exponential with time constants $\tau_{\text{fast}}=815$ ms and $\tau_{\text{slow}}=13$ s. Since the experiments occur on a much faster timescale than $\tau_{\text{slow}}$, we ignore $\tau_{\text{slow}}$ in our model.  Solving the differential equation yields \begin{equation} a(t)=n(1-e^{-t/\tau_a}).\end{equation} Thus the amount of replenishment after the $i$th pulse is given by \begin{equation} a(T)=n(1-\beta)\end{equation} where $\beta=e^{-T/\tau_a}$. Recall that we have chosen to omit the slow replenishment time constant, so after the initial release we assume only the fast-replenishing sites have time to fill during our replenishment period. Depending on the stimulus, we have a different fraction, $f_s$, of vesicles that are subject to the fast time constant. For steps to -19 mV, $f_{-19}=0.76$ and for steps to -39 mV, $f_{-39}=0.55$. Thus the number of available sites at the end of the $i$th pulse is $n = f_sA_s-A_{i}(1-P_s)$, where $A_s$ is the maximum pool size for stimulus $s$.

\section{Measuring available pool size}

\subsection{Pulse train experiments}\label{pulsetrain}
		To study the release and replenishment dynamics we consider a pulse train experiment. In this setup the cone is voltage-clamped near resting membrane potential and a steady train of pulses, i.e. voltage jumps, is applied to the presynaptic cell. Each pulse has duration $\Delta t$ ms, and the time between pulses is $T$ ms. During a given pulse the voltage jumps up to a chosen voltage step (e.g. a step to -19 mV) and between steps the voltage returns to -79 mV (see Figure \ref{fig:pulsetrain}).  The postsynaptic currents (PSCs) are measured in the postsynaptic horizontal cells.  Since postsynaptic response is the result of a linear sum of independent quantal release events (mEPSCs) \cite{mEPSC}, we can use the postsynaptic measurements to estimate the number of vesicles released from the cones\footnote{The mean amplitude of an mEPSC in the salamander retina, i.e. how much a single vesicle contributes, is around 6.5 $\pm$ 1.6 pA \cite{mEPSC}.}. 
Once the release and replenishment reach an equilibrium where release is limited by replenishment, we can measure the limiting release.  Our goal is to design a model that can predict maximum pool size and release probability using the first release and limiting release values measured during such an experiment.  

\begin{figure}
\begin{center}
		\includegraphics[width=120mm]{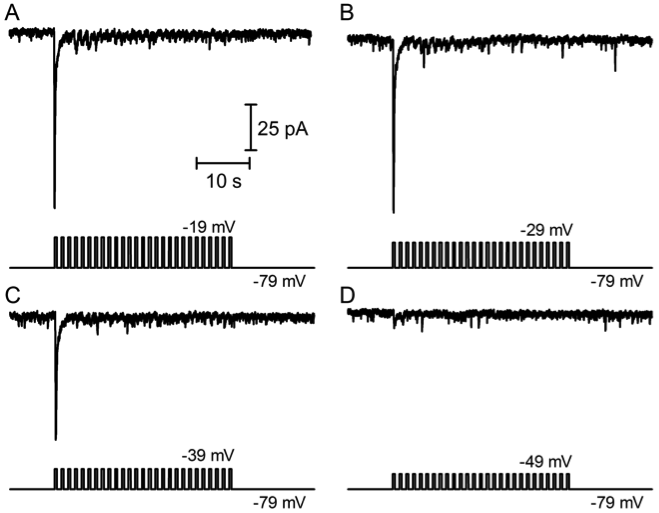}
                
                \caption{Pulse trains for voltage jumps to -19 mV,-29 mV, -39 mV, and -49 mV. Note that the stronger pulses have larger first release peaks \cite{mine2}.}
                \label{fig:pulsetrain}
\end{center}
\end{figure}

\subsection{An apparent paradox}
Previous work uses a method of back-extrapolation to estimate the maximum pool size, $A$ \cite{sakaba}.  This method considers the cumulative release curve and fits a line to the steady state response that occurs when release is limited by replenishment. Back-extrapolating to the time 0 gives an estimate of $A$. This method predicts that the maximum pool size is significantly smaller for weaker stimuli (see Figure \ref{backextrap}).  
The amplitude of the releasable pool predicted by back-extrapolation is 80 pA for -39 mV, 105 pA for -29 mV, and 132 pA for -19 mV \cite{mine2}.

\begin{figure}

\begin{center}

\includegraphics[width=150mm]{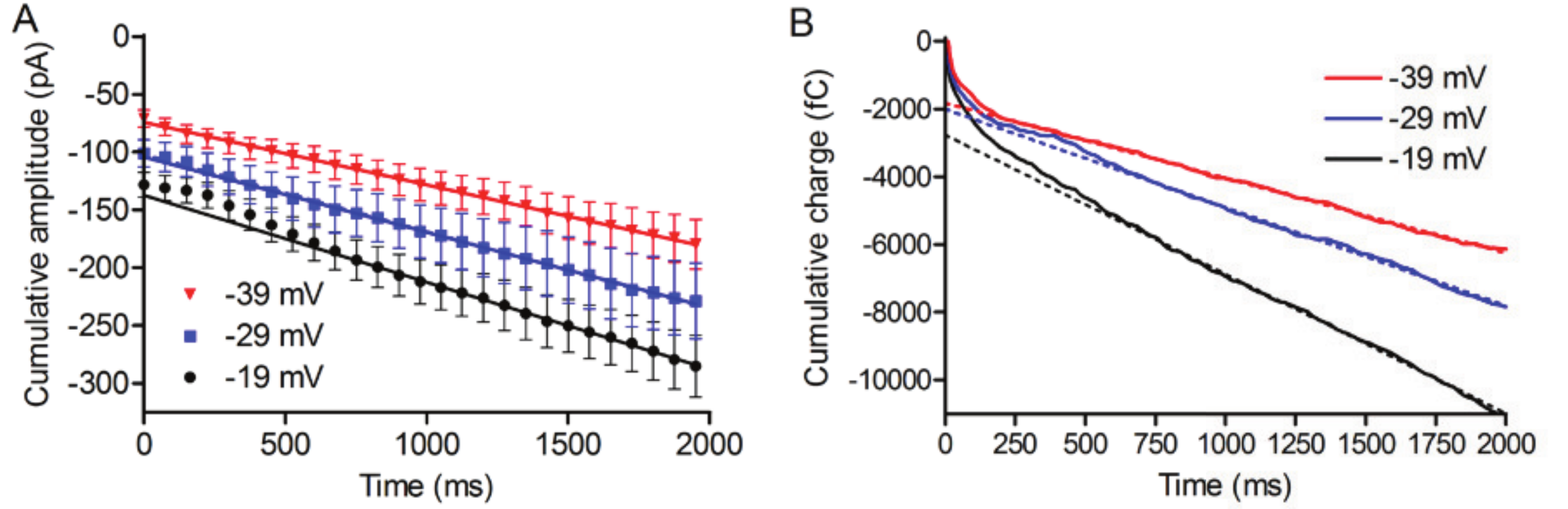}
\caption{Back-extrapolation method for -19 mV, -29 mV, and -39 mV pulses: Panel A shows a plot of cumulative amplitude in pA and Panel B show a plot of cumulative charge in fC \cite{mine2}. To predict maximum pool size we back-extrapolate from the steady state to $t=0$.}\label{backextrap}
\end{center}
\end{figure}

One of the pitfalls of the back-extrapolation method is that it assumes that the replenishment rate is constant. As we saw in Section \ref{dynamics}, the replenishment rate is certainly not constant in salamander cones. This causes the method to underestimate the maximum available pool size since replenishment is faster when the ribbon has more available space. The method is close for the stronger stimuli because the pulse train stabilizes to the limiting release right away when exposed to a strong stimulus. It also does not take into account the fast and slow replenishing sites. After the first pulse, slow sites don't have time to fill between pulses, so the limiting release reflects only the replenishment of the fast sites.

\section{Using our model to predict pool size}\label{ourmodel}

In our model, we let both the maximum pool size $A_s$ and the release probability $P_s$ vary independently to determine which causes the voltage-dependent changes in vesicle release. We want to estimate both $A_s$ and $P_s$ in terms of the measured first release $(R_1)_s$ and limiting release $R_s$ for each stimulus $s$.

\subsection{Derivation of pool size and release probability formulas}\label{derivations}

Let $A_i$ be the pool size at the beginning of the $i$th pulse, $c_i(t)$ be the cumulative release $t$ milliseconds into the $i$th pulse, and $a_i(t)$ be the cumulative replenishment $t$ seconds into the $i$th pulse.
 Then we can compute $A_i$ by taking the pool size at the beginning of the previous pulse, subtracting the release during that pulse, and adding the amount replenished before the $i$th pulse, i.e. $A_i=A_{i-1}-c_{i-1}(\Delta t)+a_{i-1}(T)$.
 \begin{figure}[h]
\begin{center}
\includegraphics[width=150mm]{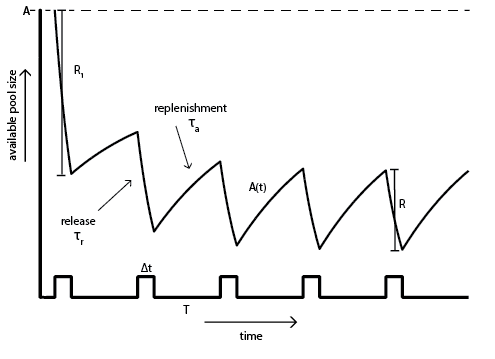}
\caption{Schematic of the pulse train setup \cite{mine2}: $A(t)$ keeps track of pool size at time $t$. Stimulus pulses of duration $\Delta t$ cause vesicle release and between pulses we have replenishment periods of duration $T$. Note that the pool size decreases during the pulses due to vesicle release and increases in between the pulses due to replenishment. The maximum possible pool size is denoted by $A$.}\label{schematic}
\end{center}
\end{figure} Note that $c_i(\Delta t) = A_iP_s$ and $a_i(T) = (f_sA_s-A_{i-1}(1-P_s))(1-\beta)$.  Thus 
\begin{equation} A_i = bA_{i-1}+c
\end{equation}
where $b = \beta(1-P_s)$ and $c = f_sA_s(1-\beta)$.

Solving the recursion we get \begin{equation} A_i = A_1b^{i-1}+c\dfrac{1-b^{i-1}}{1-b}.\end{equation}
The details of solving the recursion appear in Lemma \ref{recursionproof}.
Taking the limit as $i\rightarrow\infty$ gives \begin{equation}
A_\infty = \lim_{i\rightarrow\infty}A_i = \lim_{i\rightarrow\infty}\left(A_1b^{i-1}+c\dfrac{1-b^{i-1}}{1-b}\right)=\dfrac{c}{1-b}=\dfrac{f_sA_s(1-\beta)}{1-\beta(1-P_s)},\end{equation}
which represents the pool size at the beginning of each pulse during the steady state.
Since $R_i=P_sA_i$ is the amount released during the $i$th pulse, then the limiting release is given by \begin{equation} R_s=P_sA_\infty=\dfrac{P_sf_sA_s(1-\beta)}{1-\beta+\beta P_s}.\end{equation}
Solving for $A_s$ yields \begin{equation} A_s = \dfrac{R_s(1-\beta+\beta P_s)}{P_sf_s(1-\beta)}=\dfrac{R_s}{f_s}\left(\dfrac{1}{P_s}+\dfrac{\beta}{1-\beta}\right).\end{equation}
Also, note that $(R_1)_s=A_1P_s=A_sP_s$. Solving for $P_s$ and substituting into our equation for $A_s$ gives \begin{equation}\label{A_s} A_s=\left(\dfrac{\beta}{1-\beta}\right)\dfrac{R_s(R_1)_s}{(f_s(R_1)_s-R_s)}.\end{equation}
Using the fact that $(R_1)_s=A_sP_s$, we can also find an expression for $P_s$, \begin{equation}\label{P_s} P_s=\left(\dfrac{1-\beta}{\beta}\right)\dfrac{f_s(R_1)_s-R_s}{R_s}.\end{equation}
Equations \ref{A_s} and \ref{P_s} give formulas for independently estimating the maximum pool size and release probability for each stimulus given the first release and limiting release.

\subsection{Estimating pool size and release probability from data}

\noindent During pulse trains with steps to -19 mV and -39 mV we measure the amplitude of the first pulse as well as the limiting release. Limiting release is estimated by measuring the cumulative increase in amplitude 1--2 seconds into the pulse train \cite{mine2}. With $f_{-19}=0.76$, $f_{-39}=0.55$, $\Delta t=25$ms, $\tau_a=815$ms, and $\tau_r=5$ms, we can estimate the pool size for the two stimulus types. The results are recorded in Table \ref{poolsizetable}.  Note that although the amplitude of the first pulse varies significantly with stimulus strength, the predicted pool sizes are roughly the same. Using our formula to predict the release probabilities in the 5 mM EGTA cases, we see $P_{-19}\approx1$ for the strong stimuli and $P_{-39}\approx0.5$ for the weak stimuli.  This supports the hypothesis that changes in release probability alone cause voltage-dependent changes in release.  Note that the first release during strong pulses is nearly identical to the estimated pool size, consistent with a release probability of 1. Also recall that back-extrapolation predicted a pool size of 80 pA for a stimulus of -39 mV while our model estimates that the pool size is closer to 131 pA. 

\begin{table}\begin{center}
\begin{tabular}{|c|c|c|c|c|c|}
\hline
\footnotesize stimulus & \footnotesize EGTA & \footnotesize $T$ & \footnotesize PSC amplitude  & \footnotesize Predicted pool size, $A_s$ & \footnotesize Ratio \\
& & &  \footnotesize (first pulse, $R_1$) & \footnotesize (PSC amplitude) & \footnotesize $A_{-39}/A_{-19}$\\
\hline
-19 mV & 5 mM & 50 ms & 128.2$\pm$10.9 pA & 131.3 pA & 1.0\\
-39 mV & 5 mM & 50 ms & 70.9$\pm$7.4 pA & 131.2 pA& \\
-19 mV & 5 mM & 125 ms & 135.5$\pm$ 15.8 pA & 136.9 pA& 0.96 \\
-39 mV & 5 mM & 125 ms & 71.3$\pm$ 12.8 pA & 131.2 pA &\\
-19 mV & 0.05 mM & 50 ms & 91.1$\pm$ 18.2 pA & 110.9 pA &1.02\\
-39 mV & 0.05 mM & 50 ms & 38.5$\pm$ 9.5 pA & 113.6 pA &\\
\hline
\end{tabular}
\end{center}
\caption{Pool size predictions for several experimental conditions \cite{mine2}.}\label{poolsizetable}
\end{table}

Additional experiments were done with a weaker Ca$^{2+}$ buffer of 0.05mM EGTA compared to 5 mM EGTA. These results were similar to those with 5 mM EGTA (see Table \ref{poolsizetable}). The slightly smaller first responses in the experiments with 0.05 mM EGTA are likely due to the smaller number of ribbon contacts per postsynaptic HC (an average of 2.79 and 2.95 ribbon contacts for experiments with 5 mM EGTA versus an average of 1.98 ribbon contacts for experiments with 0.05 mM EGTA)\cite{mine2}.  This provides additional evidence that the increased spread of Ca$^{2+}$ does not increase the available pool size.

		\section{Generalization of release/replenishment model}\label{genmodel}
		In this section, we generalize the model from Section \ref{ourmodel}. We originally assumed that all pulses were of equal strength and duration. In the generalized model, both pulses and replenishment periods can vary in duration and pulses can also vary in stimulus strength. The results in this section make use of several lemmas and a proposition whose statements and proofs appear in Section \ref{computations}. 
		\subsection{Generalized pulse trains}
		\begin{figure}[h]
\begin{center}
                \includegraphics[width=150mm]{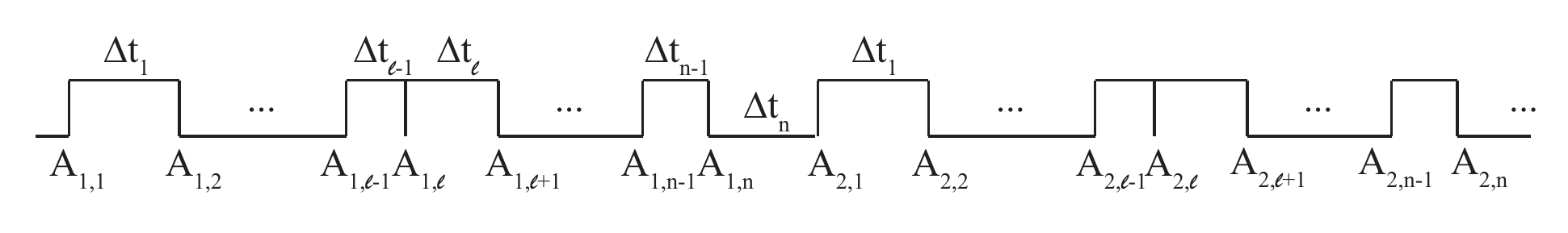}
           
                \caption{Generalized model setup: $A_{i,\ell}$ denotes available pool size at the beginning of the $\ell$th period of the $i$th cycle and $\Delta t_\ell$ is the duration of the $\ell$th period.}
                \label{fig:gensetup}
\end{center}
\end{figure}
		
We consider a generalized pulse train with $n$ periods of release and replenishment. The length of the $\ell$th period is $\Delta t_{\ell}$. We do not necessarily alternate between release and replenishment.  Once we have cycled through all $n$ periods we started at the beginning and repeat the periods in the same order. We denote the available pool size at the beginning of the $\ell$th period of the $i$th cycle by $A_{i,\ell}$ and assume the ribbon is full at $t=0$, so $A=A_{1,1}$ is the maximum available pool size. See Figure \ref{fig:gensetup}. Let $a_{i,\ell}(t)$ be the total change in pool size $t$ seconds into period $\ell$ of cycle $i$ with $a_{i,\ell}(0)=0$. If index $j$ denotes a release period then the dynamics is governed by \begin{equation*} \dfrac{da_{i,j}}{dt}=-\dfrac{1}{\tau_j}(A_{i,j}+a_{i,j}).\end{equation*}  If index $k$ denotes a replenishment period then the dynamics is governed by \begin{equation*}\dfrac{da_{i,k}}{dt}=\dfrac{1}{\tau_{k}}(A-(A_{i,k}+a_{i,k})).\end{equation*}
\subsection{Setting up and solving the recursion}
Note that this setup gives 
\begin{eqnarray}\label{gensetup}
A_{i,\ell} &=& A_{i,\ell-1}+a_{i,\ell-1}(\Delta t_{\ell-1})~~\text{for}~1<\ell\le n~\text{and } i\ge1\\ 
A_{i,1} &=& A_{i-1,n}+a_{i-1,n}(\Delta t_n) ~~~~~~~\text{for}~i>1\\\label{gensetup2}
A_{1,1} &=& A
\end{eqnarray}
Then we get \begin{equation}\label{recursion} \displaystyle A_{i,1}=A_{i-1,1}+\sum_{\ell=1}^{n} a_{i-1,\ell}(\Delta t_\ell).\end{equation}  In release periods $a_{i,j}(t)=-A_{i,j}(1-e^{-t/\tau_{j}})$ and in replenishment periods $a_{i,k}(t)=(A-A_{i,k})(1-e^{-t/\tau_{k}})$. Let $\alpha_\ell=e^{-\Delta t_\ell/\tau_\ell}$ so we can rewrite the terms in the sum using \begin{equation}\label{recursion2} a_{i,\ell}(\Delta t_\ell)=(\theta(\ell)A-A_{i,\ell})(1-\alpha_\ell)\end{equation} where \begin{equation*}
   \theta(\ell) = \left\{
     \begin{array}{ll}
       1 & \text{if } \ell\text{ indexes a replenishment period}\\
       0 & \text{if }\ell\text{ indexes a release period}
     \end{array}
   \right.
\end{equation*} Using Equations \ref{gensetup}-\ref{gensetup2} and simplifying we can write Equation \ref{recursion} in terms of $A_{i-1,1}$:
\begin{dmath}\label{genrecursion}
\displaystyle A_{i,1}=A_{i-1,1}\left(\prod_{\ell=1}^{n} \alpha_\ell\right)+A\left(\sum_{\ell\in L}(1-\alpha_\ell)\prod_{r=\ell+1}^n\alpha_r\right)
\end{dmath}
for $i\ge2$ where $\alpha_\ell=e^{-\Delta t_\ell/\tau_\ell}$, and $L=\{\ell:\theta(\ell)=1\}$. See the proof of Proposition \ref{prop1} in Section \ref{computations} for the details. Now we have a recurrence in $A_{i,1}$, which we can solve using ordinary generating functions. Note that the recurrence is of the form $A_{i}=bA_{i-1}+c$ where $b=\prod_{\ell=1}^{n} \alpha_\ell$ and $c=A\sum_{\ell\in L}(1-\alpha_\ell)\prod_{r=\ell+1}^n\alpha_r$. Then ordinary generating functions gives a solution of \begin{equation*} A_{i,1}=Ab^{i-1}+c\left(\dfrac{b^{i-1}-1}{b-1}\right).\end{equation*} Note that since $b=\prod_{\ell=1}^n \alpha_\ell<1$, then \begin{equation*} A_{\infty,1}=\displaystyle\lim_{i\rightarrow\infty} A_{i,1}=\dfrac{c}{1-b}=\displaystyle A\dfrac{\sum_{\ell\in L}(1-\alpha_\ell)\prod_{r=\ell+1}^n\alpha_r}{1-\prod_{\ell=1}^n \alpha_\ell}.\end{equation*} Thus we have a limit cycle.

\medskip

\noindent For $\ell$ s.t. $\theta(\ell)=0$, we have that the total release during that period is given by \begin{equation*}R_{i,\ell}=A_{i,\ell}(1-\alpha_\ell).\end{equation*} Lemma \ref{lemma1} gives a closed formula for $A_{i,\ell}$ in terms of $A_{i,1}$, so \begin{equation*} R_{i,\ell}=A_{i,1}(1-\alpha_\ell)\prod_{r=1}^{\ell-1}\alpha_r+A\sum_{r=1}^{\ell-1} \theta(r)(1-\alpha_r)(1-\alpha_\ell)\prod_{s=r+1}^{\ell-1}\alpha_s.\end{equation*}

\subsection{Special cases of the generalized model}

\paragraph{Release only.}
Consider the case where each period is a release period. Then $L=\emptyset$. So for the recursion we have $$A_{i,1}=\left(\prod_{\ell=1}^n\alpha_\ell\right) A_{i-1,1},$$ which has solution $$A_{i,1}=A\left(\prod_{\ell=1}^n\alpha_\ell\right)^{i-1}.$$ When we take the limit we get $$A_{\infty,1}=\displaystyle\lim_{i\rightarrow\infty} A_{i,1}=A\cdot0=0$$ as expected.

\paragraph{Replenishment only.}
Consider the case where each period is a replenishment period. Then $L=\{1,\ldots,n\}$. So for the recursion we have $$A_{i,1}=\left(\prod_{\ell=1}^n\alpha_\ell\right)A_{i-1,1}+A\sum_{\ell=1}^n (1-\alpha_\ell)\prod_{r=\ell+1}^n\alpha_r=\left(\prod_{\ell=1}^n\alpha_\ell\right)A_{i-1,1}+A\left(1-\prod_{\ell=1}^n\alpha_\ell\right),$$ which has solution $$A_{i,1}=A\left(\prod_{\ell=1}^n\alpha_\ell\right)^{i-1}+A\left(1-\left(\prod_{\ell=1}^n\alpha_\ell\right)^{i-1}\right).$$ When we take the limit we get $$A_{\infty,1}=\displaystyle\lim_{i\rightarrow\infty} A_{i,1}=A$$ as expected.

\paragraph{Alternating release and replenishment periods.}
Consider the case where we have alternating periods of release and replenishment, starting with release. Then $L=\{2,4,\ldots,n\}$. So for the recursion we have $$A_{i,1}= \left(\prod_{\ell=1}^n\alpha_\ell\right)A_{i-1,1}+A\sum_{\ell \text{ even}} (1-\alpha_\ell)\prod_{r=\ell+1}^n=\left(\prod_{\ell=1}^n\alpha_\ell\right)A_{i-1,1}+A\sum_{\ell=1}^n(-1)^\ell\prod_{r=\ell+1}^n\alpha_r,$$ which has solution $$A_{i,1}=A\left(\prod_{\ell=1}^n\alpha_\ell\right)^{i-1}+\left(A\sum_{\ell=1}^n(-1)^\ell\prod_{r=\ell+1}^n\alpha_r\right)\left(\dfrac{1-\left(\prod_{\ell=1}^n\alpha_\ell\right)^{i-1}}{1-\prod_{\ell=1}^n\alpha_\ell}\right).$$ When we take the limit we get $$A_{\infty,1}=\displaystyle\lim_{i\rightarrow\infty} A_{i,1}=\dfrac{A\sum_{\ell=1}^n (-1)^\ell\prod_{r=\ell+1}^n \alpha_r}{1-\prod_{\ell=1}^n \alpha_\ell}.$$ We can calculate the total release during release period $\ell$ by $$R_{i,\ell}=A_{i,1}(1-\alpha_\ell)\prod_{r=1}^{\ell-1}\alpha_r+A(1-\alpha_\ell)\sum_{r=1}^{\ell-1} (-1)^r \prod_{s=r+1}^{\ell-1} \alpha_s.$$

\paragraph{Alternating release and replenishment, original model. }
In the case of our experiment we have $\alpha_1=\alpha$, $\alpha_2=\beta$, $L=\{2\}$, and $n=2$. So for the recursion we have $$A_{i,1}=A_{i-1,1}\alpha\beta+A(1-\beta),$$ which has solution $$A_{i,1}=A(\alpha\beta)^{i-1}+A(1-\beta)\dfrac{1-(\alpha\beta)^{i-1}}{1-\alpha\beta}.$$ Taking the limit we get $$A_{\infty,1}=\displaystyle\lim_{i\rightarrow\infty} A_{i,1}=\dfrac{A(1-\beta)}{1-\alpha\beta}.$$ Note that $R_{i,1}=A_{i,1}(1-\alpha)$ for all $i$. So $$R_{i,1}=A(\alpha\beta)^{i-1}(1-\alpha)+A(1-\beta)(1-\alpha)\dfrac{1-(\alpha\beta)^{i-1}}{1-\alpha\beta}. $$ Thus $R=\lim_{i\rightarrow\infty} R_{i,1}=\dfrac{A(1-\alpha)(1-\beta)}{1-\alpha\beta}$, which matches our prediction from Section \ref{derivations} (without the $f$).

\paragraph{Two different release periods, but identical replenishment.}
Here we have alternating periods of release and replenishment with two different release periods, starting with release, so $L=\{2,4\}$. We also assume that all of the replenishment periods have the same dynamics. Let $\beta:=\alpha_2=\alpha_4$. So for the recursion we have $$A_{i,1}= \beta^2\alpha_1\alpha_3A_{i-1,1}+A(1-\beta)(1+\alpha_3\beta),$$ which has solution $$A_{i,1}=A(\beta^2\alpha_1\alpha_3)^{i-1}+A(1-\beta)(1+\alpha_3\beta)\dfrac{1-(\beta^2\alpha_1\alpha_3)^{i-1}}{1-\beta^2\alpha_1\alpha_3}.$$ Taking the limit we get $$A_{\infty,1}=\displaystyle\lim_{i\rightarrow\infty} A_{i,1}= \dfrac{A(1-\beta)(1+\beta\alpha_3)}{1-\beta^2\alpha_1\alpha_3}.$$

\subsection{Supporting lemmas}\label{computations}
In this section, we give the proof of the recursion in Equation \ref{recursion}. We also give proofs of the supporting lemmas used to prove Proposition \ref{prop1} (Equation \ref{genrecursion}).

\begin{lemma}\label{recursionproof}
The solution to the recursion $A_i=bA_{i-1}+c$ is $A_i = A_1b^{i-1}+c\dfrac{1-b^{i-1}}{1-b}$.
\end{lemma}
\noindent In the following proof we use a standard generating function technique for solving recursions found in \cite{genfunc}.
\begin{proof}
To solve using generating functions we first multiply the recursion by $x^i$ and sum over $i\ge 2$ to get:

\begin{equation*}
\sum_{i\ge2}A_ix^i=b\sum_{i\ge2}A_{i-1}x^i+c\sum_{i\ge2}x^i.
\end{equation*}
Let $A(x) = \sum_{i\ge1}A_ix^i$. Then
\begin{equation*}
A(x)-A_1x=bxA(x)+c\left(\dfrac{1}{1-x}-1-x\right).
\end{equation*}
Solving for $A(x)$ yields:

\begin{eqnarray*}
A(x)&=&A_1\dfrac{x}{1-bx}+c\dfrac{x^2}{(1-x)(1-bx)}= A_1x\sum_{i\ge0}b^ix^i+cx^2\sum_{i\ge0}x^i\sum_{i\ge0}b^ix^i\\
&=& \sum_{i\ge1}\left(A_1b^{i-1}+c\sum_{k=0}^{i-2}b^k\right)x^i= \sum_{i\ge1}\left(A_1b^{i-1}+c\dfrac{1-b^{i-1}}{1-b}\right)x^i.
\end{eqnarray*}
Thus, $A_i =  A_1b^{i-1}+c\left(\dfrac{1-b^{i-1}}{1-b}\right)$.

\end{proof}

\noindent Recall $\alpha_{\ell} = e^{-\Delta t_{\ell}/\tau_\ell} $ and  $ \theta(\ell) = \left\{
     \begin{array}{ll}
       1 & \text{if } \ell\text{ indexes a replenishment period}\\
       0 & \text{if }\ell\text{ indexes a release period}
     \end{array}
   \right.$. 

\begin{lemma}\label{lemma1} $\displaystyle A_{i,\ell}=A_{i,1}\prod_{r=1}^{\ell-1} \alpha_r+A\sum_{r=1}^{\ell-1}\theta(r)(1-\alpha_r)\prod_{s=r+1}^{\ell-1}\alpha_s$.
\end{lemma}
\begin{proof}
We induct on $\ell$. Note that for $\ell=1$ we have \begin{equation*} A_{i,1}\prod_{r=1}^{0} \alpha_r+A\sum_{r=1}^{0}\theta(r)(1-\alpha_r)\prod_{s=r+1}^{0}\alpha_s=A_{i,1}.\end{equation*}

So the result holds for $\ell=1$. Let $\ell>1$ and assume that the result holds for all smaller $\ell$. Then
\begin{eqnarray*}
A_{i,\ell}&=&A_{i,\ell-1}+a_{i,\ell-1}(\Delta t_{\ell-1})  \text{  by Equation \ref{gensetup}}\\
&=& A_{i-1,\ell}+(\theta(\ell-1)A-A_{i,\ell-1})(1-\alpha_{\ell-1}) \text{  by Equation \ref{recursion2}}\\
&=& A_{i,\ell-1}\alpha_{\ell-1}+\theta(\ell-1)A(1-\alpha_{\ell-1})\\
&=&\left(A_{i,1}\prod_{r=1}^{\ell-2} \alpha_r+A\sum_{r=1}^{\ell-2}\theta(r)(1-\alpha_r)\prod_{s=r+1}^{\ell-2}\alpha_s\right)\alpha_{\ell-1}+\theta(\ell-1)A(1-\alpha_{\ell-1})\\
&& \text{by the induction hypothesis}\\
&=& A_{i,1}\prod_{r=1}^{\ell-1} \alpha_r+A\sum_{r=1}^{\ell-2}\theta(r)(1-\alpha_r)\prod_{s=r+1}^{\ell-1}\alpha_s+\theta(\ell-1)A(1-\alpha_{\ell-1})\\
&=& A_{i,1}\prod_{r=1}^{\ell-1} \alpha_r+A\sum_{r=1}^{\ell-1}\theta(r)(1-\alpha_r)\prod_{s=r+1}^{\ell-1}\alpha_s.
\end{eqnarray*}
\end{proof}

\begin{lemma}\label{lemma2}
$\displaystyle\sum_{\ell=1}^n(1-\alpha_\ell)\sum_{r=1}^{\ell-1}\theta(r)(1-\alpha_r)\prod_{s=r+1}^{\ell-1}\alpha_s=\sum_{\ell=1}^{n-1}\theta(\ell)(1-\alpha_\ell)\left(1-\prod_{r=\ell+1}^n\alpha_r\right)$
\end{lemma}
\begin{proof}
\begin{eqnarray*}
\displaystyle\sum_{\ell=1}^n(1-\alpha_\ell)\sum_{r=1}^{\ell-1}\theta(r)(1-\alpha_r)\prod_{s=r+1}^{\ell-1}\alpha_s
&=& \theta(1)(1-\alpha_1)\left(\sum_{r=2}^n(1-\alpha_r)\prod_{s=2}^{r-1}\alpha_s\right)\\
&&+\theta(2)(1-\alpha_2)\left(\sum_{r=3}^n(1-\alpha_r)\prod_{s=3}^{r-1}\alpha_s\right)\\
&&+\cdots+\theta(n-1)(1-\alpha_{n-1})(1-\alpha_n) \\
&=& \sum_{\ell=1}^n \theta(\ell)(1-\alpha_\ell)\sum_{r=\ell+1}^n(1-\alpha_r)\prod_{s=\ell+1}^{r-1}\alpha_s\\
&=& \sum_{\ell=1}^n \theta(\ell)(1-\alpha_\ell)\sum_{r=\ell+1}^n\left(\prod_{s=\ell+1}^{r-1}\alpha_s-\prod_{s=\ell+1}^{r}\alpha_s\right)\\
&=& \sum_{\ell=1}^n \theta(\ell)(1-\alpha_\ell)\left(1-\prod_{r=\ell+1}^n\alpha_r\right)
\end{eqnarray*}
\end{proof}
\begin{proposition}\label{prop1}
$\displaystyle A_{i,1}=A_{i-1,1}\left(\prod_{\ell=1}^{n} \alpha_\ell\right)+A\left(\sum_{\ell\in L}(1-\alpha_\ell)\prod_{r=\ell+1}^n\alpha_r\right)$
\end{proposition}
\begin{proof}
We know that $$A_{i,1}=A_{i-1,1}+\sum_{\ell=1}^n a_{i-1,\ell}(\Delta t_\ell)=A_{i-1,1}+\sum_{\ell=1}^n (\theta(\ell)A-A_{i-1,\ell})(1-\alpha_\ell).$$ Then we have \begin{eqnarray*}
A_{i,1}&=& A_{i-1,1}+\sum_{\ell=1}^n\theta(\ell)A(1-\alpha_\ell)\\
&&-\sum_{\ell=1}^n \left(A_{i-1,1}\prod_{r=1}^{\ell-1} \alpha_r+A\sum_{r=1}^{\ell-1}\theta(r)(1-\alpha_r)\prod_{s=r+1}^{\ell-1}\alpha_s\right)(1-\alpha_\ell) \text{ by Lemma \ref{lemma1}}\\
&=& A_{i-1,1}+\sum_{\ell=1}^n\theta(\ell)A(1-\alpha_\ell)-A_{i-1,1}\sum_{\ell=1}^n \prod_{r=1}^{\ell-1} \alpha_r(1-\alpha_\ell)\\&&+A\sum_{\ell=1}^n(1-\alpha_\ell)\sum_{r=1}^{\ell-1}\theta(r)(1-\alpha_r)\prod_{s=r+1}^{\ell-1}\alpha_s\\
&=& A_{i-1,1}\prod_{\ell=1}^n\alpha_\ell+\sum_{\ell=1}^n\theta(\ell)A(1-\alpha_\ell)-A\sum_{\ell=1}^{n-1}\theta(\ell)(1-\alpha_\ell)\left(1-\prod_{r=\ell+1}^n\alpha_r\right)\\ && \text{ by Lemma \ref{lemma2}}\\
&=& A_{i-1,1}\left(\prod_{\ell=1}^{n} \alpha_\ell\right)+A\left(\sum_{\ell\in L}(1-\alpha_{\ell})-\sum_{\ell\in L}(1-\alpha_\ell)\left(1-\prod_{r=\ell+1}^n\alpha_r\right)\right)\\
&=& A_{i-1,1}\left(\prod_{\ell=1}^{n} \alpha_\ell\right)+A\left(\sum_{\ell\in L}(1-\alpha_\ell)\prod_{r=\ell+1}^n\alpha_r\right)
\end{eqnarray*}
\end{proof}

	\chapter{Random walk model of vesicle replenishment}\label{randomwalkmodel}
Previous work indicates that vesicle replenishment is the rate-limiting step in sustained vesicle release \cite{cone3}, so in this chapter we take a closer look at the replenishment process. What part of the replenishment process limits release? In Section \ref{replenishment}, we design a simple random walk model to theoretically predict the time constant of replenishment, $\tau_a$, initially discussed in Section \ref{dynamics}. Using the model we can determine which parameters affect replenishment. We discover that $\tau_a$ relies on four fundamental parameters: vesicle diffusion, vesicle concentration, vesicle size, and the probability of attachment to the ribbon. The model predicts an exponential replenishment curve with time constant $$\tau_a = \dfrac{1}{D\rho\delta s}$$ where $D$ is the vesicle diffusion coefficient, $\rho$ is the vesicle concentration, $\delta$ is the diameter of a single vesicle, and $s$ is the attachment probability. The nature of vesicle movement within the synapse leads us to ask if the random diffusion of vesicles is rate-limiting.  We compare experimental data with our model results and conclude that diffusion is not, in fact, rate-limiting.

Further exploring replenishment in Section \ref{calcium}, we introduce two variations of the original model to investigate the role of Ca$^{2+}$ on replenishment.  The results in Sections \ref{replenishment} and \ref{calcium} are published in \cite{mine1}.

We can also use the model to calculate several other quantities of interest: how many vesicles hit the ribbon per second (hit rate) and how long it takes to fill up the ribbon (expected waiting time). The derivations in Section \ref{otherquantities} are unpublished.

		\section{Replenishment timescale}\label{replenishment}
			
			In this section we wish to answer the question: Is vesicle diffusion a rate-limiting step for replenishment? We first discuss how to measure the replenishment curve experimentally. Then since vesicles move randomly in the cell terminal without a directed movement toward the ribbon or active zone \cite{cones}, we create a random walk model of vesicle movement and replenishment. 

			\subsection{Paired pulse recordings}\label{pairedpulse}
		
To experimentally measure the replenishment of vesicles onto the synaptic ribbon, we use paired pulse recordings. First, when the ribbon is full, a large voltage jump, or pulse, is applied to the cell (similar to the pulse trains in Section \ref{pulsetrain}) and vesicle release is measured. The size of the voltage jump is such that all vesicles from the ribbon are released.  After $t$ seconds a second pulse of the same amplitude is applied and vesicle release is again measured (see Figure \ref{fig:pairedpulse}). This is repeated for multiple values of $t$ to approximate the replenishment curve.  For each inter-pulse interval, $t$, we plot the ratio $R_2/R_1$ where $R_1$ is the release on the first pulse and $R_2$ is the release on the second pulse (see Figure \ref{fig:pairedpulse}). The ratio, $R_2/R_1$, can be thought of as the percentage of vesicles replenished. In the case of salamader cones, the replenishment curve can be fit with a double exponential with time constants $\tau_{\text{fast}}$=815 ms (76\%)
 and $\tau_{\text{slow}}$=13 s \cite{mine1}.

\begin{figure}[h]
\begin{center}
                \includegraphics[width=160mm]{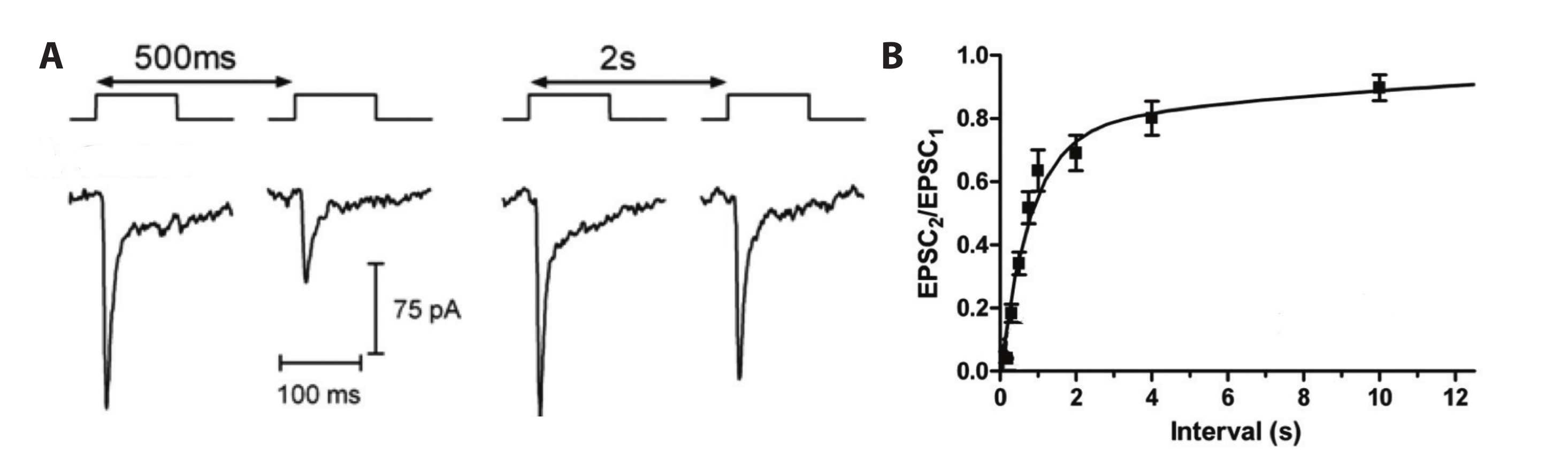}
      \caption{Panel A shows an example of two paired pulse recordings with interpulse intervals of 500ms and 2s. Note that the longer interpulse interval gives more time for the ribbon to replenish and hence the the second pulse is larger in the 2s trial.  Panel B shows the replenishment curve. The horizontal axis gives the interpulse interval and the vertical axis gives the ratio of the two responses. This ratio can be thought of as the percentage replenished. Adapted from \cite{mine1}.}
                \label{fig:pairedpulse}
\end{center}
\end{figure}

\noindent In the next section our goal is to predict the time constant of replenishment theoretically using a random walk model.

			\subsection{Derivation of the replenishment time constant}\label{randwalksetup}
  To answer this question, we developed a three-dimensional random walk model. We modeled the vesicle motion in the synapse by spherical vesicles undergoing random walks on a rectangular lattice of spacing $\delta$.  During each time step, $\Delta t$, every vesicle moves to an adjacent lattice site in each dimension. We update each of the three dimensions simultaneously, resulting in a diagonal move overall.  The (macroscopic) diffusion coefficient, $$D = \dfrac{\delta^2}{2\Delta t},$$ 
relates $\delta$ and $\Delta t$ in the (microscopic) random walk model, so these quantities cannot be chosen independently \cite{berg}.  Moreover, we would like to assume that each lattice site can be occupied by at most one vesicle, and that the occupation probabilities for distinct lattice sites are independent.  These two assumptions can only be satisfied if we choose $\delta$ to be equal to the vesicle diameter.  

\begin{figure}[h]
\begin{center}
                \includegraphics[width=100mm]{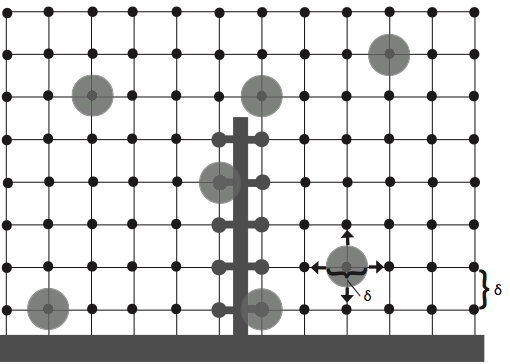}
                \caption{Random walk model of ribbon replenishment: the vesicles undergo a random walk on a rectangular lattice of spacing $\delta$ \cite{mine1}.}\label{fig:randomwalk}
\end{center}
\end{figure}

We use $p$ to denote the probability that a given lattice site (or tethering site) on the ribbon will become occupied in a given time step. If we assume the vesicles are distributed randomly and uniformly within the cell, the probability of a given lattice site being occupied is independent from one time step to the next.  For a lattice site far from the ribbon, this probability is simply given by the vesicle density per lattice site, $\rho \delta^3$, where $\rho$ is the overall density of vesicles inside the cell.  Since the ribbon sites can only be accessed from one side, for ribbon sites we have collision probability
$$p = \dfrac{1}{2}\rho\delta^3.$$
Let $s$ be the attachment probability, the probability that a vesicle that comes into contact with the ribbon will ``stick.''\footnote{i.e. become tethered to the ribbon until release, not drifting away at a future time step.} 
Then $sp$ is the probability of a vesicle actually sticking to a ribbon site in a given time step. Thus the probability of having to wait at least $t$ seconds before a ribbon site is ``permanently'' 
occupied is:
$$P(t) = (1-sp)^{t/\Delta t},$$
with $t/\Delta t$ giving the total number of time steps that have elapsed in $t$ seconds. Note that $1-sp$ is the probability
that a given lattice site on the ribbon is {\textit{not}} occupied permanently in a given time step.

Now we make a crucial approximation for $P(t)$, which is valid for $sp << 1.$\footnote{In fact, the approximation is still quite good up to values of $sp \sim 0.1$.}  The approximation stems from the fact that $\ln(1+x) \approx x$ for 
$|x| << 1$.  To use it, we first take the natural log of the $P(t)$ equation, and then plug in $\Delta t = \delta^2/2D$ and $p=\rho \delta^3/2$:
\begin{eqnarray*}
\ln P(t) &=& \dfrac{t}{\Delta t} \ln (1-sp) \approx \dfrac{t}{\Delta t} (-sp)\\
&=& -\dfrac{2Dps}{\delta^2}t =  -\dfrac{2D\rho \delta^3 s}{2\delta^2}t =-(D\rho \delta s)t.
\end{eqnarray*}
Exponentiating both sides we obtain
$$P(t) \approx e^{-t/\tau_a}, \quad \text{where } \;\;\; \tau_a = \dfrac{1}{D\rho\delta s}.$$

\noindent Solving for $P(t)$ without making the approximation we get $$P(t)=e^{-t/\tau_{\text{exact}}}, \quad \text{where } \;\;\; \tau_{\text{exact}}=\frac{-\delta^2}{2D\ln(1-\frac{1}{2}\rho\delta^3s)}.$$

Next, observe that the expected number of ribbon sites that are filled at time $t$, assuming all sites are empty at $t=0$, is given by
$$a(t) = \sum_{m=1}^n m {n \choose m} (1-P(t))^m P(t)^{n-m} = n(1-P(t)).$$
The second equality is obtained using a familiar variant of the Binomial Theorem. {Recall that 
$\sum_{m=0}^n {n \choose m} x^m y^{n-m} = (x+y)^n,$ by the Binomial Theorem.  Differentiating
with respect to $x$:
$$\sum_{m=1}^n m{n \choose m} x^{m-1} y^{n-m} = n(x+y)^{n-1}.$$
Now, letting $x = 1-P(t)$ and $y = P(t),$ we obtain
$$\sum_{m=1}^n m{n \choose m} (1-P(t))^{m-1} P(t)^{n-m} = n.$$
Finally, multiplying both sides by $1-P(t)$ we obtain the desired result.}
Note that each term in the sum corresponds to the probability that exactly $m$ sites are 
``permanently'' occupied at time $t$, weighted by $m$.  On the other hand, given that each of the $n$
ribbon sites has an occupation probability of $1-P(t)$ at time $t$, it is intuitive that the expected number of occupied sites at this time is $a(t) = n(1-P(t))$.  

Using the approximate expression for $P(t)$ we derived above, we obtain, $a(t)$, the expected number of vesicles on the ribbon at time $t$, in terms of our fundamental constants:
\begin{equation}a(t) = n(1-e^{-t/\tau_a}), \quad \text{where } \;\;\; \tau_a = \dfrac{1}{D\rho\delta s}.\end{equation}

			\subsection{Comparison of model predictions with data}
			
We wish to determine whether diffusion is rate-limiting for replenishment and thus sustained release.  We answer this question by comparing the experimentally measured time constant with the model predictions. Recall that in Section \ref{pairedpulse}, we fit an exponential replenishment curve with two time constants to the data with $\tau_{\text{fast}}$=815 ms (76\%)
 and $\tau_{\text{slow}}$=13 s.  

.  

Table \ref{fundconstants} shows the experimental values for all of the fundamental constants for salamander cones. Since we are interested in knowing whether diffusion is rate-limiting, we use our model to calculate the fastest possible timescale of vesicle replenishment due to vesicle diffusion. To do this we set the attachment probability $s$ equal to 1. Hence if all vesicles that collide with the ribbon due to diffusion attach to it with probability $s = 1$ then the predicted time constant is 
$$\tau_a = \dfrac{1}{(.11)(2210)(.045)} \; \text{seconds} = 91 \;\text{ms}.$$
Thus the model predicts that the fastest replenishment time constant for salamander cone ribbons is 91 ms, which is about an order of magnitude faster than the experimentally measured $\tau_{\text{fast}}$ of 815 ms.  This suggests that other factors beyond the rate of vesicle collisions with the ribbon, such as an attachment probability $s<1$, time of descent down the ribbon, and/or vesicle priming must play a role in slowing down the rate of vesicle accretion.

\begin{table}
\begin{center}
\begin{tabular}{|c|c|c|}
\hline
constant & meaning & measured value\\
\hline
$n$ & max no. of vesicles on the ribbon & 110 vesicles \cite{poolsize}\\
$D$ & vesicle diffusion coefficient & $0.11\; \mu m^2/s$ \cite{cones} \\
$\rho$ & (mobile) vesicle density & 2210 vesicles/$\mu m^3$ \cite{thoreson2004}\\
$\delta$ & vesicle diameter & 45 nm = 0.045 $\mu m$\cite{vesiclediam} \\
$s$ & attachment probability & $0 < s \leq 1$\\ \hline
\end{tabular}
\caption{Experimentally measured parameters for model of replenishment in salamader cones.}\label{fundconstants}
\end{center}
\end{table}

Since our theoretical model does not take into account ribbon geometry aside from assuming that the sites are only accessible from one side, it is reasonable to use this model to predict the replenishment time constant for other ribbon and conventional synapses provided their vesicles also exhibit random motion. The vesicles in rod bipolar cells, goldfish bipolar cells, and hippocampal cells all  appear to exhibit random motion \cite{diamond,holt,hippocampalvesicles}. Rod bipolar cells and goldfish bipolar cells both contain ribbons \cite{diamond,holt}, but hippocampal cells do not \cite{hippocampalvesicles}. Table \ref{othersynapses} gives the parameters for rod bipolar cells, goldfish bipolar cells, and hippocampal cells. In these cells, note that our model predicts a replenishment time constant that is slower than the measured replenishment time constant, indicating that the motion of vesicles may be rate-limiting for replenishment in these synapses.

\begin{table}[h]
\begin{center}
{\footnotesize{
\noindent\begin{tabular}{|c|c|c|c|}
\hline
    & Rod Bipolar Cells & Goldfish Bipolar Cells & Hippocampal Cells \\
    \hline
Diffusion coefficient, $D$ & $0.015\; \mu m^2/s$ \cite{holt} & $0.015\; \mu m^2/s$ \cite{holt} & $0.0042\; \mu m^2/s$ \cite{hippocampalvesicles} \\
\hline
Vesicle diameter, $\delta$  & 38 nm \cite{diamond}  & 30 nm \cite{goldfishdelta} & 38 nm \cite{hippocampaldelta}\\
\hline
Vesicle Concentration, $\rho$  & 1933 v/$\mu m^3$ \cite{diamond} & 445 v/$\mu m^3$ \cite{holt,cones}  & 270-465 v/$\mu m^3$ \cite{hippocampalrho, hippocampalrho2} \\
\hline
Measured $\tau_a$  & 400 ms \cite{RBC} & 4 s \cite{goldfishtau} & 7 s\cite{hippocampaltau, hippocampaltau2, hippocampaltau3} \\
\hline
Predicted $\tau_a$  & 908 ms &  5 s  & 13-23 s  \\
\hline 
\end{tabular}}}
\caption{Model predictions for the fastest possible timescale of replenishment for other synapses based on experimentally measured $D$, $\delta$, and $\rho$.}\label{othersynapses}
\end{center}
\end{table}

		\section{Role of calcium in replenishment}\label{calcium}
			
			In this section we use variations on our model to test two different mechanisms by which calcium (Ca$^{2+}$) and calmodulin (CaM), a calcium-binding messenger protein, might govern the attachment probability, $s$. It is known that Ca$^{2+}$ speeds replenishment \cite{cone2}. 
			 However, the mechanism by which this occurs is unknown. Data suggests that Ca$^{2+}$/CaM do not accelerate vesicles from the top of the ribbon to the release sites, nor do they increase the fusion rate at the membrane \cite{mine1}. Increased intracellular Ca$^{2+}$ does not affect the mobility of vesicles in the terminal \cite{cones,holt}. Hence $D$ and $\rho$ would not be affected by calcium. It appears that vesicle size (quantal amplitude) is also not affected by calcium \cite{mine1}, so we posit that Ca$^{2+}$/CaM increases the attachment probability, $s$.  
			 In this section, we will use two variations of the model to test two hypotheses regarding the role of  Ca$^{2+}$/CaM. 						
The first variant, which we call Model 1, assumes that Ca$^{2+}$/CaM acts as a switch making some vesicles more ``sticky'' than others.  In the second variant, Model 2, we assume Ca$^{2+}$/CaM again acts as a switch, but this time on the ribbon tethering sites, making some ribbon sites more ``sticky'' than others but leaving the vesicles unchanged.  Perhaps surprisingly, these two models produce qualitatively different results.  This may enable us to distinguish between the two possible functions of Ca$^{2+}$/CaM embedded into each model, by comparing the model predictions to experimental observations.

\begin{figure}[h]
\begin{center}
                \includegraphics[width=150mm]{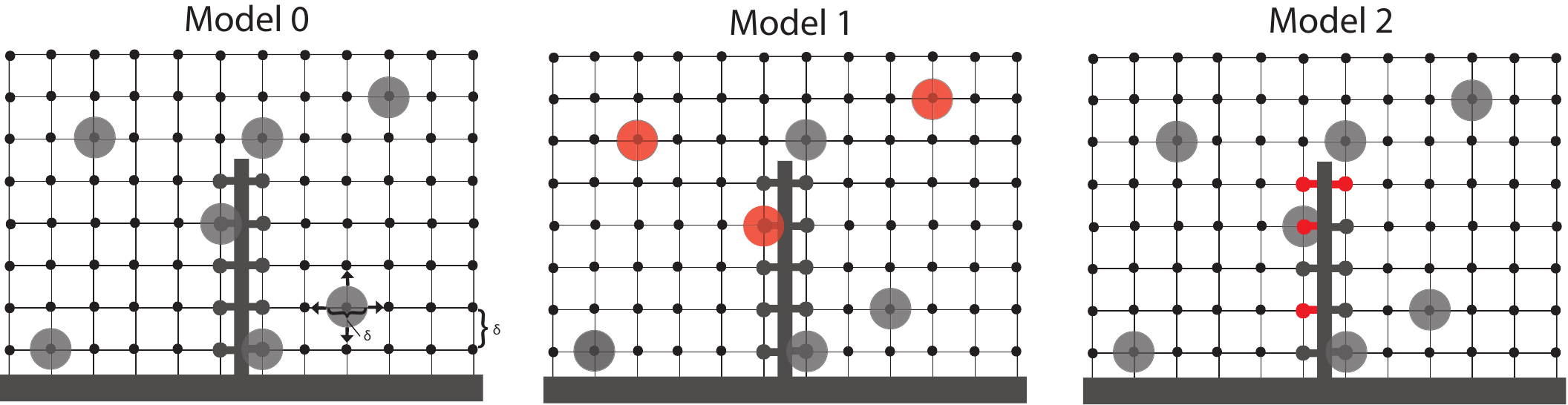}
                \caption{Calcium Hypotheses: Model 0 is the initial setup from Section \ref{randwalksetup}, Model 1 assumes changes in $s$ occur at the vesicles, and Model 2 assumes changes in $s$ occur at the ribbon sites \cite{mine1}. }
                \label{fig:threemodels-row}
\end{center}
\end{figure}

\subsection{Calcium affects vesicles}
For Model 1, suppose we have two populations of vesicles, A and B.  Vesicles in population A have higher attachment probability when they collide with the ribbon, given by attachment probability $s_A$.  Vesicles in population B have a lower attachment probability, given by attachment probability $s_B$.  Then, $0 < s_B \leq s_A \leq 1.$
Let $f$ be the fraction of vesicles in population A, with $1-f$ the fraction in population B.  Since the ribbon
sites are identical, the probability of a vesicle collision resulting in attachment is simply given by the weighted average of these attachment probabilities:
$$s = f s_A + (1-f) s_B.$$
The rest of the model remains unchanged. In particular, we still have 
$$a(t) = n(1-e^{-t/\tau_a}), \quad \text{where } \;\;\; \tau_a = \dfrac{1}{D\rho\delta s},$$
and $s$ is the ``average'' attachment probability computed above.  If the effect of Ca$^{2+}$/CaM is to change
the fraction $f$ of vesicles in the stickier population, then this effect will manifest itself as a change in 
the vesicle accretion timescale, $\tau_a$.  Inhibition of Ca$^{2+}$/CaM should cause a decrease in $f$, and hence an increase in $\tau_a$.  Note that this model does {\textit{not}} predict the existence of a second timescale, even
though there are two populations of vesicles.

\subsection{Calcium affects the ribbon}
For Model 2, suppose all vesicles are identical, but we have two populations, A and B, of tethering sites on the ribbon.
The ribbon sites in population A are more sticky, modeled by a higher attachment probability $s_A$, while the ribbon sites in population B are less sticky, with $s_B < s_A$.  We let $n_A$ and $n_B$ denote the number of sites in each population, with $n = n_A +n_B$.  If $f$ is the fraction of ribbon sites in population A, then
$n_A = fn$ and $n_B = (1-f)n.$

Since attachment probabilities are different for different ribbon sites, we must use different expressions for $P(t)$:  $P_A(t) = (1-s_Ap)^{t/\Delta t}$ for the sites in population A, while $P_B(t) = (1-s_Bp)^{t/\Delta t}$ for population B.  The result is that expected number of vesicles on the ribbon at time $t$ is given by the sum of two terms:
$$a(t) = n_A(1-e^{-t/\tau_A}) +  n_B(1-e^{-t/\tau_B}) , \quad \text{where} \;\; 
\tau_A = \dfrac{1}{D\rho\delta s_A}, \; \text{and} \;\; \tau_B = \dfrac{1}{D\rho\delta s_B}.$$
Note that since $s_A > s_B$, the population A timescale is faster, $\tau_A < \tau_B$.  The presence of 
two timescales makes this model qualitatively different from Model 1.  This difference is also seen in thinking about the effect of Ca$^{2+}$/CaM in this model.  If Ca$^{2+}$/CaM changes the fraction of ribbon sites $f$ that belong to the stickier population, then this will manifest itself as a change in the amplitudes $n_A$ and $n_B$ for each component of $a(t)$.  Inhibition of Ca$^{2+}$/CaM should cause a decrease in $f$, and hence a decrease in $n_A$ and an increase in $n_B$.   This model predicts no Ca$^{2+}$/CaM effect on the time constants, in contrast to Model 1.

			\subsection{Comparison to experimental results}

		Recall that paired pulse experiments showed that the replenishment curve can be fit with a double exponential. 
Table \ref{calciumtable} shows the results of experimentally decreasing Ca$^{2+}$/CaM using BAPTA, nifedipine, Calmidazolium, and MLCK \cite{mine1}.  When the fast timescale is constrained to match the control, we see that the fit is comparable to the unconstrained case, but the percentage of fast-replenishing sites is much lower. Thus, inhibiting Ca$^{2+}$/CaM in these experiments caused slight changes in the fast timescale, but more substantial changes to the amplitude of the fast component.  This is consistent with the predictions seen in Model 2, where we have two time constants and inhibition of Ca$^{2+}$/CaM causes the amplitude of the fast component to decrease.   We conclude that Ca$^{2+}$/CaM more likely acts on ribbon sites rather than vesicles.  It is possible that vesicles are affected by Ca$^{2+}$/CaM as well, but changes in the fast timescale were not consistent across trials.

\begin{sidewaystable}
\begin{center}
{{
\begin{tabular}{|c|c|c|c|c|c|c|}
\hline
Parameter & 5 mM EGTA & 1 mM BAPTA & 3 $\mu$M nifedipine & 20 $\mu$M Calm. & 20 $\mu$M MLCK-& 20 $\mu$M MLCK\\
&  (control) & &  &  & control  & \\\hline
\textbf{Unconstrained} & & & & & & \\
$\tau_{\text{fast}}$ (s) & 0.816 & 2.1 & 1.39 & 1.98 & 0.951 & 0.917\\
\% fast & 75.7 & 56.2 & 37.1 & 33.1 & 68.1 & 23.9\\
$\tau_{\text{slow}}$ (s) & 12.9 & 28.8 & 17.6 & 25.9 & 17.6 & 13.8\\
$R^2$ & 0.97 & 0.99 & 0.99 & 0.99 & 0.98 & 1.0\\
 \textbf{Constrained} & & & & & & \\ 
 $\tau_{\text{fast}}$ (s) &  & 0.816$^a$ & 0.816$^a$  & 0.816$^a$ &  & 0.951$^b$\\
\% fast &  & 22.1 & 26.9 & 17.5 &  & 25.3\\
$\tau_{\text{slow}}$ (s) &  & 9.56 & 12.2 & 17.4 &  & 14.1\\
$R^2$ &  & 0.97 & 0.99 & 0.99 &  & 1.0\\
$n$ & 10 & 7 & 5 & 4 & 6 & 17\\

\hline

\end{tabular}}}
\end{center}
\caption{Effects of Ca$^{2+}$/CaM: The replenishment time constant is fit with a double exponential function. Note that the fits are similar in the unconstrained and constrained cases. Aside from the two control cases, each other experimental condition decreases calcium in some way. BAPTA restricts the spread of intracellular Ca$^{2+}$, nifedipine reduces Ca$^{2+}$ influx, and Calmidazolium (Calm.) and MLCK are calmodulin inhibitors. \hspace{\textwidth}$^a$Fits constrained to $\tau_{\text{fast}}$ for 5 mM EGTA control condition (0.816 s).
 $^b$Fits constrained to $\tau_{\text{fast}}$ for MLCK-control condition (0.951 s). Adapted from \cite{mine1}.
}\label{calciumtable}

\end{sidewaystable}
			
	\section{Other quantities of interest}\label{otherquantities}	
	
	\subsection{Hit rate}\label{hitrate}
		In this section we compute the hit rate, i.e. the number of vesicles coming in contact with the ribbon per second.  One way to compute the hit rate is to do so macroscopically by first calculating the flux near the ribbon and then multiplying by the surface area of the ribbon.  We consider the concentration of vesicles to be zero on the ribbon.  The concentration of mobile vesicles not attached to the ribbon is $\rho$ as before. Thus over the distance $\delta$ (one lattice step from a ribbon site to a nonribbon site) we have a change in concentration from $\rho$ to $0$ giving us a flux of $J=D\dfrac{\rho-0}{\delta}=\dfrac{D\rho}{\delta}\dfrac{\text{vesicles}}{\mu m^2\cdot s}$ \cite{berg}. The total surface area of the ribbon is $n\delta^2~ \mu m^2$.  Thus computing the hit rate using this method yields $H=\dfrac{D\rho}{\delta}n\delta^2=D\rho\delta n $ = $\dfrac{n}{\tau_a}$ vesicles/s where $\tau_a=\dfrac{1}{D\rho\delta}$ is the replenishment time constant from our previous calculations.

We can also microscopically compute the hit rate. First, find the expected number of sites filled in a single time step.  We know that the probability of a ribbon site becoming occupied in the next time step is $\dfrac{1}{2}\rho\delta^3$ and there are $n$ sites on the ribbon. Thus the expected number of sites filled in a single time step is given by $\dfrac{1}{2}\rho\delta^3n$. We know that each time step $\Delta t$ is $\dfrac{\delta^2}{2D}$ seconds. Thus the formula for hit rate in our model is $$H=\dfrac{\dfrac{1}{2}\rho\delta^3n~\text{vesicles}}{1~\text{time step}}\cdot\dfrac{1~ \text{time step}}{\dfrac{\delta^2}{2D}~\text{seconds}}=D\rho\delta n ~\text{vesicles/s},$$ in agreement with the flux-based calculations above.

Lastly, we can again confirm this result by approximating the hit rate from $\dfrac{da}{dt}$, which is the formula for the rate of accumulation of vesicles onto the ribbon from previous calculations. The hit rate computed above ignores the decrease in available surface area due to vesicles already on the ribbon and thus corresponds to the accumulation rate only for small values of $t$.  Recall that $a(t)=n(1-e^{-t/\tau_a})$ is the  expected number of vesicles on the ribbon at time $t$ where $\tau_a=\dfrac{1}{D\rho\delta s}$. Here we assume $s=1$ since we are just finding how many vesicles come in contact with the ribbon per second and ignoring how many stick. For $t=0$ the rate of accumulation also corresponds to the hit rate since the form of $a(t)$ assumes that the ribbon is empty at $t=0$. Then $\dfrac{da}{dt}=\dfrac{n}{\tau_a}e^{-t/\tau_a}$, so at $t=0$ we have that the rate of accumulation is $\dfrac{n}{\tau_a}=D\rho\delta n$ vesicles/s. So again the hit rate is given by $H=D\rho\delta n$ vesicles/s.

In summary, when the ribbon is empty, the hit rate is $$H=\dfrac{n}{\tau_a}=D\rho\delta n ~\text{vesicles/s},$$ but as the ribbon becomes filled the hit rate decreases as $$H(t)=\dfrac{n}{\tau_a}e^{-t/\tau_a}.$$ Note that $H=H(0)$ and $$\displaystyle\int_0^\infty H(t) dt=\dfrac{n}{\tau_a}\int_0^\infty e^{-t/\tau_a} dt=\dfrac{n}{\tau_a}(-\tau_a e^{-t/\tau_a})\bigg|_0^\infty=\dfrac{n}{\tau_a}\tau_a=n.$$

		\subsection{Expected waiting time}\label{expectedwaitingtime}
		To calculate the expected waiting time, $T_{\text{wait}}$, to fill all $n$ lattice sites on the ribbon, we first consider an individual lattice site.  Let $P(t)$ be the probability that we wait at least $t$ seconds to fill the given lattice site. Then $1-P(t)$ is the probability that the given lattice site fills before $t$ seconds have passed and $r(t)=(1-P(t))^n$ is the probability that all $n$ sites have filled before $t$ seconds have passed. Hence the probability we wait {\textit{exactly}} $t$ seconds is $r'(t)dt$.

\noindent Thus the expected waiting time $T_{\text{wait}}$ is given by
$$T_{\text{wait}}=E[t]=\int_0^\infty tr'(t)dt=\tau_aH_n$$ 

\noindent where $\tau_a$ is the vesicle accretion timescale and $H_n=\displaystyle\sum_{k=1}^{n} \dfrac{1}{k}$ is the $n$th harmonic number. This result is proven in Lemma \ref{ewt}.

\begin{lemma}\label{ewt} $\displaystyle{\int_0^\infty tr'(t)dt= \tau_aH_n}$ where $r(t)=(1-e^{-t/\tau_a})^n$, and $\displaystyle{H_n=\sum_{k=1}^n \frac{1}{k}}$ is the $n$th Harmonic number.
\end{lemma}

\begin{proof} We have $r'(t)dt=\frac{n}{\tau_a}e^{-t/\tau_a}(1-e^{-t/\tau_a})^{n-1}dt$, so 
\begin{eqnarray*}
\int_0^\infty tr'(t)dt &=&\frac{n}{\tau_a}\int_0^\infty te^{-t/\tau_a}(1-e^{-t/\tau_a})^{n-1}dt\\
&=& \frac{n}{\tau_a}\int_0^\infty te^{-t/\tau_a}\sum_{k=0}^{n-1}\binom{n-1}{k}(-1)^ke^{-kt/\tau_a} dt \\
&=& \frac{n}{\tau_a}\sum_{k=0}^{n-1}\binom{n-1}{k}(-1)^k\int_0^\infty te^{-(k+1)t/\tau_a} dt\\
&=& \frac{n}{\tau_a}\sum_{k=0}^{n-1}\binom{n-1}{k}(-1)^k\dfrac{\tau_a^2}{(k+1)^2} \\
&=& \tau_a \sum_{k=0}^{n-1}n\binom{n-1}{k}(-1)^k \frac{1}{(k+1)^2}.
\end{eqnarray*}

\noindent Using the identity $\displaystyle(k+1)\binom{n}{k+1}=n\binom{n-1}{k}$ and reindexing, we can rewrite $$\tau_a \sum_{k=0}^{n-1}n\binom{n-1}{k}(-1)^k \frac{1}{(k+1)^2}=\tau_a \sum_{k=1}^{n}\binom{n}{k}(-1)^{k-1} \frac{1}{k}.$$

\noindent Now, we claim that $\displaystyle{\sum_{k=1}^{n}\binom{n}{k}(-1)^{k-1} \frac{1}{k}=H_n}$.  First note that $$\int_0^1{\frac{1-x^n}{1-x} dx}=\int_0^1{\sum_{k=0}^{n-1}x^kdx}=\sum_{k=0}^{n-1}\int_0^1{x^k}=\sum_{k=1}^{n}\frac{x^{k-1}}{k}\bigg|_0^1=\sum_{k=1}^{n}\frac{1}{k}= H_n.$$

\noindent Then, letting $u=1-x$, we have,
\begin{eqnarray*}
\int_0^1{\frac{1-x^n}{1-x} dx}&=&-\int_1^0{\frac{1-(1-u)^n}{u} du}\\&=& \int_0^1{\frac{1-\sum_{k=0}^n{\binom{n}{k}(-1)^ku^k}}{u} du}\\
&=& \int_0^1{\frac{1-{\binom{n}{0}(-1)^0u^0}-\sum_{k=1}^n{\binom{n}{k}(-1)^ku^k}}{u} du}\\ &=& -\sum_{k=1}^n\binom{n}{k}(-1)^k\int_0^1{u^{k-1}du}\\
&=&  \sum_{k=1}^n\binom{n}{k}(-1)^{k-1}\left[\frac{u^k}{k}\bigg|_0^1\right]= \sum_{k=1}^n\binom{n}{k}(-1)^{k-1}\frac{1}{k}.\\
\end{eqnarray*}

\noindent Thus, $\displaystyle{\sum_{k=1}^{n}\binom{n}{k}(-1)^{k-1} \frac{1}{k}=H_n}$ and therefore $\displaystyle{\int_0^\infty tr'(t)dt=\tau_aH_n}$. 
\end{proof}

\noindent Recall that we predicted $\tau_a=91$ ms and $H_{110}$ is approximately 5.2882, so the expected waiting time $T_{\text{wait}}$ is 481 ms. Note that this calculation is useful to set the duration for computer simulations in the computational model of replenishment (see Chapter \ref{compmodel}).

\paragraph{Expected waiting time in Model 1} 
Recall that in Model 1, we have two populations of vesicles where $s_A$ and $s_B$ are the attachment probabilities for population $A$ and population $B$, respectively. Then the expected waiting time for Model 1 is $T_{\text{wait}}=\tau_aH_n=\dfrac{H_n}{D\rho\delta s}$ where $s = f s_A + (1-f) s_B$ and $f$ is the fraction of vesicles in population $A$. 

\paragraph{Expected waiting time in Model 2}
Recall that in Model 2, the stickiness occurs in the ribbon sites instead. We have $n_A$ ribbon sites with attachment probability $s_A$ and $n_B$ sites with attachment probability $s_B$. We know that the probability of having to wait at least $t$ seconds before a ribbon site in population $A$ is occupied is given by $P_A(t)=e^{-t/\tau_A}$ where $\tau_A=\frac{1}{D\rho\delta s_A}$ and the probability of having to wait at least $t$ seconds before a ribbon site in population $B$ is occupied is given by $P_B(t)=e^{-t/\tau_B}$ where $\tau_B=\frac{1}{D\rho\delta s_B}$. Then $r(t)=(1-P_A(t))^{n_A}(1-P_B(t))^{n_B}=(1-e^{-t/\tau_A})^{n_A}(1-e^{-t/\tau_B})^{n_B}$ is the probability that all $n=n_A+n_B$ sites have filled before $t$ seconds have passed. Hence the probability we wait {\textit{exactly}} $t$ seconds is $r'(t)dt$. 

Now, we have that the expected waiting time for the ribbon to fill is 

\begin{eqnarray*}
T_{\text{wait}}&=&\int_0^\infty tr'(t)dt\\
&=&\int_0^\infty \biggl(t(1-e^{-t/\tau_A})^{n_A}(1-e^{-t/\tau_B})^{n_B-1}e^{-t/\tau_B}\frac{n_B}{\tau_B}\\
&&+ t(1-e^{-t/\tau_B})^{n_B}(1-e^{-t/\tau_A})^{n_A-1}e^{-t/\tau_A}\frac{n_A}{\tau_A}\biggr)dt.\\
&=&\frac{n_B}{s_B}\sum_{i=0}^{n_A}\sum_{j=0}^{n_B-1}\binom{n_A}{i} \binom{n_B-1}{j}(-1)^{i+j}\frac{1}{(j+\frac{\tau_B}{\tau_A}i+1)^2}\\
&&+\frac{n_A}{s_A}\sum_{k=0}^{n_B}\sum_{l=0}^{n_A-1} \binom{n_B}{k}\binom{n_A-1}{l}(-1)^{k+l}\frac{1}{(l+\frac{\tau_A}{\tau_B}k+1)^2}
\end{eqnarray*}

	\chapter{Computational model}\label{compmodel}
	The random walk model of vesicle replenishment described in Chapter \ref{randomwalkmodel} does not take into account the geometry of the ribbon. What effect does ribbon geometry have on replenishment? In this chapter we discuss a computational model of replenishment including ribbon geometry that was designed to complement the theoretical model. This model is currently unpublished.  
In Section \ref{overview} we describe the setup of the model. Then in Section \ref{geom}, we compare the results of the computational model with the results of the theoretical model to determine the role geometry plays in replenishment. The Matlab code for the computational model of replenishment can be found in Appendix \ref{compmodelcode}.
		
		\section{Description of the computational model}
		
		\label{overview}
The cell space is modeled by a 3-dimensional array with entries in $\{0,1\}$ where 1s indicate locations of vesicles within the cell. The array is randomly generated with a given concentration of 1s computed from the vesicle concentration $\rho$. The total number of 1s denoted by $N$. The matrix $S$ is a $N\times3$ matrix where the $i$th row records the current position of the $i$th vesicle within the array. In each time step the matrix is updated by adding a random $N\times3$ matrix with entries in $\{-1,1\}$. This ensures that each vesicle moves one lattice space per dimension in each time step, a diagonal move overall. A set of $n$ coordinates, where $n$ is the maximum number of vesicles that fit on the ribbon, are designated as ``ribbon sites.'' The coordinates of these sites are recorded in the matrix SiteMat and during each time step the coordinates of all $N$ vesicles are checked against SiteMat to determine how many vesicles are occupying ribbon sites.  Then with probability $s$, the attachment probability discussed in Section \ref{randwalksetup}, a vesicle occupying a ribbon site becomes permanently stuck and does not update in subsequent time steps.  This is done by zeroing out the corresponding row in the update matrix. At each time step we record how many vesicles are permanently stuck to the ribbon. This gives us the computational replenishment curve, which we can compare with our theoretical prediction. 

Figure \ref{fig:ribbonshapes} shows the arrangement of ribbon sites in the rectangular ribbon case and the ``nonribbon'' case.  The placement of sites in the nonribbon case allows us to study the effects of ribbon geometry. For cases with the rectangular ribbon, we also make sure that the ribbon is solid by not allowing any updates that would represent a vesicle passing through the ribbon. This is achieved by returning any vesicles that pass through the ribbon in the current time step to their original position before moving to the next time step. 

		\begin{figure}[h]
\centering
\includegraphics[width=160mm]{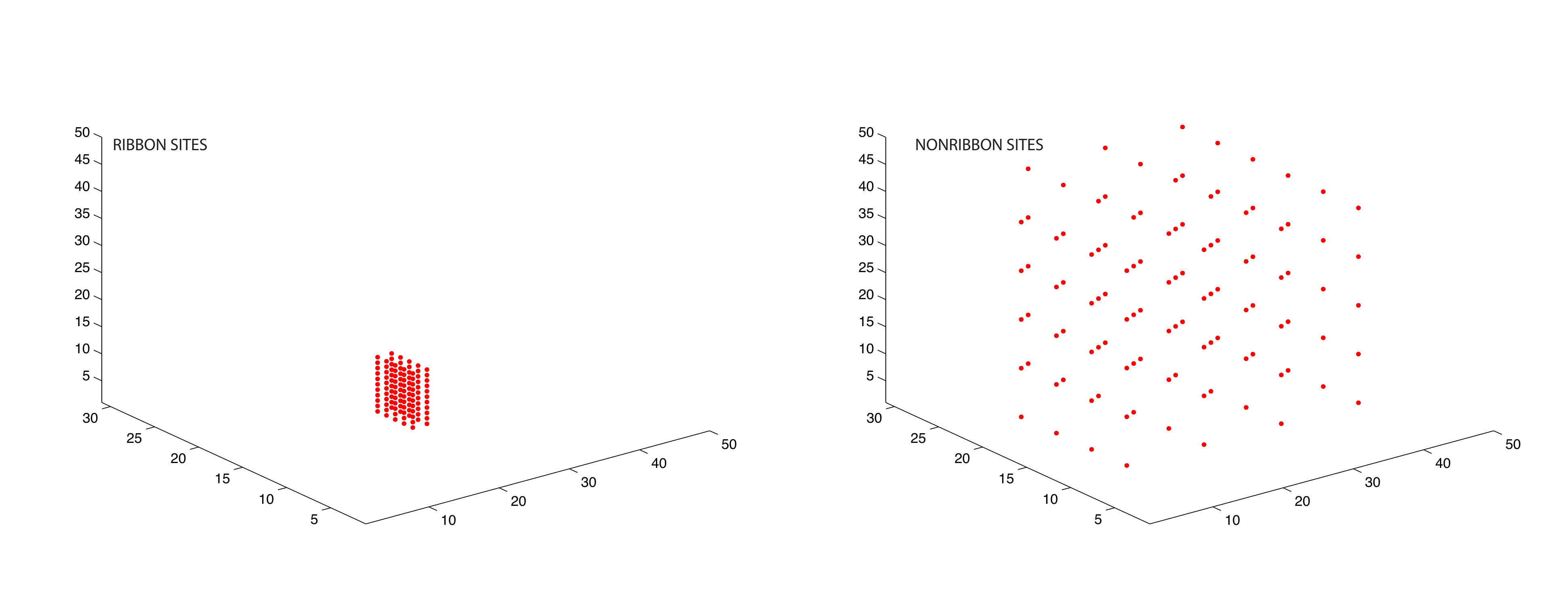}
                \caption{Ribbon and nonribbon attachment sites:  The left panel shows the ribbon sites arranged in a flat rectangular plate based off the structure seen in cone photoreceptors. The right panel shows the sites spread out in the cell space to act as a control when studying the effects of ribbon geometry.}
                \label{fig:ribbonshapes}
\end{figure}

Note that because of the way the vesicles update it is possible for vesicles to occupy the same lattice site during the same time step. For small concentrations (around 300 v/$\mu$m$^3$), less than 1\% of the vesicles are occupying the same site as another vesicle and for larger concentrations (around 2300 v/$\mu$m$^3$), less than 10\% of the vesicles are occupying the same site as another vesicle.

Since the theoretical model does not take into account geometry, the computational model and the theoretical model should be close in the nonribbon case\footnote{Recall that in the analytical model we have a factor of 1/2 that represents the fact that the ribbon sites are only accessible from one side. When using the analytical model to predict replenishment in nonribbon cases, we leave out the factor of 1/2 since these sites are accessible from all sides.}. Figure \ref{fig:sanitycheck} shows the comparison between the two models in the nonribbon case for several different concentrations and attachment probabilities. Note that the models closely match across a wide range of parameters.  

		\begin{figure}[h]

\includegraphics[width=160mm]{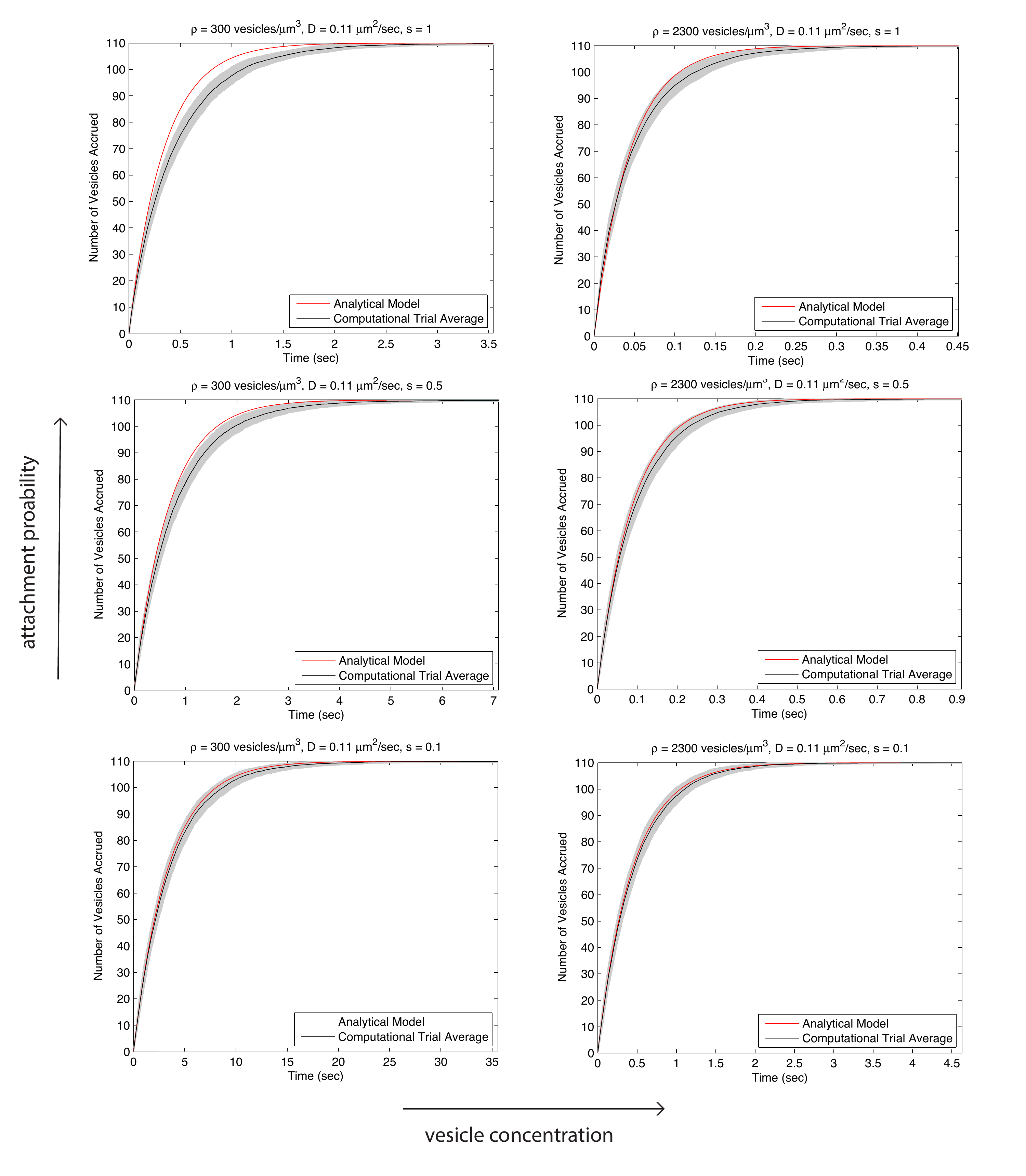}
                \caption{Comparison to the theoretical model: The computational model results for the nonribbon case (see Figure \ref{fig:sanitycheck}) are averaged over 100 trials and the gray area represents one standard deviation from the mean. We show trials for a low (300 vesicles/ $\mu m^3$) and a high concentration (2300 vesicles/ $\mu m^3$) as well as three different attachment probabilities (0.1, 0.5, and 1). Note that the theoretical and computational models appear to closely match across a variety of parameters, as expected in the nonribbon case.}
                \label{fig:sanitycheck}
\end{figure}

		\section{Effect of ribbon geometry on replenishment}\label{geom}
		Since the theoretical model does not take into account the geometry of the ribbon we use our computational model to approximate the replenishment curve in the case where we have a rectangular ribbon attached to the edge of the cell space. 

		\begin{figure}[h]

\includegraphics[width=150mm]{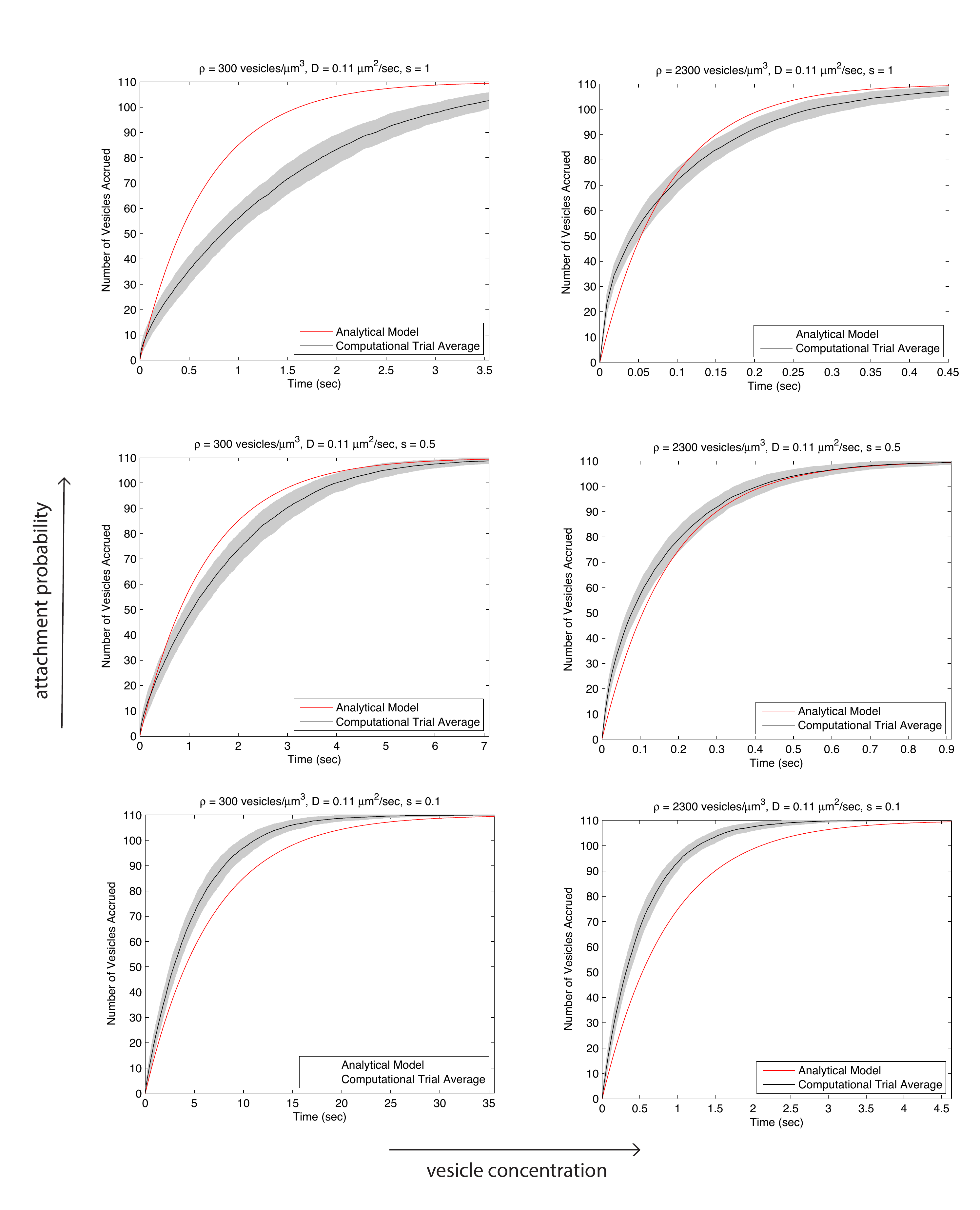}
                \caption{Effect of ribbon geometry on replenishment: The computational model results for the rectangular ribbon case (see Figure \ref{fig:sanitycheck}) are averaged over 100 trials and the gray area represents one standard deviation from the mean. We show trials for a low (300 vesicles/ $\mu m^3$) and a high concentration (2300 vesicles/ $\mu m^3$) as well as three different attachment probabilities (0.1, 0.5, and 1). Note that the computational model shows the greatest deviation from the theoretical prediction in the low concentration/high attachment probability case.  }
                \label{fig:ribbongeometry}
\end{figure}

Figure \ref{fig:ribbongeometry} indicates that the geometry of the ribbon does in fact play a role in replenishment. The ribbon sites in this case have a rectangular shape based on the ribbons in cone photoreceptors and the trials are run for varying vesicle concentrations and attachment probabilities.  For low attachment probability and high vesicle concentration, the computational model trial average shows faster replenishment than predicted by the analytical model. For high attachment probability and low vesicle concentration, the computational model trial average shows slower replenishment than predicted by the analytical model. 

Recall that cone photoreceptor synapses have a high vesicle concentration and based on our random walk model of replenishment in Chapter \ref{randomwalkmodel} are also likely have a low attachment probability. Studying the case of high vesicle concentration and low attachment probability in our computational model, we note that the computational trial average is faster than the theoretical model prediction. This indicates that having a synaptic ribbon for this parameter regime actually speeds replenishment compared to having no ribbon where the vesicles dock directly with the cell membrane. This may provide evidence for why photoreceptor cones contain ribbons, but exactly how the ribbon accelerates replenishment in this case is still unclear.

In the case of high attachment probability and low concentration, we hypothesize that once the ribbon starts to fill up, the local concentration near the ribbon decreases causing the ribbon to fill slower than predicted.
To test this we calculate the concentration of vesicles close to the ribbon and far away from the ribbon. Figure \ref{fig:localconc} shows the results of this calculation. The concentration near the ribbon drops steeply as the ribbon fills up and the concentration further away stays relatively constant. This drop in local concentration is most pronounced in the high $s$/low $\rho$ cases. This may account for the slower replenishment we see in these cases. 

		\begin{figure}[h]
\centering
\includegraphics[width=220mm]{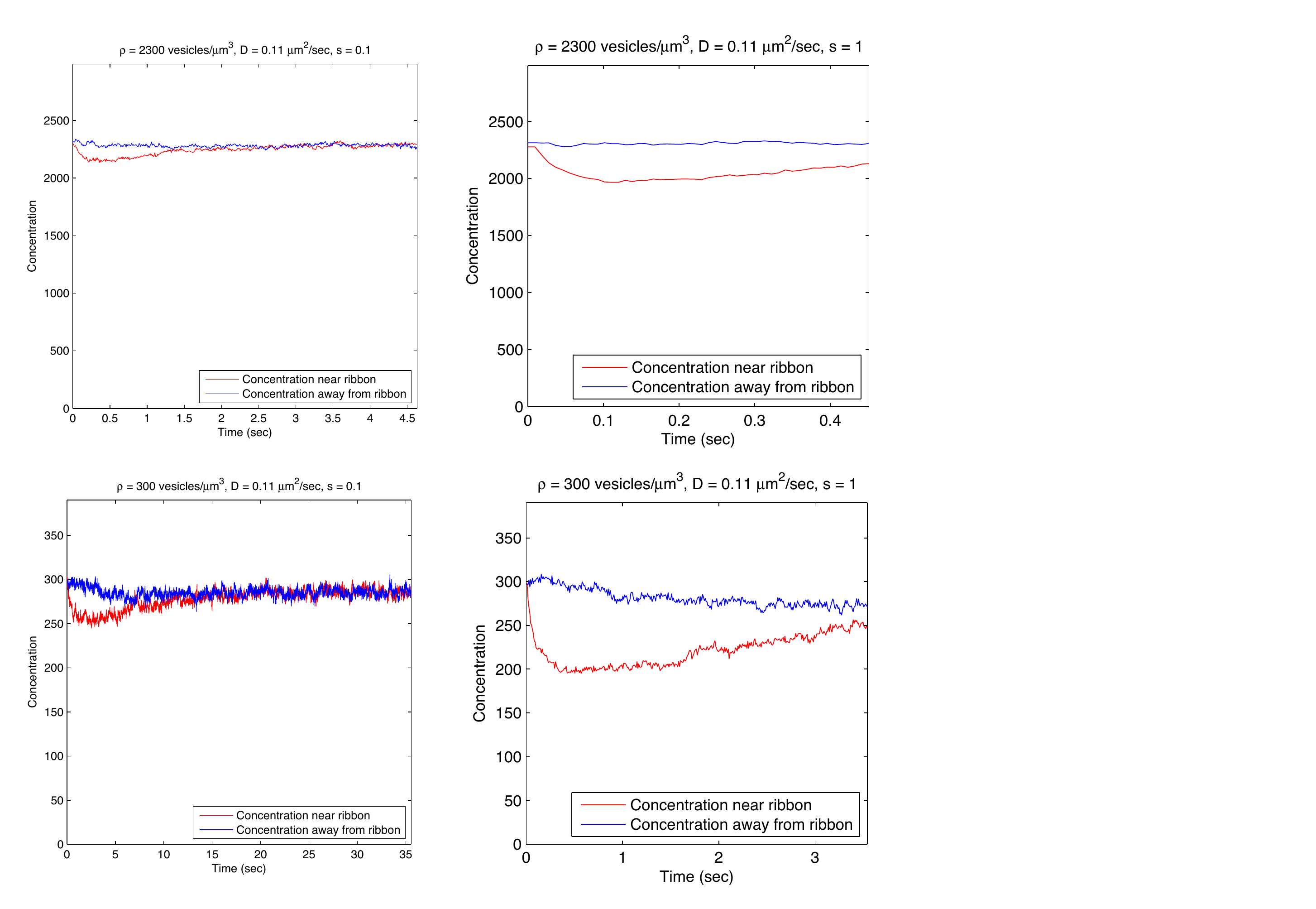}
                \caption{Local concentration: The above plots show the difference between the local concentration near the ribbon (in red) and away from the ribbon (in blue) for two different overall vesicle concentrations (300 vesicles/$\mu m^3$ and 2300 vesicles/$\mu m^3$) and two different attachment probabilities ($s=0.1$ and 1) in the case where the ribbon sites are arranged in a rectangular plate. Notice the sharp drop in the local concentration near the ribbon when the ribbon first begins to fill. The percentage drop is largest for the low concentration/high attachment probability case.   }
                
                           \label{fig:localconc}
\end{figure}

Since the theoretical model incorporates the factor of 1/2 indicating that the sites are only accessible from one side, but not specific ribbon geometry, the theoretical prediction gives a reasonable approximation for time constant of vesicles reaching the cell membrane in a terminal without a ribbon.  Figure \ref{fig:ribbongeometry} indicates that having a ribbon actually may slow replenishment in synapses with low vesicles concentration.  This suggests that having a ribbon would not be advantageous in synapses with low vesicle concentration and random motion of vesicles. This is consistent with the case of hippocampal synapses which have a low vesicle concentration and random motion, but do not contain ribbons \cite{hippocampalvesicles}.

\section{Future work}

\paragraph{Local concentration.}  The computational model discussed in Chapter \ref{compmodel} revealed that the local concentration near the ribbon drops sharply near the ribbon as the ribbon fills up.  This contradicts our assumption that the vesicle concentration is constant. To improve our random walk model, we would like to find a formula to describe the change in concentration as the ribbon fills up.

\paragraph{Movement on ribbon and vesicle fusion.} The random walk model does not take into account the movement of vesicles along the ribbon.  As more becomes known about this process, we would like to incorporate this step into the model. This model also does not take into account vesicle release.  Adding these features will allow us to explore more questions regarding the function of the ribbon.

	\part{Neural Sequences in Threshold-Linear Networks}\label{Seq_part}
\chapter{Introduction to Part \ref{Seq_part}}
	
Part \ref{Seq_part} focuses on neural networks and the interplay between network connectivity and neural activity.  In particular, we are interested in studying how network structure shapes the behavior of the network. 

To do this, we study the dynamics of a combinatorial family of competitive threshold-linear networks constructed from simple directed graphs (the CTLN model) as defined in \cite{CTLN}.  This family of networks is particularly well suited for our study because the network construction guarantees that differences in dynamics arise solely from differences in the connectivity of the underlying graph.  This allows us to focus on the properties of the graphs themselves when trying to predict the behavior of the corresponding network. This robust family of dynamical systems exhibits several different nonlinear behaviors including limit cycles, quasiperiodic attractors, and chaos.  Figure \ref{fig:multistability} shows an example of a network that exhibits multiple behaviors depending on the choice of initial conditions. 

In this part, we begin by giving some background about competitive threshold-linear graphs and the CTLN model.  We then use the CTLN model to study how the graph structure affects the resulting dynamics.  Computational experiments show that most CTLN networks yield limit cycles. We present an algorithm that uses the structure of the underlying graph to predict the sequence of firing of neurons in the limit cycle.  Our algorithm predicts the sequence correctly for most small graphs, but sometimes fails for certain classes of larger graphs. To gain further insight into how the structure of the underlying graph shapes the dynamics, we classify the behavior we see for small networks ($n\le 5$ nodes) arising from oriented graphs. 
Both of these results work towards the larger goal of better understanding high-dimensional nonlinear dynamics.

\begin{figure}[h]
\centering
\includegraphics[width=155mm]{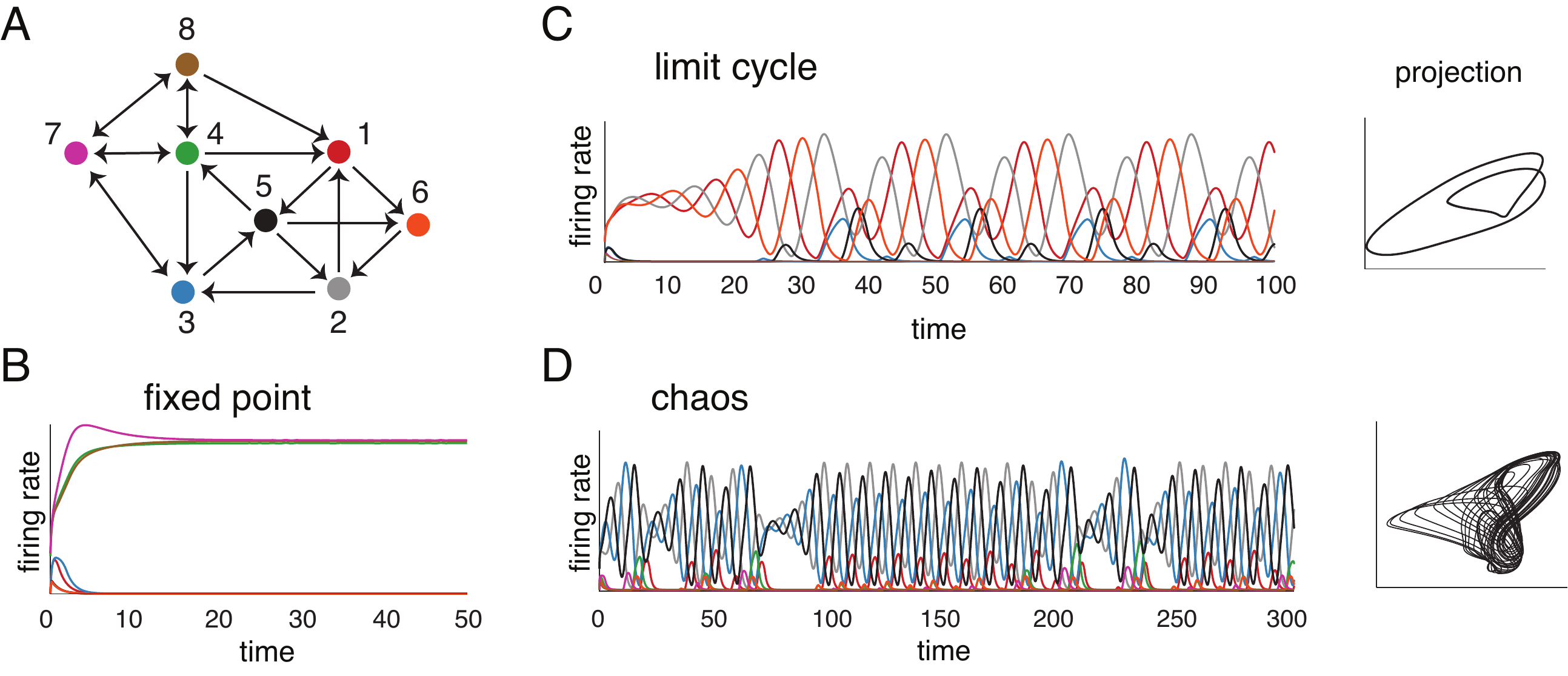}

                \caption{An example on $n=8$ nodes having several different behaviors based on initial conditions: Panel A shows the graph. Panel B shows one of two stable fixed points, Panel C shows a limit cycle, and Panel D shows a chaotic attractor. The traces of activity are color-coded to match the colors of the nodes in the graph. The plots on the far right show random two-dimensional projections of the 8-dimensional trajectories corresponding to the limit cycle and the chaotic attractor. Adapted from \cite{CTLN}.}
                \label{fig:multistability}
\end{figure}

\section{Threshold-linear networks}

The CTLN model is a specific type of threshold-linear network. Neuroscientists use threshold-linear networks to model recurrent neural networks \cite{neuralnetworktext}. These networks are thought to be involved in perception and memory processes \cite{seung}.  Development of the mathematical theory behind threshold-linear networks is ongoing \cite{thr-lin,seung,CTLN}.   

\begin{defint} A \textit{threshold-linear network} on $n$ neurons is defined by the following system of differential equations:
\begin{equation}\label{thrlin}\dfrac{dx_i}{dt} = -x_i + \left[\sum_{j=1}^{n}W_{ij}x_j+\theta\right]_+,\hspace{5mm} i\in [n].\end{equation}
\noindent where $x_i$ is the firing rate of the $i$th neuron, $W$ is the matrix of connection strengths, $\theta\in\R$ is the external drive to the network, and $[y]_+ =$ max$\{0,y\}$ is the threshold nonlinearity.\end{defint}

In our neural network context, the $-x_i$ represents the leak term and guarantees that the activity of neuron $i$ will die out in the absence of other inputs. Inside the nonlinearity we have a sum of inputs from all other neurons weighted by the connection strengths. In inhibitory networks ($W_{ij}\le0$), the parameter $\theta$ must be positive in order for the nonlinear term to be nonzero. 

We study the behavior of threshold-linear networks of $n$ neurons as they are one of the simplest examples of a nonlinear system of ordinary differential equations. In particular, we are interested in studying the dynamics of a competitive threshold-linear network defined from a simple directed graph.  \begin{defint} A \textit{competitive} threshold-linear network is governed by Equation \ref{thrlin} with the added restriction that $W_{ij}\le 0$ and $W_{ii}=0$ for all $i,j=1,\ldots,n$ and $\theta>0$.\end{defint}

In the next section we will describe the CTLN model, which is a particular type of competitive threshold-linear network.

\section{Description of the CTLN model}

The Combinatorial Threshold-Linear Network model (CTLN model) was first introduced by Curto et al.\ in \cite{patterncompletion} and further explored by Morrison et al.\ in \cite{CTLN}.  This model was designed as a way to study high-dimensional nonlinear dynamics without using a linear approximation \cite{CTLN}.  Linear models are limited as tools for approximation as they do not demonstrate complex behaviors such as limit cycles, multistability, and chaos. The nonlinearity in the CTLN model captures the full range of nonlinear behaviors, but is still simple enough that it is possible to develop a corresponding mathematical theory.  In this chapter we describe the CTLN model and necessary background.

\begin{defint}The \textit{Combinatorial Threshold-Linear Network (CTLN) model} refers to the competitive threshold linear network constructed from a simple directed graph with only two values for the inhibitory connection strengths. For any $\delta>0$ and $0<\varepsilon<1$, the $n\times n$ connectivity matrix $W$ is given by \begin{equation}\label{Wrule} W_{ij}=\begin{cases}
\hspace{0.2in} 0 \hspace{0.36in} \text{if}~ i=j\\
-1+\varepsilon~~~ \text{if} ~i\leftarrow j ~\text{in}~ G\\
-1-\delta~~~\text{if} ~i\nleftarrow j ~\text{in}~ G\\
\end{cases}\end{equation}
where $i \leftarrow j$ represents a directed edge from node $j$ to node $i$ in the graph $G$ and $i \nleftarrow j$ means that no such edge exists in $G$ \cite{CTLN}.
\end{defint}

The CTLN networks are therefore fully inhibitory, with each node acting on its neighbors by quieting their activity.  A biological motivation for the model is shown in Figure \ref{fig:seaofinhibition}.  Inhibitory interneurons (gray circles) inhibit all neighboring excitatory pyramidal cells (colored triangles) equally \cite{patterncompletion}. Connections between excitatory neurons therefore have two strengths. A directed edge represents an overall connection strength of $-1-\varepsilon$, i.e. inhibition has been weakened by an excitatory connection. Lack of an edge represents an overall connection strength of $-1-\delta$. The following theorem from \cite{CTLN} gives constraints on the graph and the relationship between $\delta$ and $\varepsilon$ that guarantees bounded activity, but disallows stable fixed points. 

\begin{figure}[h]
\centering
\includegraphics[width=120mm]{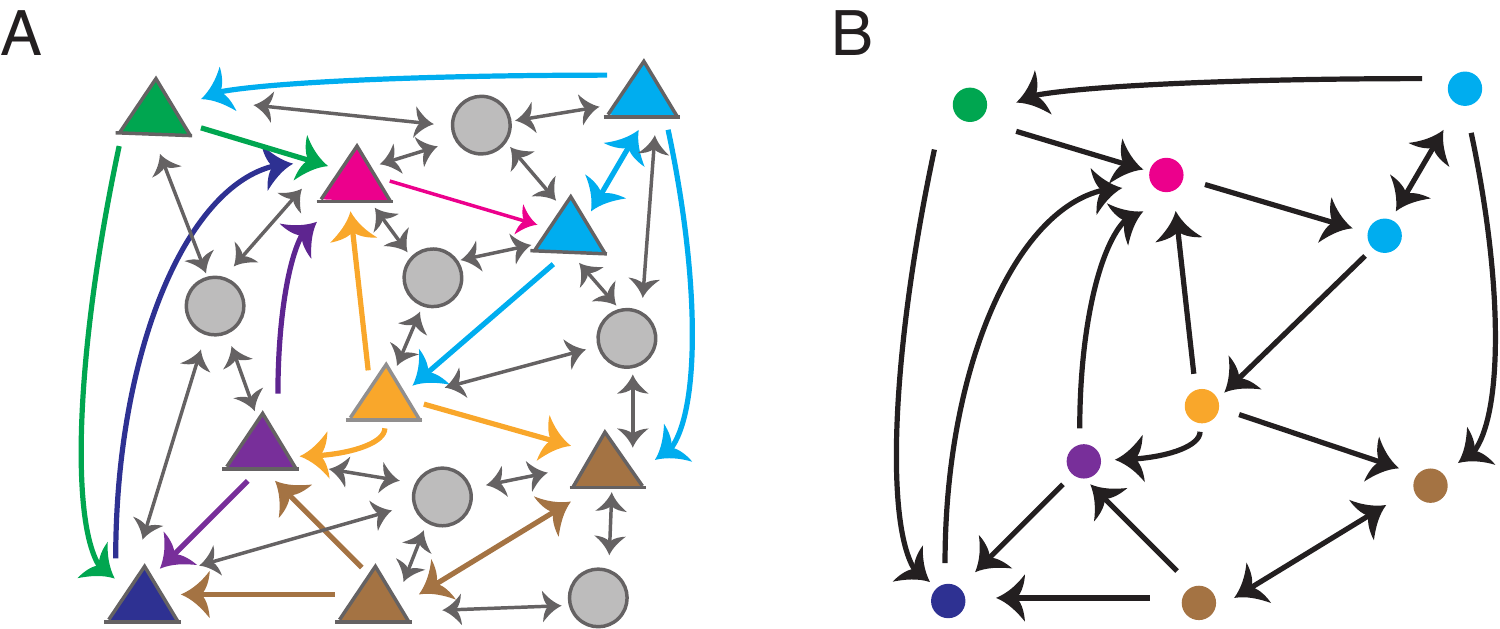}
                \caption{Diagram of excitatory and inhibitory connections:  The left panel shows inhibitory interneurons (gray circles) and excitatory pyramidal cells (colored triangles). Arrows indicate connections between neurons.  The right panel shows just the excitatory neurons and their connections. Adapted from \cite{patterncompletion}.} 
                \label{fig:seaofinhibition}
\end{figure}

\begin{theorm}\label{nosinks}\text{\cite{CTLN}}
Let $G$ be an oriented\footnote{An oriented graph is a directed graph with no bi-directional connections.} graph with no sinks (i.e. every vertex has outdegree at least 1), and consider the associated CTLN model with $W=W(G,\varepsilon,\delta)$. If $\varepsilon<\dfrac{\delta}{1+\delta}$, then the network has bounded activity and no stable fixed points.
\end{theorm}

By forbidding stable fixed points, Theorem \ref{nosinks} guarantees that the activity of the network is either oscillatory or chaotic.  Computational experiments show that most of the small networks for which Theorem \ref{nosinks} holds exhibit limit cycles where the neurons often appear to fire in sequence. The construction of the CTLN model guarantees that any differences in dynamics arise solely from differences in the underlying graph. Our goal is to use the structure of the graphs to predict the resulting sequences.  In the following sections we will assume $\theta=1$, $\varepsilon = 0.25$, $\delta = 0.5$, and we will focus on oriented graphs with no sinks. Theorem \ref{nosinks} holds in these cases. The next section shows some examples of these networks and their behaviors.  

\section{Examples and behaviors}\label{examples}

In this section, we will explore examples of CTLN networks with a small number of nodes. We will start with the simplest example of a network satisfying the conditions of Theorem \ref{nosinks}, a three-cycle as seen in Figure \ref{fig:3-cycle}.  Since a directed edge represents a less inhibited connection, the activity of such networks often follows the direction of the arrows, though not always. Note that in this case the dynamics are a limit cycle where the peak firing of the three nodes happens in the same order as the three-cycle in the graph.  

\begin{figure}[h]
\centering
\includegraphics[width=130mm]{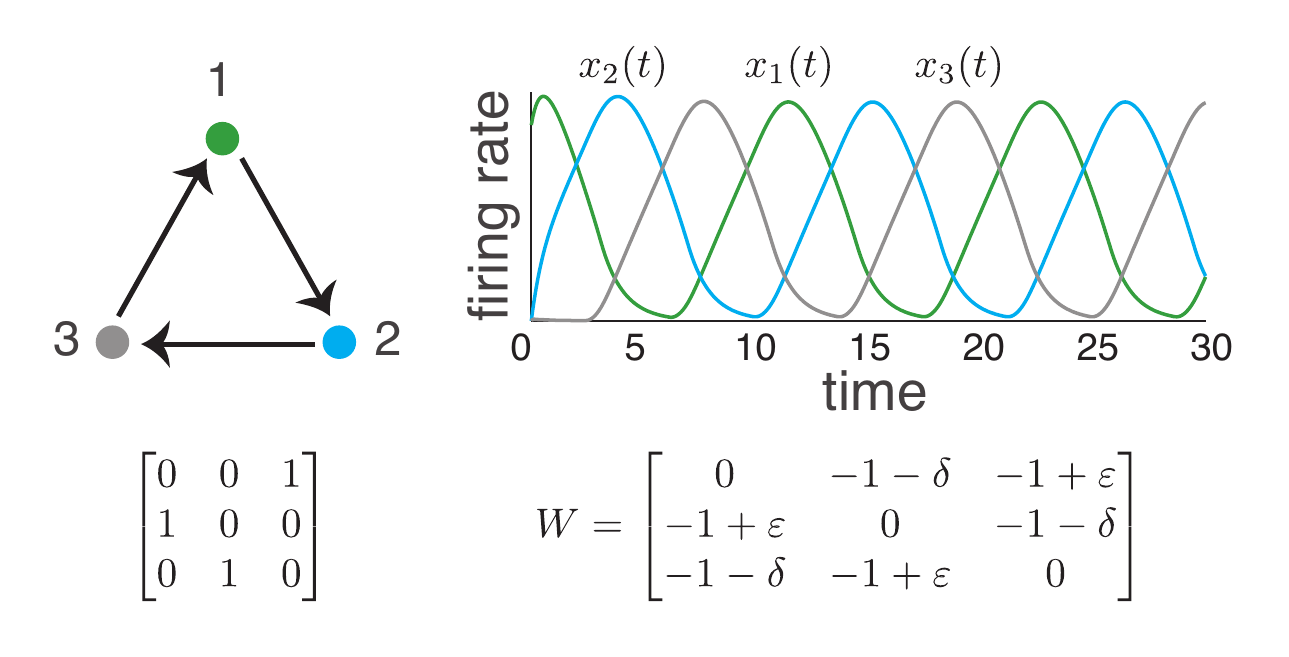}
                \caption{Example on $n=3$ nodes: In the top left we show the only oriented graph on $n=3$ vertices with no sinks. The top right panel shows the dynamics of this network: a limit cycle where the nodes fire in the order 123. The bottom left shows the transposed adjacency matrix (which is used for the CTLN model construction). The bottom right shows the matrix of connection strengths constructed using Equation \ref{Wrule}.  Adapted from \cite{CTLN}.}
                \label{fig:3-cycle}
\end{figure}

This three-cycle is the only graph on $n=3$ nodes that meets the criteria of Theorem \ref{nosinks}, i.e. is an oriented graph with no sinks. On $n=4$ nodes there are seven such graphs and on $n=5$ there are 152 such graphs. See Appendix \ref{catalogue} for the full catalogue of oriented graphs with no sinks on $n\le5$ nodes.  The number of graphs explodes when looking at oriented graphs with no sinks on $n>5$ nodes.  

Figure \ref{fig:n=5} shows several examples of networks and their dynamics on $n=5$ nodes. Panels A, B, and C show networks with limit cycles and the Panel D shows a chaotic attractor. Note that some of the nodes have different peak firing rates. In many cases on $n=5$ nodes we see three nodes with a relatively high peak firing rate and the remaining two nodes have a much smaller firing rate as seen in Figure \ref{fig:n=5} Panel A. We sometimes see synchronous firing of nodes, where the nodes fire at exactly the same rate, often resulting from a graph automorphism, as in Panel C. Chaotic attractors occur in networks as small as $n=5$ nodes. See Figure \ref{fig:n=5} Panel D for an example.  Note that it is possible for a network to have multiple limit cycles or chaotic attractors. For example, the network in Panel D has four different chaotic attractors, only one of which is shown.
 
\begin{figure}[h]
\centering
\includegraphics[width=155mm]{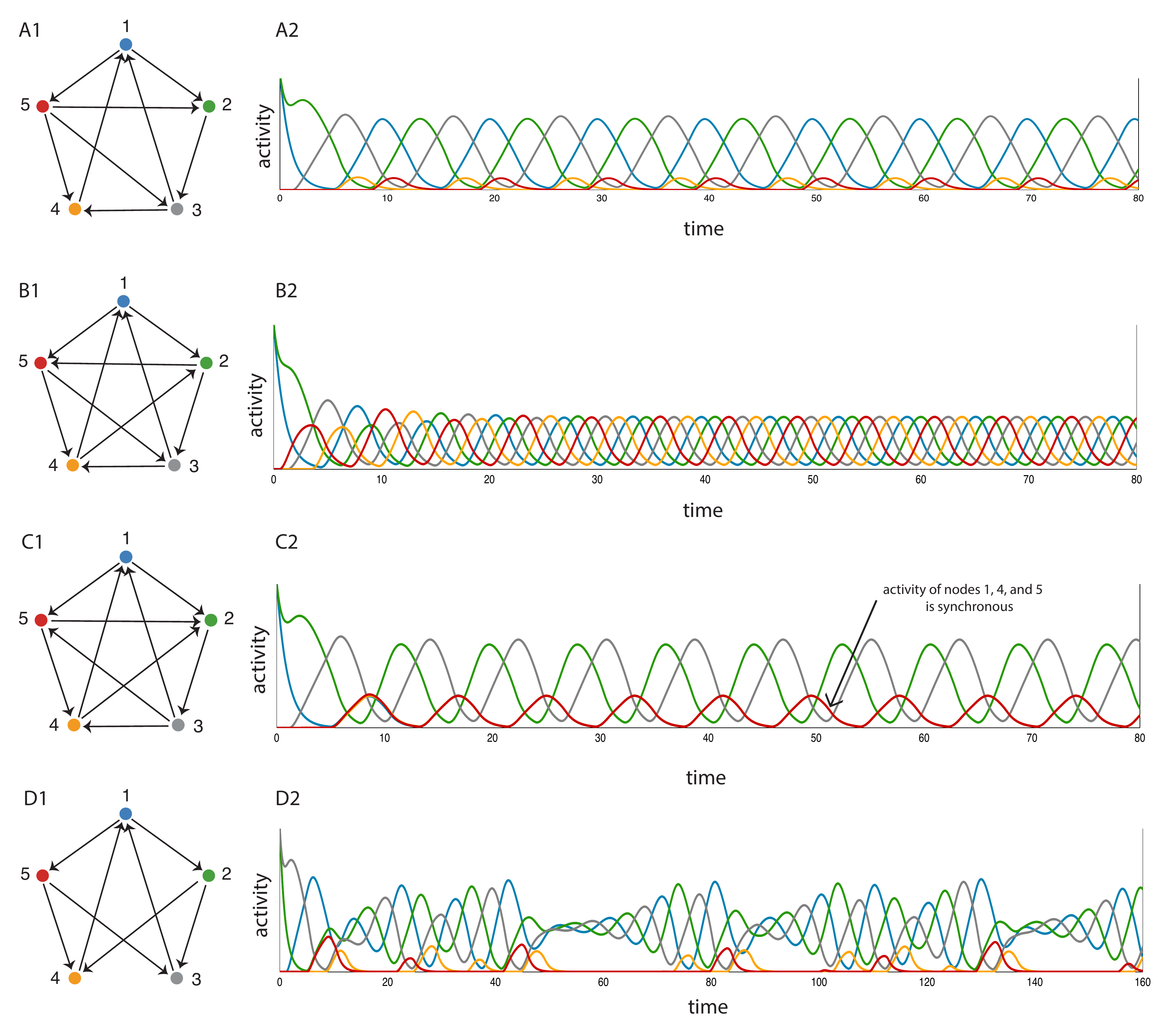}
                \caption{Examples on $n=5$ nodes:  Panels A1-D1 show some examples of oriented graphs on $n=5$ nodes. Panels A2-D2 show the dynamics of the networks corresponding to the graphs. Panel A2 shows a typical limit cycle. Note that nodes 4 and 5 fire at a much lower rate than nodes 1, 2, and 3. Panel B1 shows an example of a balanced subgraph. Note that each node has indegree 2 and outdegree 2. Panels C1-C2 show an example with synchronous firing: nodes 1,4, and 5 fire at the same rate. This is caused by the graph automorphism in C1. Panels D1-D2 show an example of a network with a chaotic attractor. Note that we show a longer trial in D2 to show the chaotic behavior. }
                
                \label{fig:n=5}
\end{figure}

 Could we have predicted these dynamics by looking at the graphs? The limit cycles in panels A and B, each follow a cycle in the corresponding graph. This is common in smaller graphs, but for larger $n$ we have seen examples where this is not the case. Also, what happens when there is more than one 5-cycle in the graph? Note that the graph in Panel B has two 5-cycles, 12534 and 15423, but the network has only one limit cycle.  Why does the network preferentially choose one 5-cycle over the other?  Figure \ref{fig:n=7} shows an example on $n=7$ nodes. Note that the limit cycle shown in the activity trace corresponds to a 6-cycle in the graph. Why does node 2 stop firing? It receives input from three other nodes while nodes 3 and 5 only receive two inputs each. Additionally, there are multiple 7-cycles in the graph. Why doesn't the limit cycle correspond to one of these 7-cycles? These are all questions that our algorithm must address. The algorithm must be able to discern between limit cycles and chaos, predict the firing sequence for limit cycles, and must also identify which nodes stop firing, if any.

\begin{figure}[h]
\centering
\includegraphics[width=150mm]{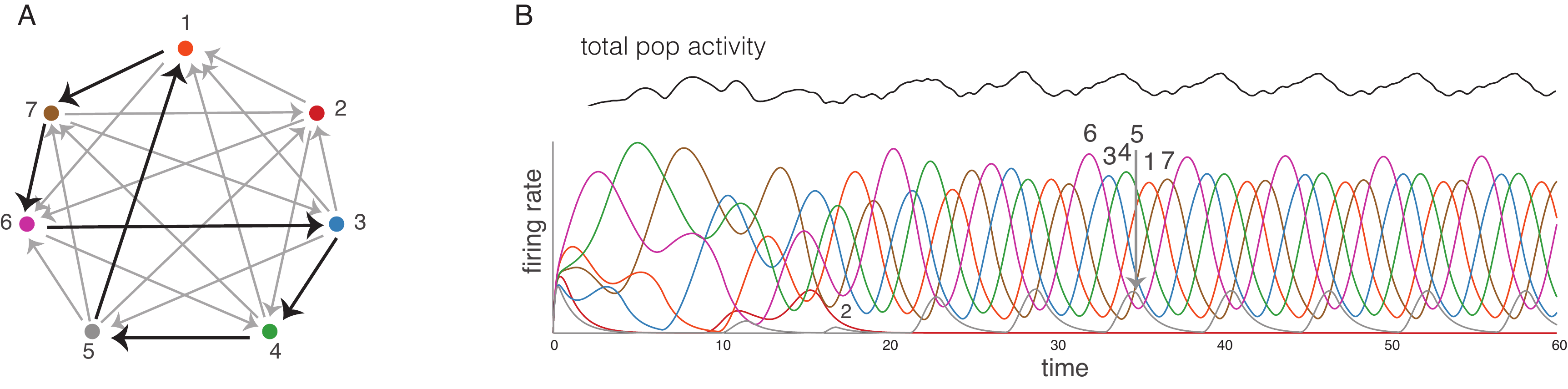}

                \caption{Example on $n=7$ nodes: The limit cycle has a sequence of 634517 as indicated above the dynamics. Note that node 2 (red) stops firing. Adapted from \cite{CTLN}.}
                \label{fig:n=7}
\end{figure}

\chapter{Sequence prediction algorithm}\label{algorithm}
In this section, we will discuss an algorithm we designed to predict the neural sequence from the graph for CTLN networks.  The basic premise of this algorithm comes from the idea of removing the ``weakest'' node and looking at the dynamics of the remaining network. Once we know how the smaller network behaves we work to figure out a way to tell where the deleted node fits in the sequence. We start with a description of the algorithm, discuss some examples, and then state some conjectures about when the algorithm is successful. We will look first at the case of tournaments\footnote{A tournament is a complete simple directed graph, where complete means that there is an edge between each pair of vertices.} without sinks and then examine oriented graphs without sinks.

\section{Description of the algorithm }
 The algorithm has two separate phases. In the deconstruction phase, we will first deconstruct the graph by deleting one vertex at a time, keeping track of the order of deletion. Then in the reconstruction phase, we start with a base sequence based on our deconstruction and rebuild the neural sequence by adding back in the deleted vertices in reverse order. Let $\mathcal{G}$ be a tournament on $n$ vertices with no sinks.
\paragraph{Deconstruction phase.} At each step of the algorithm we delete one of the vertices of $\mathcal{G}$ with smallest indegree such that the resulting reduced tournament has no sinks. We continue to delete vertices until we can no longer do so. This occurs when the resulting tournament is a three-cycle (see Proposition \ref{3cycle}). We refer to this three-cycle as the \textit{core cycle}. At each step we record the current tournament and the vertex we deleted. See Figure \ref{fig:algorithm} for an example. Note that the choice of vertex to delete is not necessarily unique, so it is possible for the algorithm to output multiple sequences.

\begin{figure}
\centering
\includegraphics[width=160mm]{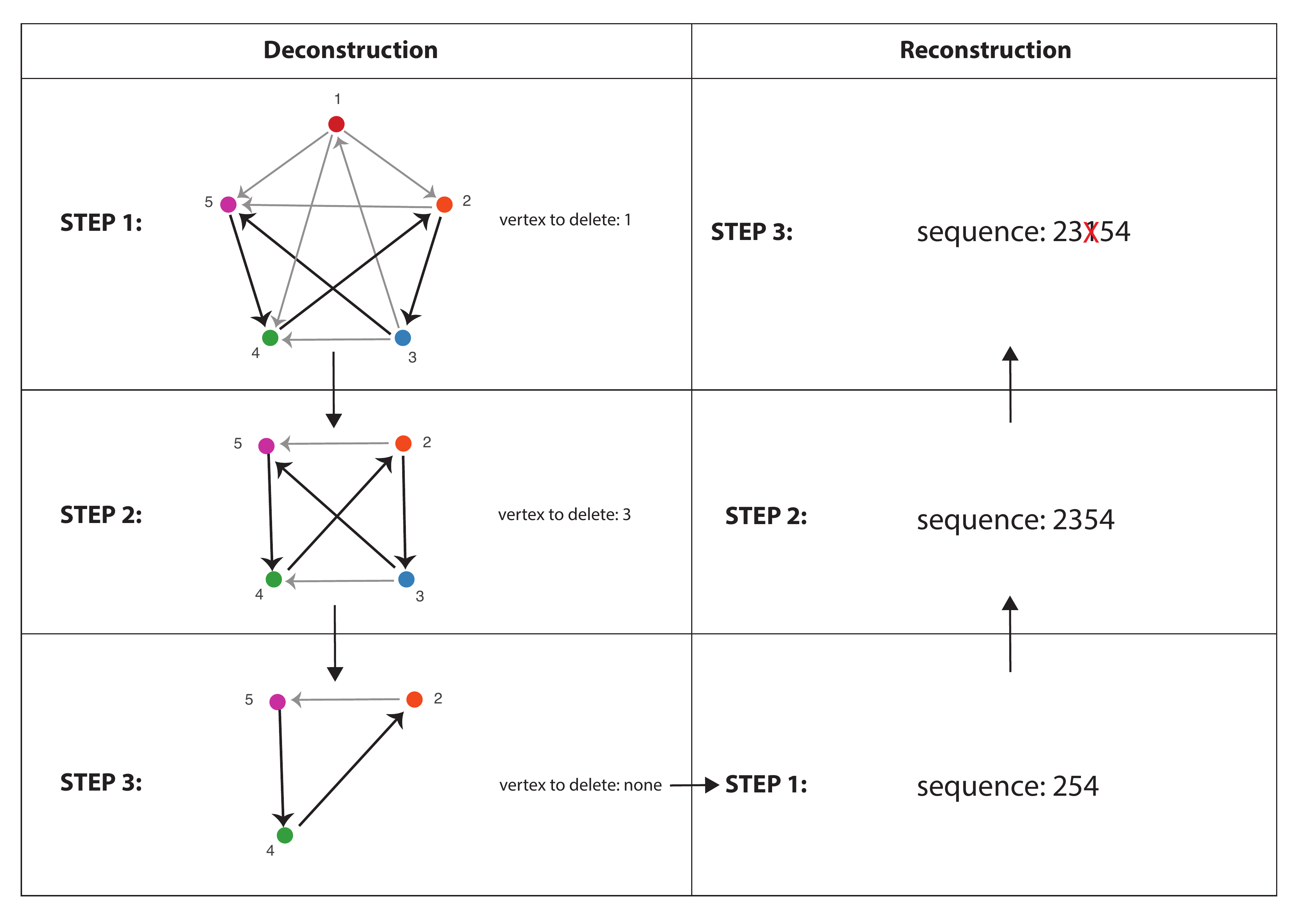}
                \caption{Example of algorithm on $n=5$}
                \label{fig:algorithm}
\end{figure}

\paragraph{Reconstruction phase.}
To reconstruct the sequence, we start with the three-cycle from the last step of the deconstruction phase. Recording the vertices in order of the three-cycle we insert the other vertices into this sequence in reverse order of deletion. Proceeding backwards through the list of deleted vertices, we add each vertex back into the sequence following the vertex that feeds into it in the graph from the preceding step in the deconstruction. If more than one vertex feeds into the vertex to be added, we look at the subgraph induced by these possibilities and if one of these possibilities is a sink in the induced subgraph, we place the vertex to be added after the sink in the sequence. See Figure \ref{fig:algorithm} for an example. Note that it is possible for the algorithm to fail if there are two or more edges feeding into the node we are adding back in.

\paragraph{Node death.}
When reconstructing the sequence, there are rules to predict the death of a node. If the vertex we delete at a given step has indegree zero then we do not add that node back in during reconstruction. If the vertex we delete at a given step has indegree one with the one edge coming from a vertex not in the core cycle then we do not add that node back in during reconstruction.

\paragraph{Final sequence.} 
To predict the final sequence, we first consider the full list of possibilities.  If there are two possibilities that are identical except one is missing a node, then we choose the shorter sequence. 
 If the possibilities are different but the same length, some neurons might fire synchronously. We predict the synchronous firing of a subset of neurons if that subset appears in the same cyclic order in each of the possibilities but with a different starting point. The neurons not in the subset appear in the same order in all possibilities. For example, we would predict that 2, 3, and 4 fire synchronously if the algorithm output sequences 12345, 13425, and 14235.

\paragraph{Implementation of the algorithm in Matlab.} We have developed a Matlab code to automate the prediction algorithm. The code can be found in Appendix \ref{algorithmcode}.

\begin{proposition}\label{3cycle}
For a tournament $\mathcal{G}$ with no sinks, the deconstruction phase of the algorithm will terminate if and only if the graph corresponding to the current step is a three-cycle. 
\end{proposition}
\begin{proof}
Let $\mathcal{G}$ be a tournament on $n$ vertices with no sinks. Assume you reach a step in the algorithm where there are currently $m$ nodes remaining and deleting any vertex results in an illegal graph. Then each node must have at least one incoming edge from a vertex with outdegree exactly 1. This implies that every vertex has outdegree exactly 1. Thus the graph has a total outdegree of $m$, i.e. we have $m$ total edges. The number of edges can also be given by $\binom{m}{2}$ since it is a tournament, so $m=\binom{m}{2}$. Solving for $m$ gives $m=3$. Since we have outdegree 1 at every node, the current graph is a 3-cycle. Also note that if the graph at the current step is a three-cycle, then deleting any vertex, will result in a graph with two vertices and a directed edge between them. Thus one of the vertices is a sink, so the deconstruction phase terminates at at step $n-3$. 
\end{proof}

\section{Performance of the algorithm}

\begin{proposition}\label{algworks}
For $\varepsilon = 0.5$, $\delta = 0.25$, and $\theta = 1$, the algorithm correctly predicts the neural sequence for tournaments without sinks on $n\le5$ nodes.
\end{proposition}
The proof of this proposition is done by checking each case computationally.
 We also note that for oriented graphs, the algorithm appears to predict which neurons will have high firing rates. For small examples, we often see three neurons with higher firing rates than the remaining neurons. These three high-firing neurons appear to correspond to the neurons in the core cycle predicted by the algorithm.  Further exploration of  tournaments lacking sinks on $n>5$ nodes indicates that there are networks where the algorithm does not correctly predict the behavior, often in the form of spurious predictions or incorrectly predicting neuron death. Analysis of these graphs indicates that the networks for which the algorithm fails appear to have the common property of having a balanced subgraph on $n\ge5$ nodes or are an outerneuron construction.

\begin{defint}
A \textit{balanced subgraph} is a complete induced subgraph of an oriented graph $\mathcal{G}$ where all nodes have the same outdegree. Note that for a balanced subgraph of size $m$, where $m$ is odd, the outdegree of each vertex is $\frac{m-1}{2}$. There are no balanced subgraphs of even size. See Panel B in Figure \ref{fig:n=5} for an example of a balanced graph on 5 vertices. Each vertex has indegree 2 and outdegree 2. \end{defint}

\begin{defint}
The \textit{outerneuron construction} is the process of taking a simple directed graph on $n$ vertices and adding two vertices to the graph: one vertex with edges directed to all vertices in the original graph (a pseudo-source) and one vertex who receives directed edges from all the vertices in the original graph (a pseudo-sink). We then add a directed edge from the pseudo-sink to the pseudo-source to guarantee that the new graph on $n+2$ vertices has no sinks. \end{defint}

We have looked at all tournaments having no sinks on up to $n=7$ vertices. Proposition \ref{algworks} gives that the algorithm works for the 11 such graphs on $n\le5$ vertices. 
On $n=6$ nodes there are 44 graphs and on $n=7$ nodes there are 400 such graphs. The algorithm fails for only two of the $n=6$ graphs, one with a spurious prediction and one with a spurious deletion. Both of these graphs have a balanced subgraph on 5 vertices. For $n=7$, if we ignore graphs with an outerneuron construction and graphs with balanced subgraphs with $n\ge5$ vertices (153 total), the algorithm only fails for 5 out of the remaining 247 graphs. Additionally, of the 153 graphs having an outerneuron construction or a balanced subgraph with $n\ge5$, 62 of these are still correctly predicted by the algorithm.

\section{Extending the algorithm to oriented graphs}

We have also explored the success of the algorithm on oriented graphs without sinks on $n\le 5$ nodes. Recall that unlike tournaments, oriented graphs do not require an edge between every pair of vertices. As a result we are not guaranteed that the algorithm will terminate at a three-cycle for oriented graphs. In fact, we have seen examples on $n=5$ where the algorithm terminates in a 4- or 5-cycle. A comprehensive list of oriented graphs without sinks on $n\le5$ nodes appears in Appendix \ref{catalogue}. The algorithm correctly predicts the behavior in all but 6 of the 160 total networks on $n\le5$ nodes. Using the algorithm on the Graph \#147, 148, 149, 152, and 158 predicts the correct sequence, but also produces a spurious prediction. For example, the network corresponding to Graph \#147 has a limit cycle with sequence 12(45)3 where the parentheses indicate that neurons 4 and 5 fire synchronously. The algorithm makes three predictions: 12453, 12543, and 12534.  The first two predictions result in a correct final sequence of 12(45)3, but 12534 is a spurious prediction. Using the algorithm on Graph \#153, a node is deleted that does not die, which keeps the algorithm from predicting the synchronous activity.

One weakness that arises when using the algorithm on oriented graphs rather than tournaments is that it is possible to disconnect the graph during deconstruction. Since the reconstruction rules will not necessarily make sense in this case, additional rules will be necessary for the deconstruction of oriented graphs to avoid breaking the graph into multiple components. We also still sometimes make incorrect predictions in the case of graphs with the outerneuron construction and/or a balanced subgraph. The next section shows a comprehensive study of all the oriented graphs on $n\le5$ vertices in order to investigate ways to adjust the algorithm for oriented graphs.

\section{An application: classification of oriented graphs on $n\le5$}

To investigate why the algorithm fails in some cases, we perform an exhaustive study and classification of oriented graphs on $n\le5$ nodes. See Appendix \ref{catalogue} for a complete list of graphs and their classification.  If we look at the graphs for which the algorithm failed, we see that they all fall into a category where there are two or more different $n=4$ subgraphs possible in the first step in the algorithm. Using this classification we hope to start refining the algorithm to work in more generality.

In the next section we also sort the graphs by dynamics.  Note that the graphs in the same entry in the dictionary often also appear in the same category of the classification in Appendix \ref{catalogue}.

\section{Dictionary of attractors for $n\le5$}\label{dictionary}
In this chapter we create a dictionary of graph behaviors. We sort the graphs into groups based on the dynamics of the corresponding network. Each entry in the dictionary corresponds to a specific limit cycle or chaotic attractor.  AT denotes ``attractor type.''  For each entry, we show a representative graph for that particular attractor type along with its dynamics.  To make the dynamics plots we use $\varepsilon = 0.5$, $\delta = 0.25$, and $\theta=1$ as before. Initial conditions used are listed above the dynamics plots. We use an asterisk (*) to indicate which initial condition we used to create the plot. The sequence listed is for the representative graph only. We underline low-firing neurons and synchronous neurons are in parentheses. All graphs listed in that entry exhibit the same behavior (up to permutation) as the representative graph for some initial condition. If a graph has more than one attractor we annotate the graph number with \_ic1, \_ic2, etc. The labelling of the graphs corresponds to the numbering in the catalogue found in Appendix \ref{catalogue}. The compilation of the dictionary was carried out in collaboration with Katherine Morrison. 

\newpage
\includegraphics[width=155mm]{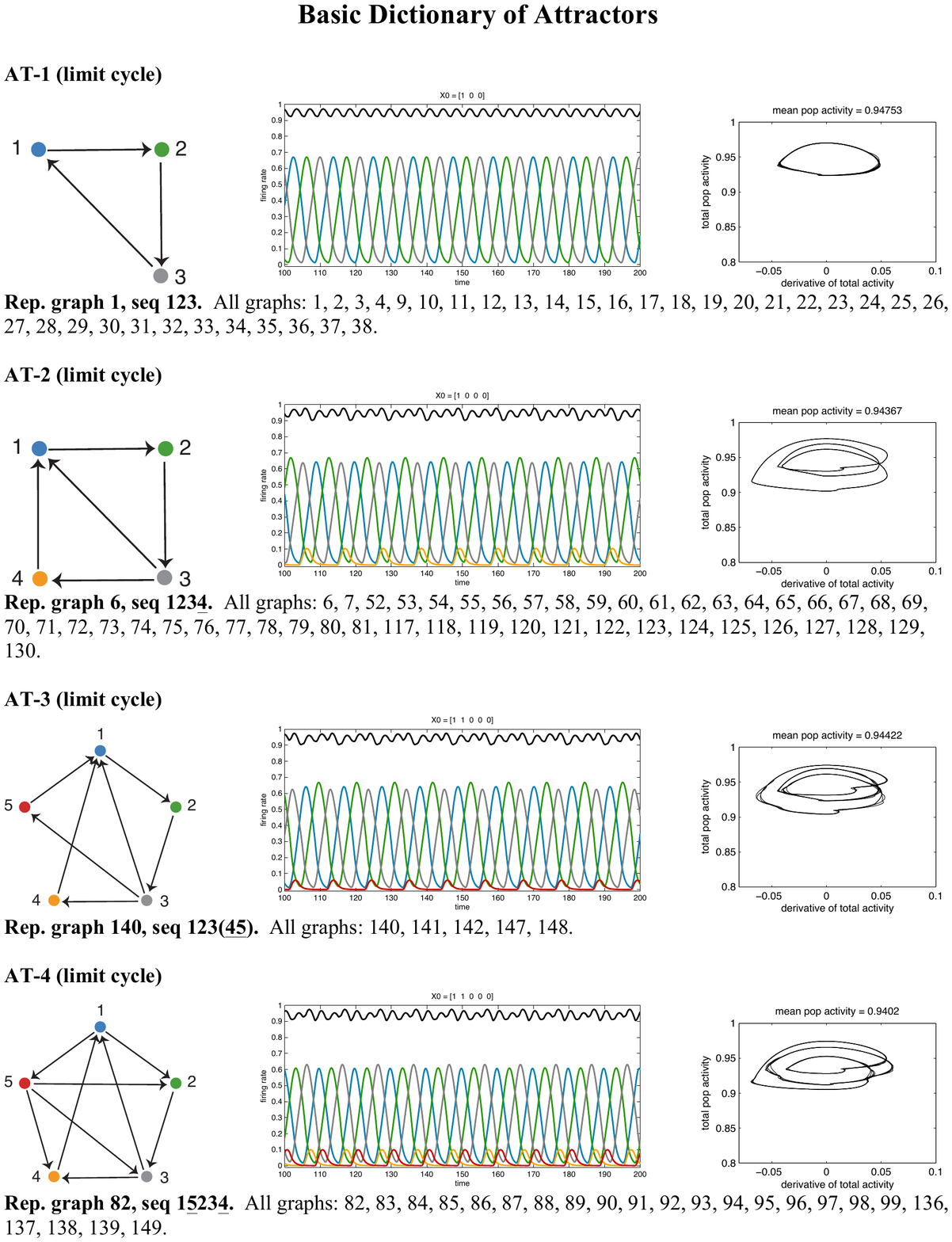}
\newpage
\includegraphics[width=155mm]{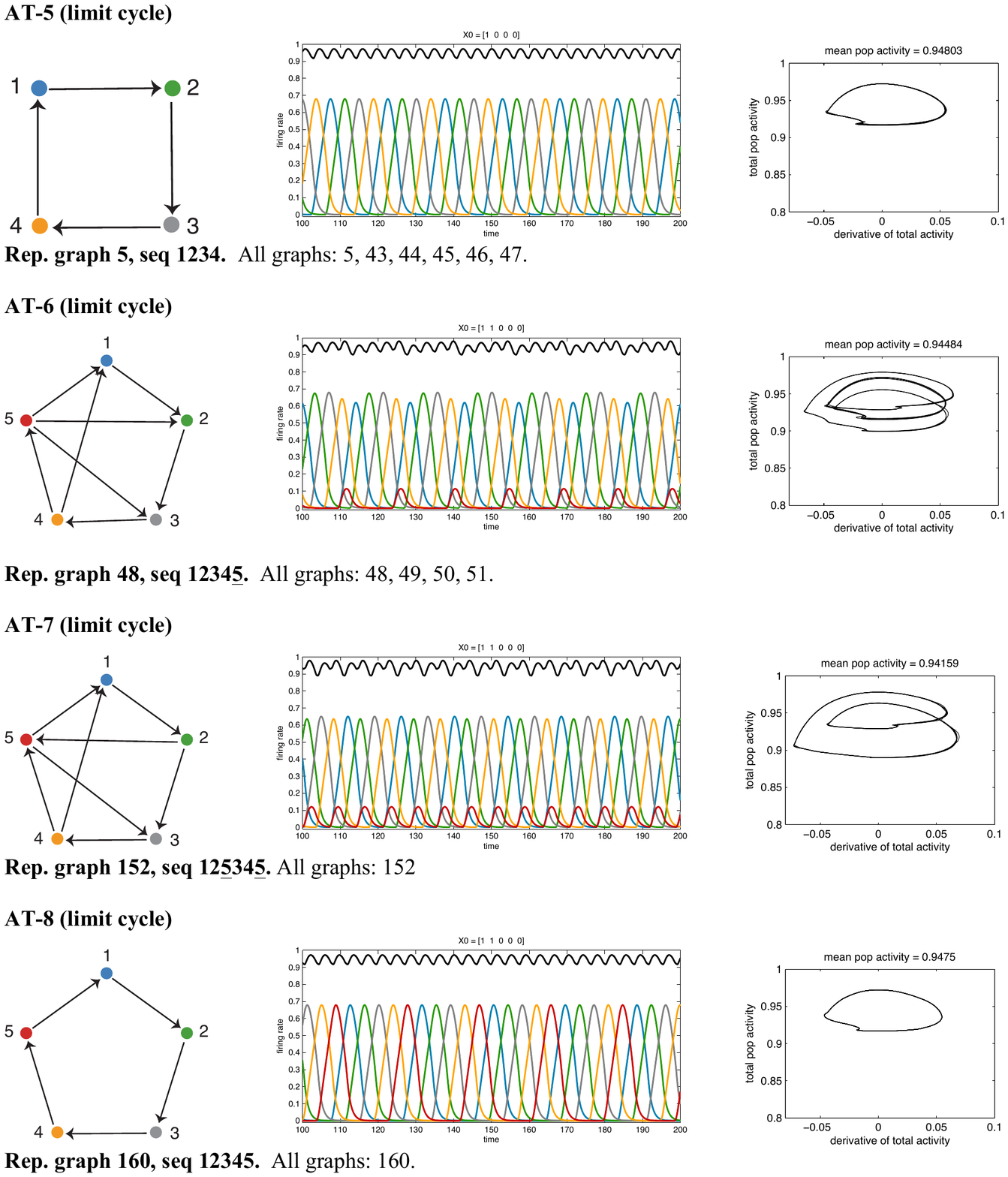}
\newpage
\includegraphics[width=155mm]{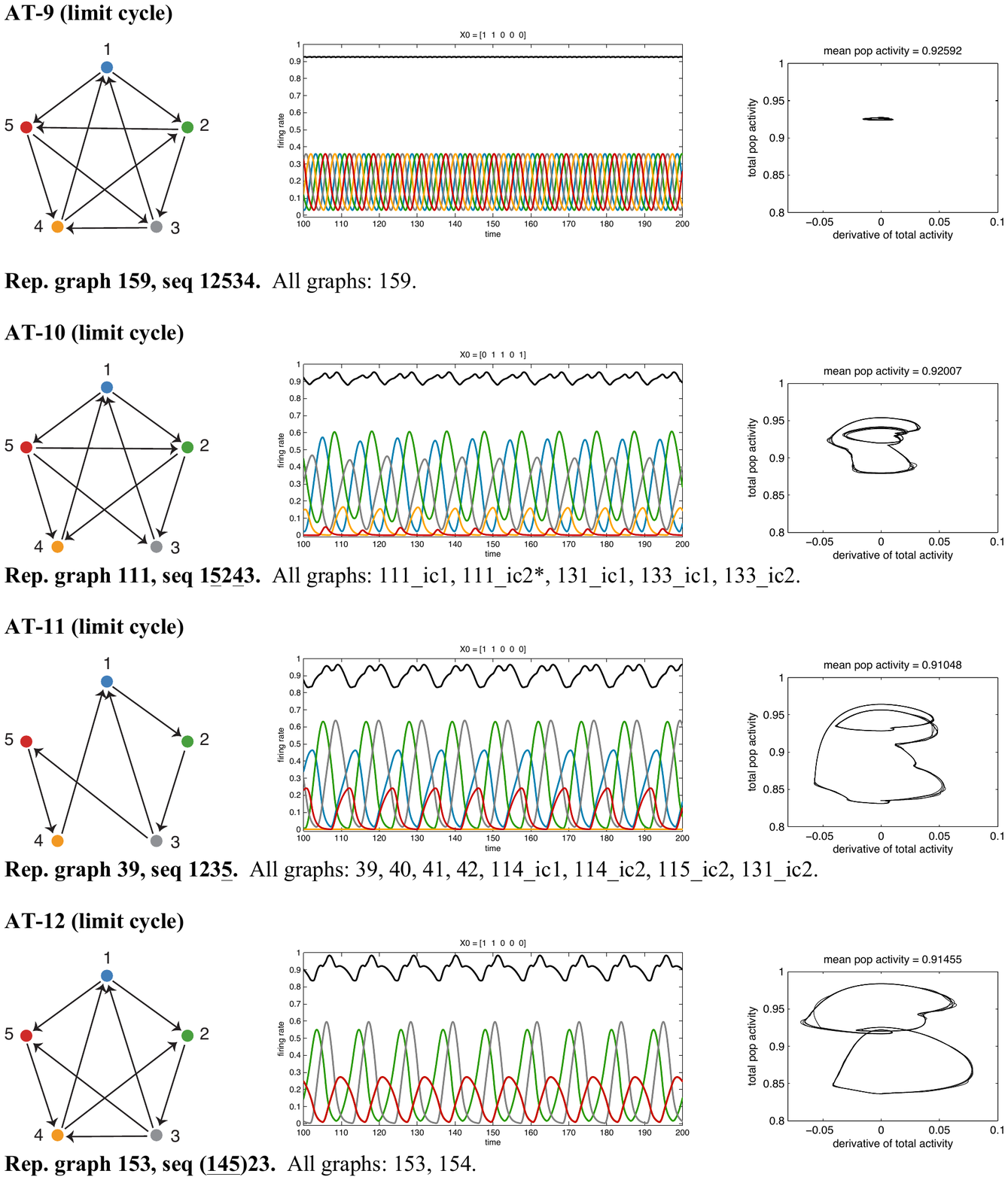}
\newpage
\includegraphics[width=155mm]{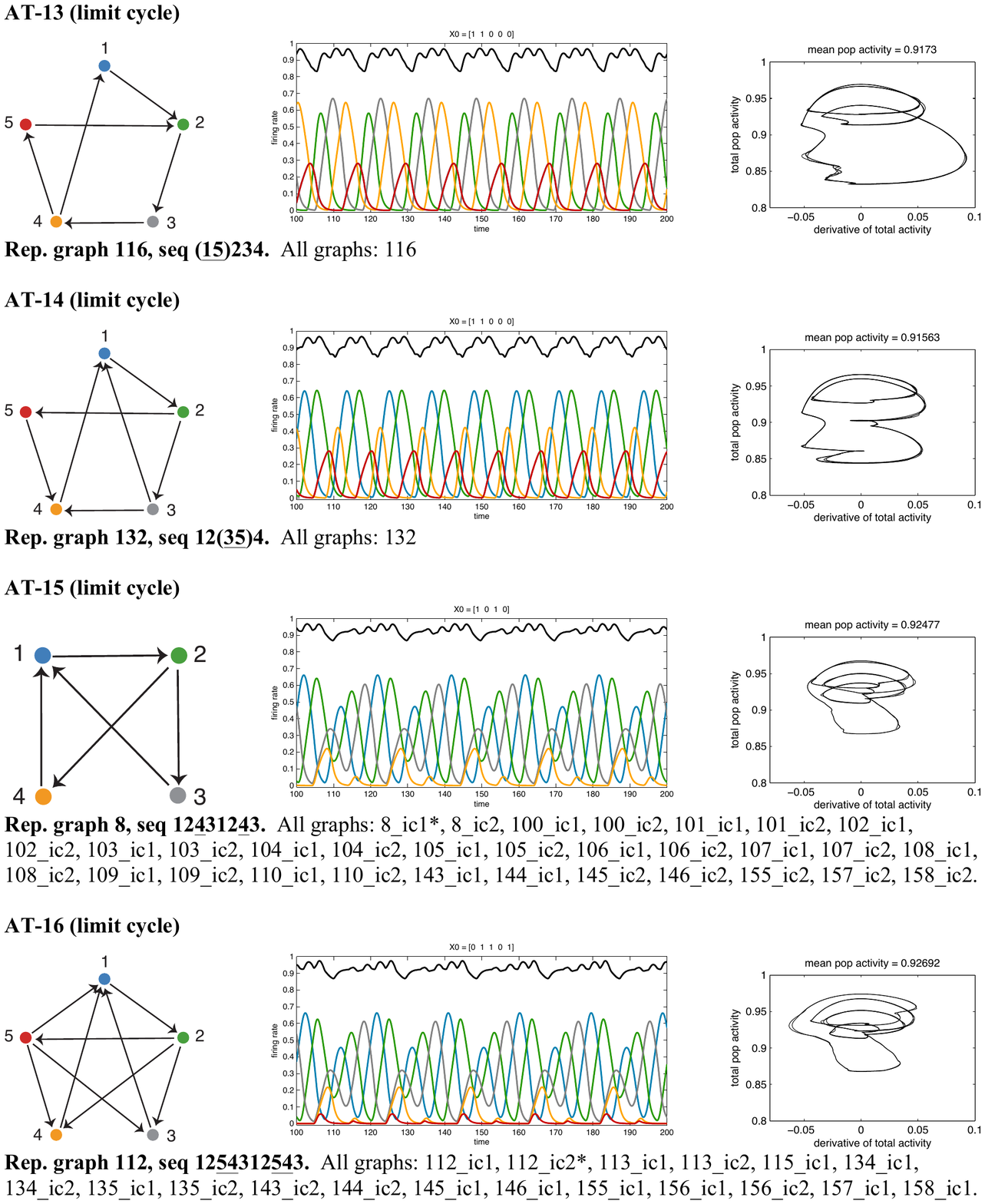}
\newpage
\includegraphics[width=155mm]{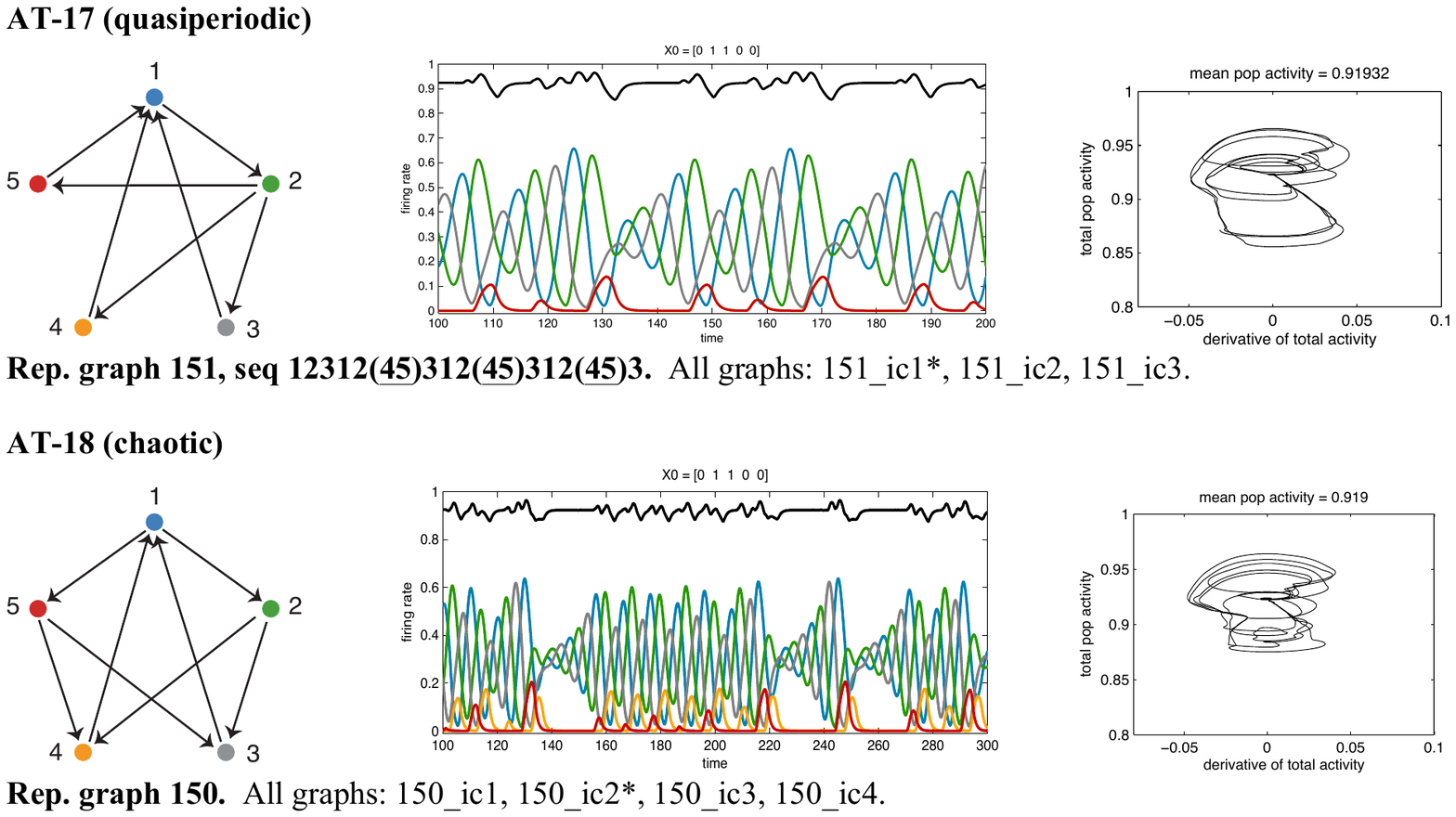}

\section{Future work}

\paragraph{Balanced subgraphs and outerneuron constructions.} For the subset of graphs whose sequences cannot be predicted by our current algorithm, we need to adjust the current algorithm or design a new algorithm to handle these cases.

\paragraph{Refining the algorithm for oriented graphs.} Using the classification of oriented graphs on $n\le5$ nodes, we can look for commonalities in the graphs where the algorithm fails and adjust the algorithm accordingly.

\paragraph{Exploring larger networks.} After classifying the networks and behaviors for oriented graphs on $n\le 5$ one of the next steps is to take a closer look at graphs with $n>5$.

\paragraph{Proving conjectures about algorithm.} We have many conjectures about the prediction algorithm, so we would like to look for proofs or counterexamples.

\begin{conjecture}
In an oriented graph with no sinks on $n=5$ nodes, two nodes fire synchronously if the algorithm predicts two sequences that are identical except for with the two nodes switched and neither node is part of the core cycle.
\end{conjecture}

\begin{conjecture} During the reconstruction phase, if node $u$ has indegree zero in the corresponding subgraph, then this node dies.
\end{conjecture}

\begin{conjecture}
During the reconstruction phase, if node $u$, with indegree 1 in the corresponding subgraph, is added to the sequence following a node that was not part of the base sequence then $u$ dies, unless $u$ is part of a balanced subgraph.
\end{conjecture}

\begin{conjecture}
If the node we are placing back in the sequence has indegree $>$ 1, then consider the subgraph induced by the vertices which contribute to the indegree. If the induced subgraph has a sink then place the vertex after the sink in the sequence. 
\end{conjecture}


\backmatter

\appendix

\chapter{Code for random walk model of replenishment}\label{compmodelcode}
This appendix gives the Matlab code for the random walk model of replenishment. There are six functions total:
\begin{itemize}
\item Trials\_for\_plot
\item R\_evolve
\item create\_SiteMat
\item R\_update
\item N\_evolve
\item N\_update
\end{itemize}
 The main function Trials\_for\_plot takes ribbon type (1 for ribbon and 0 for nonribbon), vesicle diameter (in $\mu$m), diffusion coefficient (in $\mu$m$^2$/s), attachment probabilities for both populations (as in Models 1 and 2), the fraction in population $A$, the vesicle concentration (in vesicles/$\mu$m$^3$), and the number of trials and outputs a cell array where entry $\{i,j,1\}$ is a vector of the number of vesicles on the ribbon at each time step in the computational model and entry $\{i,j,2\}$ is a vector of the number of vesicles on the ribbon at each time step as predicted by the theoretical random walk model for the $i$th concentration value and the $j$th trial. Plotting these two vectors against time gives us the replenishment curves in both cases.

{\small
\begin{lstlisting}[style=Matlab-editor]%[basicstyle=\linespread{0.8}\listingfont]

function [Trials] = Trials_for_plot(ribbon,delta,Diff,s_A,s_B,f,concen,numtrials)
%%%%%% initial state %%%%%%%%
W = 50; %width of matrix (pick an even number)
H = 50; %height of matrix
D = 31; %depth of matrix  
n = 110; %number of ribbon sites, must be a multiple of 10
f = 1; %fraction of fast-replenishing sites 0.757 for -10mV and 0.54 or -30mV
s = s_A*f+s_B*(1-f); %average attachment probability
dim_vec = [H,W,D];
timevector = floor(((harmonic(n)*2)/(delta^3*s))./(concen.*fracmob)); %corresponding length of time for each conc.
concenmax = length(concen); %number of concentration values 
Trials = cell(concenmax,numtrials,2); %collect data
%%%%%% start trials %%%%%%%%%%%
for concenidx = 1:concenmax
    rho = concen(concenidx); %vesicle density per micrometer^3 
    den = rho*delta^3; %density of occupied lattice sites
    T = timevector(concenidx); %number of time steps    
for trial = 1:numtrials
%%%%%%% creates initial state matrix %%%%%%%    
for i = 1:D
Mat(:,:,i) = rand(H,W)>(1-den); 
end
m = nnz(Mat); %number of vesicles
[r, c, l] = ind2sub(size(Mat),find(Mat == 1));
S_initial = [r,c,l]; % matrix whose rows are the coordinates of the vesicles
%%%%%% chooses update function %%%%%%%
if ribbon == 1 
    evolve_fun = @(S_initial) R_evolve(S_initial,dim_vec,T,s_A,s_B,f,n,shape);
else
    evolve_fun = @(S_initial) N_evolve(S_initial,dim_vec,T,s_A,s_B,f,n);
end
tic
[S_array,vec] = evolve_fun(S_initial); % update matrix T times
toc
if ribbon==1 
tau_approx = @(s) 1/(Diff*rho*delta*s);   %approx time constant for ribbon
else
tau_approx = @(s) 1/(2*Diff*rho*delta*s);  %approx time constant for nonribbon
end
end
x = 1:T-1;
x1 = [0 x].*(delta^2/(2*Diff));
y = zeros(1,T);
z = zeros(1,T);
for i=1:T
    y(i) = length(vec{i}); %vector of vesicles on ribbon at each time step
    z(i) = (n*f)*(1-exp(-x1(i)/tau_approx(s_A)))+(n*(1-f))*(1-exp(-x1(i)/tau_approx(s_B))); %vector of theoretically predicted number of vesicles at each time step
end
    Trials{concenidx,trial,1} = y;
    Trials{concenidx,trial,2} = z;   
figure; % plot theoretical prediction vs. computational trial
plot(x1,y,'-k')
hold on;
plot(x1,z,'-b')
hold off;    
end
end
end
\end{lstlisting} 

\begin{lstlisting}[style=Matlab-editor]

function [S_array,vec] = R_evolve(S,dim_vec,Tsteps,s_A,s_B,f,n)
% evolve function for ribbon case
SiteMat = create_SiteMat(dim_vec,n/10); % creates matrix of ribbon sites
W = dim_vec(1);
D = dim_vec(3);
perm = randperm(n);  %permutation of num of vesicles to randomly determine populations A and B
frac = floor(f*n);   %number in population A
indices = [];
for k = 1:length(S(:,1))
    if S(k,1)>(W/2)-1 && S(k,1)<(W/2)+1 && S(k,2)<10 && S(k,3)>((D+1)/2)-3 && S(k,3)<((D+1)/2)+3
       indices = [indices, k];
    end
end   
newindices = setdiff([1:size(S,1)],indices);
S = S(newindices,:);
RibbonMat = ones(size(S)); %matrix of ones with rows corresponding to vesicles on ribbon zeroed out
FilledMat = [];
[rowdim,coldim] = size(S);
S_array = zeros(rowdim,coldim,Tsteps);
S_array(:,:,1) = S;
vec = cell(1,Tsteps);
vec{1} = [];   
for t = 2:Tsteps
    index = find(ismember(S,SiteMat,'rows'));
    ves_idx = [vec{t-1}];
    index = setdiff(index,ves_idx);
    for j = 1:length(index)
        vesicle_idx = find(RibbonMat(:,1)==0);
        if isempty(find(ismember(S(vesicle_idx,:),S(index(j),:),'rows'),1)) == 1 
           if isempty(find(ismember(SiteMat(perm(1:frac),:),S(index(j),:),'rows'))) == 0
              if rand(1) <= s_A
                 RibbonMat(index(j),:) = 0;
                 FilledMat = [FilledMat; S(index(j),:)];
              end
           else
               if rand(1) <= s_B
                  RibbonMat(index(j),:) = 0;
                  FilledMat = [FilledMat; S(index(j),:)];
               end
           end    
        end
    end
    vec{t} = find(RibbonMat(:,1) == 0);
    S = R_update(S,RibbonMat,dim_vec,FilledMat,shape,n);
    S_array(:,:,t) = S;
end
end
\end{lstlisting}
\begin{lstlisting}[style=Matlab-editor]

function SiteMat = create_SiteMat(dim_vec,ht)
%creates matrix of the indices of the ribbon
W = dim_vec(1);
D = dim_vec(3);
SiteMat = [];
for ii = [W/2-1,W/2+1]
    for jj = 1:ht;
        for kk = (D+1)/2-2:(D+1)/2+2;
            SiteMat = [SiteMat;[ii,jj,kk]];
        end
    end
end 
\end{lstlisting}
\begin{lstlisting}[style=Matlab-editor]
function [S_new] = R_update(currentS,RibbonMat,dim_vec,FilledMat,n) 
%update function for the ribbon case
[rowdim,coldim]=size(currentS);
changeS=randi([0,1],size(currentS));
changeS=2*changeS-1;
D = dim_vec(3);
    Illegal_Mat = FilledMat;
    for ii = 1:n/10
        for jj = (D+1)/2-2:(D+1)/2+2
            Illegal_Mat = [Illegal_Mat;[dim_vec(1)/2,ii,jj]];
        end
    end
test_changeS = changeS.*RibbonMat;
test_S = currentS+changeS;
for ii = 1:length(dim_vec)
    test_S(:,ii) = min(max(test_S(:,ii),1),dim_vec(ii));
end
change_entries = find(ismember(test_S,Illegal_Mat,'rows'));
changeS(change_entries,:) = 0;       
changeS = changeS.*RibbonMat;
preS = currentS+changeS;
for jj = 1:length(dim_vec)
    S_new(:,jj) = min(max(preS(:,jj),1),dim_vec(jj));
end
end

\end{lstlisting}

\begin{lstlisting}[style=Matlab-editor]

function [S_array,vec] = N_evolve(S,dim_vec,Tsteps,s_A,s_B,f,n)
% evolve function for the nonribbon case
SiteMat = [[4,7,6];[4,16,6];[4,25,6];[4,34,6];[4,43,6];
[15,7,6];[15,16,6];[15,25,6];[15,34,6];[15,43,6];
[26,7,6];[26,16,6];[26,25,6];[26,34,6];[26,43,6];
[37,7,6];[37,16,6];[37,25,6];[37,34,6];[37,43,6];
[4,7,10];[4,16,10];[4,25,10];[4,34,10];[4,43,10];
[15,7,10];[15,16,10];[15,25,10];[15,34,10];[15,43,10];
[26,7,10];[26,16,10];[26,25,10];[26,34,10];[26,43,10];
[37,7,10];[37,16,10];[37,25,10];[37,34,10];[37,43,10];
[4,7,14];[4,16,14];[4,25,14];[4,34,14];[4,43,14];
[15,7,14];[15,16,14];[15,25,14];[15,34,14];[15,43,14];
[26,7,14];[26,16,14];[26,25,14];[26,34,14];[26,43,14];
[37,7,14];[37,16,14];[37,25,14];[37,34,14];[37,43,14];
[4,7,18];[4,16,18];[4,25,18];[4,34,18];[4,43,18];
[15,7,18];[15,16,18];[15,25,18];[15,34,18];[15,43,18];
[26,7,18];[26,16,18];[26,25,18];[26,34,18];[26,43,18];
[37,7,18];[37,16,18];[37,25,18];[37,34,18];[37,43,18];
[15,7,22];[15,16,22];[15,25,22];[15,34,22];[15,43,22];
[26,7,22];[26,16,22];[26,25,22];[26,34,22];[26,43,22];
[37,7,22];[37,16,22];[37,25,22];[37,34,22];[37,43,22];
[15,7,26];[15,16,26];[15,25,26];[15,34,26];[15,43,26];
[26,7,26];[26,16,26];[26,25,26];[26,34,26];[26,43,26];
[37,7,26];[37,16,26];[37,25,26];[37,34,26];[37,43,26]]; 
W = dim_vec(1);
D = dim_vec(3);
perm = randperm(n);  %permutation of num of vesicles to randomly determine populations A and B
frac = floor(f*n);   %number in population A
indices = [];
for k = 1:length(S(:,1))
    if S(k,1)>(W/2)-1 && S(k,1)<(W/2)+1 && S(k,2)<10 && S(k,3)>((D+1)/2)-3 && S(k,3)<((D+1)/2)+3
       indices = [indices, k];
    end
end   
newindices = setdiff([1:size(S,1)],indices);
S = S(newindices,:);
RibbonMat = ones(size(S));
FilledMat = [];
[rowdim,coldim] = size(S);
S_array = zeros(rowdim,coldim,Tsteps);
S_array(:,:,1) = S;
vec = cell(1,Tsteps);
vec{1} = [];    
for t = 2:Tsteps
    index = find(ismember(S,SiteMat,'rows'));
    ves_idx = [vec{t-1}];
    index = setdiff(index,ves_idx);
    for j = 1:length(index)
        vesicle_idx = find(RibbonMat(:,1)==0);
        if isempty(find(ismember(S(vesicle_idx,:),S(index(j),:),'rows'),1)) == 1 
           if isempty(find(ismember(SiteMat(perm(1:frac),:),S(index(j),:),'rows'))) == 0
              if rand(1) <= s_A
                 RibbonMat(index(j),:) = 0;
                 FilledMat = [FilledMat; S(index(j),:)];
              end
           else
               if rand(1) <= s_B
                  RibbonMat(index(j),:) = 0;
                  FilledMat = [FilledMat; S(index(j),:)];
               end
           end    
        end
    end
    vec{t} = find(RibbonMat(:,1) == 0);
    S = N_update(S,RibbonMat,dim_vec);
    S_array(:,:,t) = S;
end
end
\end{lstlisting} 

\begin{lstlisting}[style=Matlab-editor] 

function [S_new] = N_update(currentS,RibbonMat,dim_vec) 
%update function for the nonribbon case
[rowdim,coldim]=size(currentS);
changeS=randi([0,1],size(currentS));
changeS=2*changeS-1;    
changeS=changeS.*RibbonMat;
preS=currentS+changeS;
S_new = zeros(length(currentS),3);
for i=1:length(dim_vec)
    S_new(:,i)=min(max(preS(:,i),1),dim_vec(i));
end
end
\end{lstlisting} 

}

\chapter{Code for sequence prediction algorithm}\label{algorithmcode}
In this appendix, we present the Matlab code for the implementation of the  neural sequence algorithm. There are three functions total:
\begin{itemize}
\item run\_Algorithm
\item DeconstructGraph
\item ReconstructCycle
\end{itemize}
 The main function, run\_Algorithm takes the transposed adjacency matrix of a graph as an input and outputs the list of predicted limit cycles for the corresponding network. This function calls two subfunctions for the two phases of the algorithm: DeconstructGraph and ReconstructCycle, also included here. This particular implementation of the algorithm was done in collaboration with Katherine Morrison.

{\footnotesize
\begin{lstlisting}[style=Matlab-editor]
function [ExpectedLimitCyles] = run_Algorithm(sA)
% This function takes a transposed adjacency matrix (sA matrix) as an input and outputs a list of expected limit cycles.
DeletedNodes=[ ];
% This should always be initialized as empty
DeletedNodesList=[ ];
% This should always be initialized as empty
CoreCyclesList=[ ];
% This should always be initialized as empty
[DeletedNodesList,CoreCyclesList] = DeconstructGraph(sA,DeletedNodes,DeletedNodesList,CoreCyclesList);
n=size(sA,2);
ExpectedLimitCycles=zeros(size(CoreCyclesList,1),2*n);
% We will reconstruct a full cycle for every core cycle and we will have every unique core cycle listed as well since we expect unstable fixed points at those core cycles.  We will allow the ExpectedLimitCycles to have length up to 2*n because of the potential for period doubling -- hopefully we won't have any cycles longer than this, but if we do, we'll get an error here
UnstableFixedPoints=unique(CoreCyclesList,'rows');
% Insert all unique core cycles at the end of our ExpectedLimitCycles list
for i=1:size(CoreCyclesList,1)
    DeletedNodes=DeletedNodesList(i,:);
    CoreCycle=CoreCyclesList(i,:);
    FullCycle=ReconstructCycle(sA,DeletedNodes,CoreCycle);
    ExpectedLimitCycles(i, 1:length(FullCycle))=FullCycle;
    % Insert each reconstructed full cycle into a long row padded with zeros
end
ExpectedLimitCycles
end
\end{lstlisting} 

\begin{lstlisting}[style=Matlab-editor]

function [DeletedNodesList,CoreCyclesList] = DeconstructGraph(sA,DeletedNodes,DeletedNodesList,CoreCyclesList)
% sA is the full adjacency matrix
% DeletedNodes is a row vector of the nodes deleted thus far along a single path to to one core cycle
% DeletedNodesList is a matrix that will have rows of length n-3 added to it once a full row vector of DeletedNodes has been completed (this row vector will be padded with zeros if the size of the core cycle is greater than 3
% CoreCyclesList is a matrix that will have rows of length n added to it as each core cycle is found (the core cycle will be padded with zeros to make it length n)
% This function recursively calls itself (in the manner of a depth-first search) until there are no additional nodes that can be deleted and yield a valid graph.  At this point, the vector of deleted nodes is padded with zeros and added as a row vector to DeletedNodesList.  The corresponding core cycle is computed, padded with zeros, and added as a row vector to CoreCyclesList

if nargin<2
    DeletedNodes=[ ];
end
if nargin<3
    DeletedNodesList=[ ];
end
if nargin<4
    CoreCyclesList=[ ];
end
NodesToDelete=FindDeleteNodes(sA,DeletedNodes);
for i=1:length(NodesToDelete)
    [DeletedNodesList,CoreCyclesList] = DeconstructGraph(sA,[DeletedNodes, NodesToDelete(i)],DeletedNodesList,CoreCyclesList);
    % Keep recursively calling the function, adding the next node to delete to the end of the vector DeletedNodes (in the function call but not outside it so that we can loop through and call the function multiple times with different nodes to delete added on the end)
end
n=size(sA,2);
if isempty(NodesToDelete)
    % In this case, we've gotten down to a core cycle after an empty loop
    DeleteRow=zeros(1,n-3);
    DeleteRow(1:length(DeletedNodes))=DeletedNodes;
    % Fill in the first entries of DeleteRow so that the remaining entries are all the padded zeros
    DeletedNodesList=[DeletedNodesList;DeleteRow];    
    CycleRow=zeros(1,n);
    [sAsubmat,labels]=MakeSubmat(sA,DeletedNodes);
    CoreCycleIndices(1)=1;
    % Always start the cycle at an index of 1 (which will have a label corresponding to the lowest remaining node)   
    for j=1:length(labels)-1
        NextIndex=find(sAsubmat(:,CoreCycleIndices(j))==1);
        % This finds the unique node in the core cycle that CoreCyclesIndices(j) feeds into
        CoreCycleIndices(j+1)=NextIndex;
    end    
    CycleRow(1:length(labels))=labels(CoreCycleIndices);
    % Fill in the first entries of CycleRow so that the remaining entries are all the padded zeros
    CoreCyclesList=[CoreCyclesList;CycleRow];
end
\end{lstlisting} 

\begin{lstlisting}[style=Matlab-editor]

function FullCycle=ReconstructCycle(sA, DeletedNodes, CoreCycle)
% sA is the full adjacency matrix
% DeletedNodes is a row vector of all the nodes to  delete in the graph to get down to the core cycle
% CoreCycle is a row vector of the nodes in the core cycle in order according to which node feeds into which
% This function reconstructs the full cycle by appropriately reinserting the deleted nodes in reverse order of when they were deleted.
DeletedNodes=DeletedNodes(DeletedNodes~=0);
% This removes any padded zeros from the end
CoreCycle=CoreCycle(CoreCycle~=0);
% This removes any padded zeros from the end
cycle=CoreCycle;
for i=length(DeletedNodes):-1:1
    % run through the DeletedNodes in reverse order
    NodeToInsert=DeletedNodes(i);    
    [sAsubmat, labels]=MakeSubmat(sA, DeletedNodes(1:i-1));
    % This reconstructs the subgraph when the nodes before the current node have been deleted -- when we're on the last node, i.e. i=1, DeletedNodes(1:i-1) will be empty and we'll just get back the original sA and labels=1:n   
    IdxToInsert=find(labels==NodeToInsert);
    % This finds the index in the submatrix corresponding to the node to be inserted  
    IdxIntoNode=find(sAsubmat(IdxToInsert, :)==1); 
    % This finds all the indices of nodes that feed into the node to be inserted (i.e. the locations of 1s in the row corresponding to the node to be inserted)    
    if length(IdxIntoNode)>=3
        NodesInSubMat = sAsubmat(IdxIntoNode,IdxIntoNode);
        SinkNode = find(sum(NodesInSubMat,1)==0);        
        if length(SinkNode) == 1                 
        IdxIntoNode = SinkNode;
        cycle=InsertNode(cycle, NodeToInsert, labels(IdxIntoNode));              
        else   
        disp(['Error -- vertex ' num2str(NodeToInsert) ' has 3 or more inputs in the subgraph for the delete sequence ' mat2str(DeletedNodes) '.  Check this graph by hand to update algorithm']);
        return
        end    %added
        % At this point, we kill the function via 'return' because the algorithm doesn't know how to handle this situation.  This should be updated once we've seen graphs that have this feature and determine a heuristic for how to handle this situation
    elseif length(IdxIntoNode)==2 && sAsubmat(IdxIntoNode(1),IdxIntoNode(2))==0 && sAsubmat(IdxIntoNode(2),IdxIntoNode(1))==0
        % In this case there are 2 nodes that feed into the NodeToInsert and they are not adjacent to each other, so the algorithm says that the new node should be inserted into the sequence after both incoming vertices
        cycle=InsertNode(cycle, NodeToInsert, labels(IdxIntoNode(1)));
        % This guarantees that we are only placing nodes back into the sequence if they directly follow a node from the core cycle
        cycle=InsertNode(cycle, NodeToInsert, labels(IdxIntoNode(2)));
    elseif length(IdxIntoNode)==2
        % One of the nodes to insert feeds into the other.  The one that is fed into should have the new node inserted after it      
        if sAsubmat(IdxIntoNode(1),IdxIntoNode(2))==1
            % Then node 2 feeds into node 1, so the new node should be inserted after node 1
            cycle=InsertNode(cycle, NodeToInsert, labels(IdxIntoNode(1)));
        else
            % Then node 1 feeds into node 2, so the new node should be inserted after node 2
            cycle=InsertNode(cycle, NodeToInsert, labels(IdxIntoNode(2)));
        end
    elseif length(IdxIntoNode)==1
        if isempty(intersect(CoreCycle,labels(IdxIntoNode))) == 0
        cycle=InsertNode(cycle, NodeToInsert, labels(IdxIntoNode));
        end
    end
end
FullCycle=cycle;
\end{lstlisting}

}

\chapter{Catalogue of $n\le5$ oriented graphs with no sinks}\label{catalogue}
This appendix includes a list of all the oriented graphs on $n\le 5$ vertices with no sinks. The graphs on $n=5$ vertices are sorted by which $n=4$ subgraph(s) they reduce to in the first step of the algorithm from Chapter \ref{algorithm}. We group the $n=4$ graphs into four classes: \#2, 3, and 4 reduce to a three-cycle, \#5 is a 4-cycle, \#6 and 7 have three strong neurons and one weak, and \#8 has two distinct limit cycles. For graphs whose networks exhibit more than one behavior, like \#8, we list the initial conditions for each attractor type.  Table \ref{indextable} gives an index for the classification indicating which category each of the oriented graphs on $n=5$ nodes without sinks fall into.

\begin{table}
\begin{center}
\begin{tabular}{|c|c|c|}
\hline
Reduces to & Graph indices & Attractor type\\
\hline
2/3/4 & 9--38 & AT-1\\
 & 39--42 & AT-11\\ \hline
 5& 43--47 & AT-5\\
 & 48--51& AT-6\\ \hline
6/7 & 52--81 & AT-2 \\
 & 82--99& AT-4\\ \hline
8 & 100--110 & AT-15 \\
 & 111 & AT-10\\
 & 112--113 & AT-16\\ \hline
2 or 2 & 114 & AT-11 \\ \hline
2 or 5 & 115 & AT-11, AT-16\\ \hline
5 or 5 & 116 & AT-13\\ \hline
2/3/4 or 6/7 & 117--130 & AT-2\\
 & 131 & AT-10, AT-11 \\ \hline
5 or 6 & 132 & AT-14\\ \hline
6/7 or 6/7 & 133 & AT-10 \\ 
 & 134--135 & AT-16\\
 & 136--139 & AT-4\\
 & 140--142 & AT-3\\ \hline
2/3/4 or 8 & 143--146 & AT-15, AT-16\\ \hline
6/7 or 8  & 147--148 & AT-3 \\ 
& 149 & AT-4\\ \hline
8 or 8 & 150 & AT-18 \\
 & 151 & AT-17\\ \hline
3, 3, or 3 & 152 & AT-7\\ \hline
4, 6, or 7 & 153 & AT-12\\ \hline
7, 7, or 7 & 154 & AT-12\\ \hline
2, 6, or 8 & 155 & AT-15, AT-16\\ \hline
6, 6,  or 8 & 156 & AT-16\\ \hline
3, 7, or 8  & 157 & AT-15, AT-16\\ \hline
2, 8, or 8  & 158 & AT-15, AT-16\\  \hline
7, 7, 7, 7, or 7 & 159 & AT-9\\ \hline
none & 160 & AT-8\\
\hline
\end{tabular}
\end{center}
\caption{Index for classification of oriented graphs with no sinks on $n=5$ nodes: The left column gives the indices of the possible $n=4$ subgraph(s) that appear after step one of the algorithm, the middle column shows the indices of the graphs that have this reduction (based on the classification indexing), and the right column shows which attractors these graphs have (based on the dictionary). Subgraphs 2, 3, and 4 are grouped together since they have the same behavior and subgraphs 6 and 7 are similarly grouped.}
\label{indextable}
\vspace{2in}
\end{table}

\newpage

\begin{center}
\includegraphics[width=180mm]{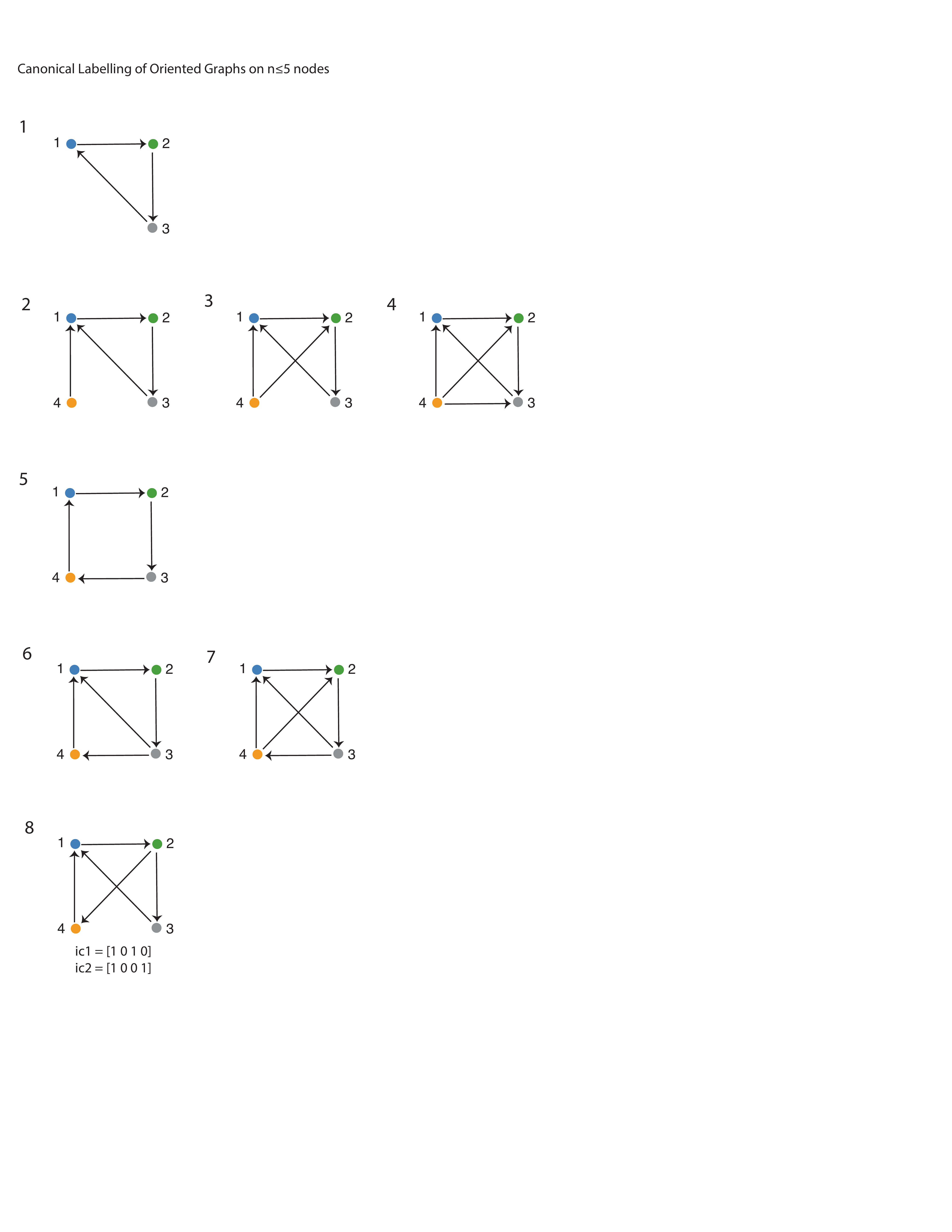}
\end{center}
\newpage
\begin{center}
\includegraphics[width=150mm]{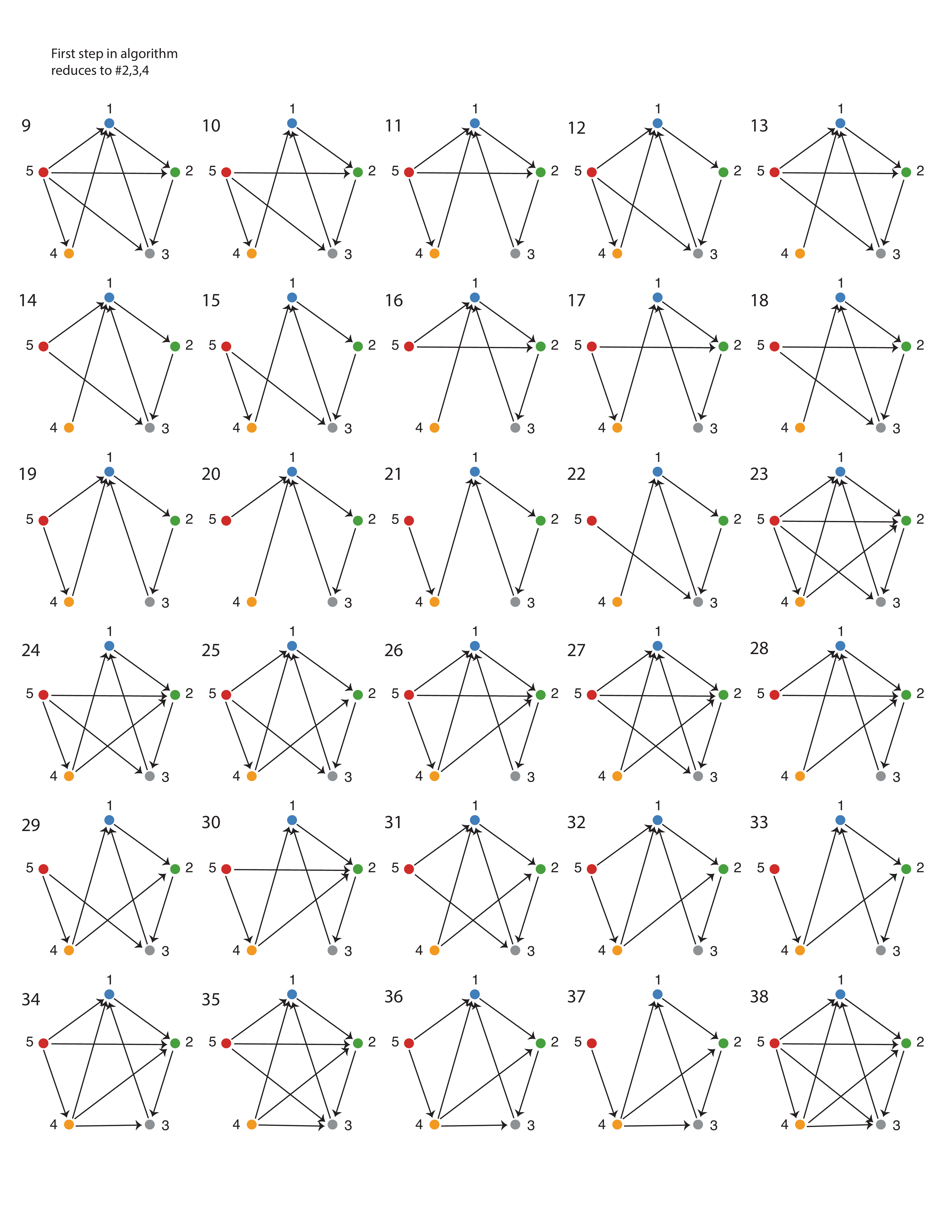}
\end{center}
\newpage
\begin{center}
\includegraphics[width=150mm]{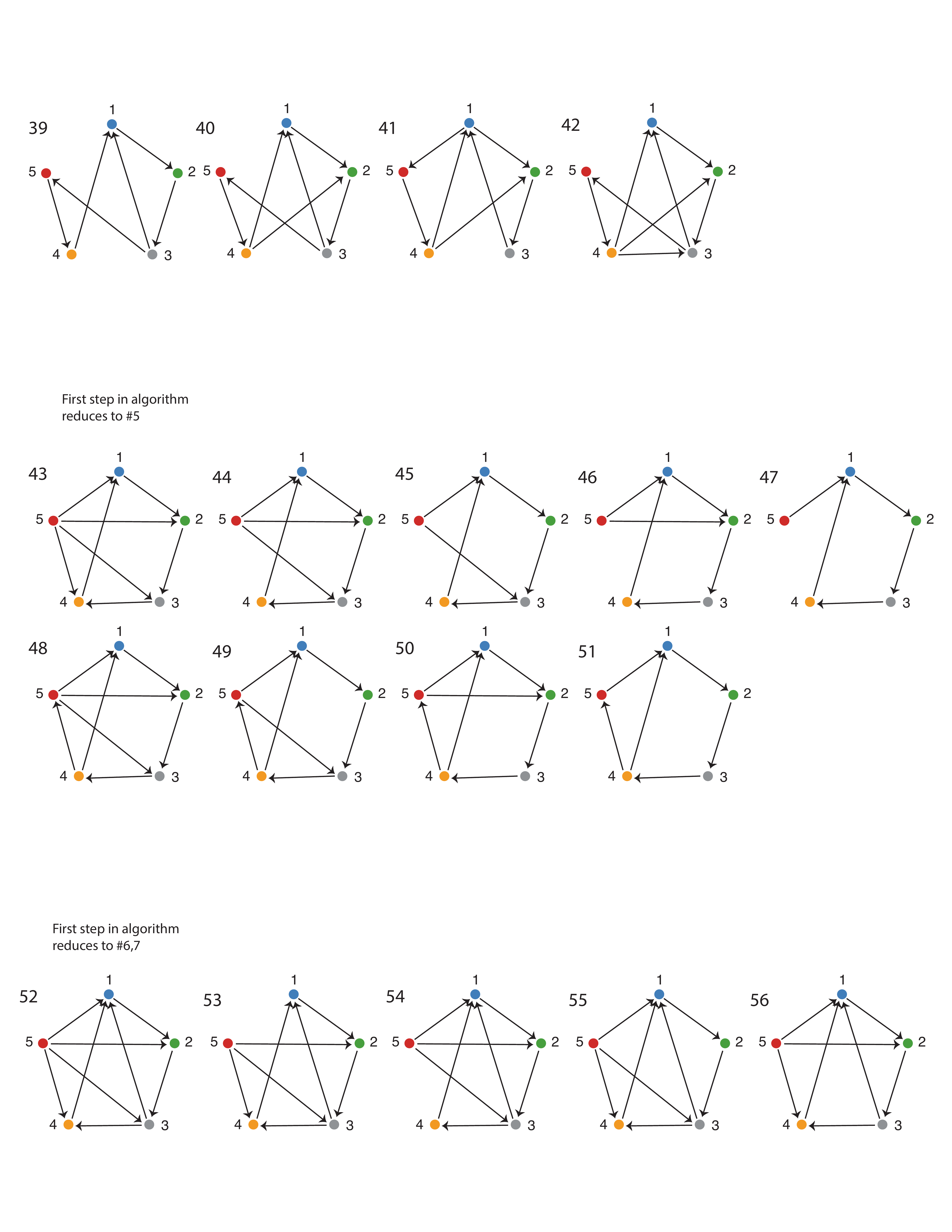}
\end{center}
\newpage
\begin{center}
\includegraphics[width=150mm]{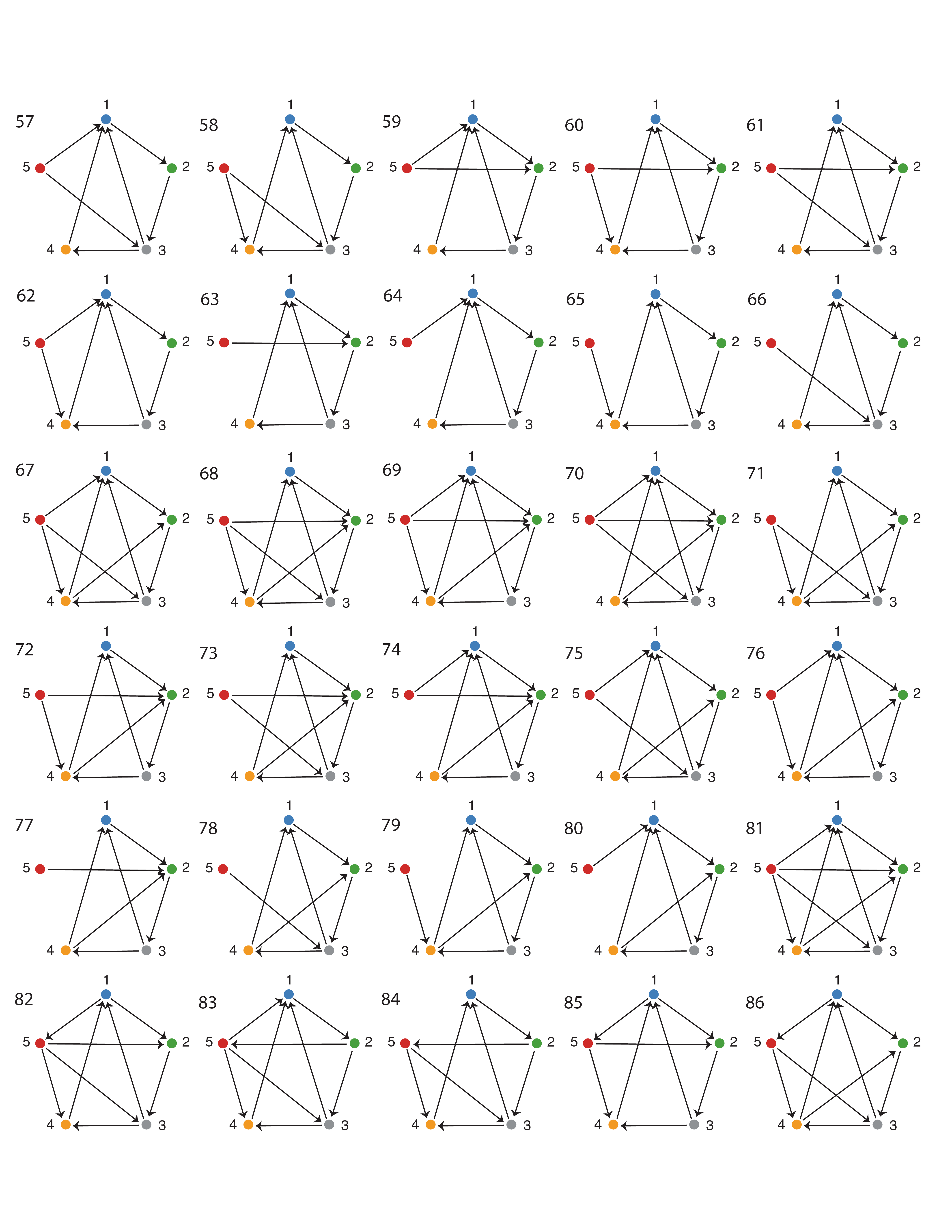}
\end{center}
\newpage
\begin{center}
\includegraphics[width=150mm]{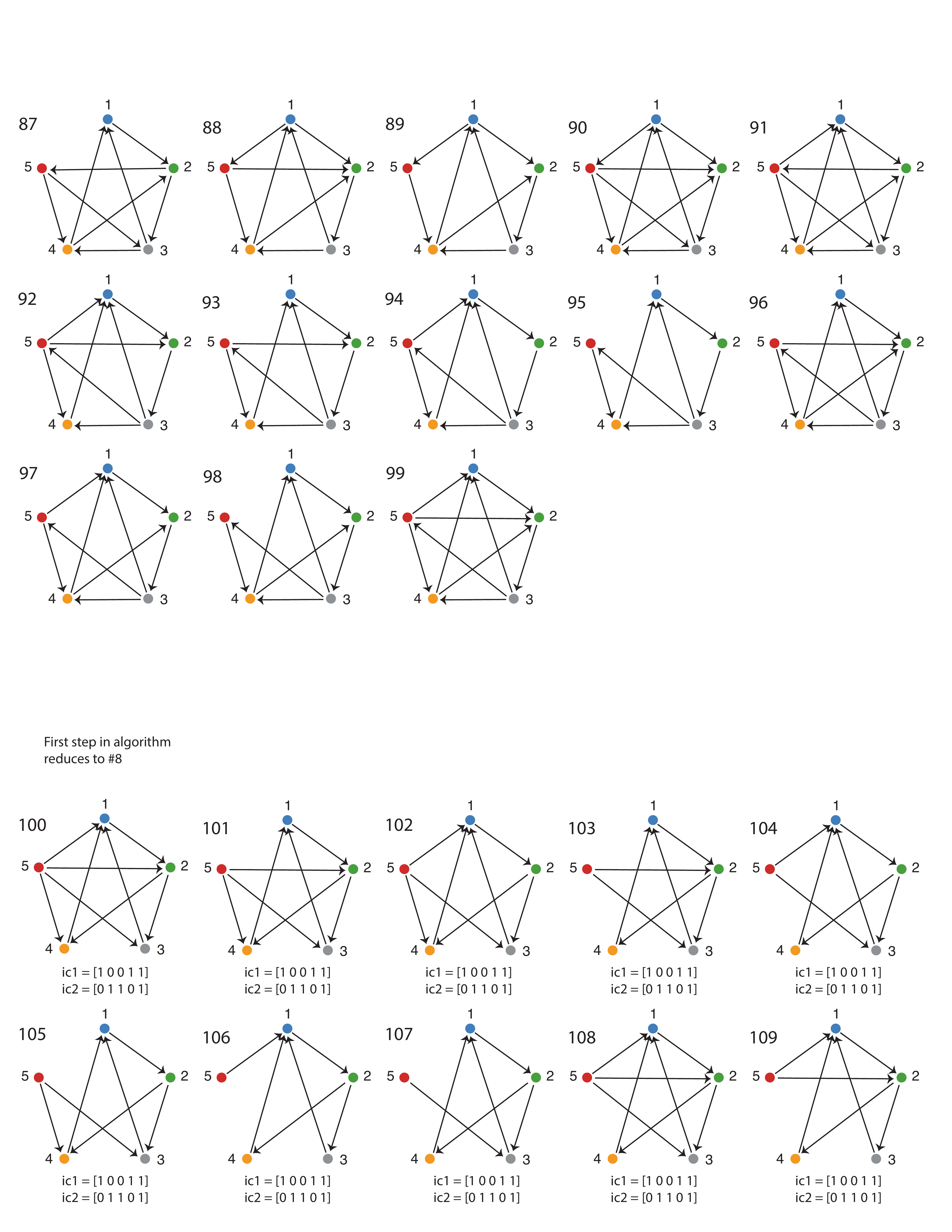}
\end{center}
\newpage
\begin{center}
\includegraphics[width=150mm]{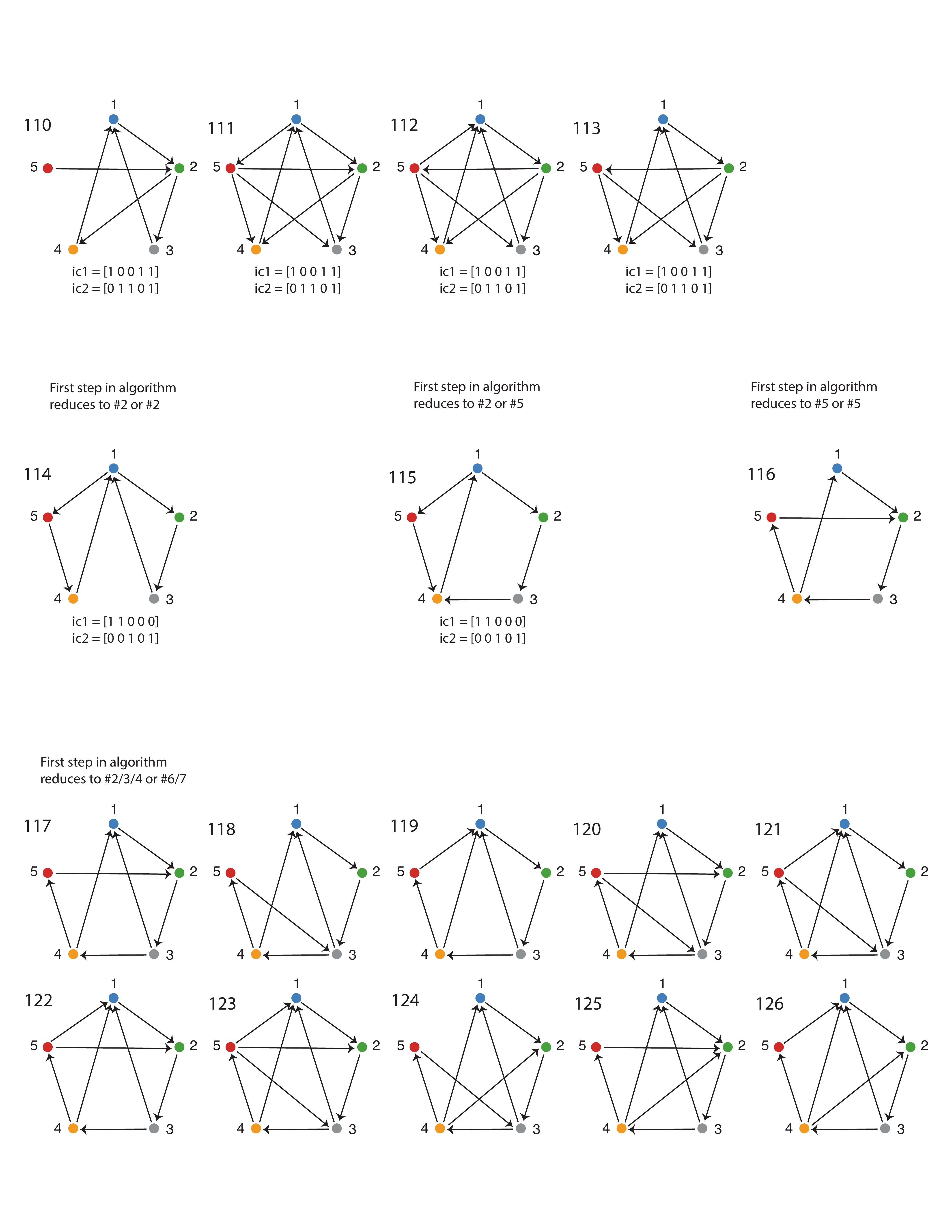}
\end{center}
\newpage
\begin{center}
\includegraphics[width=150mm]{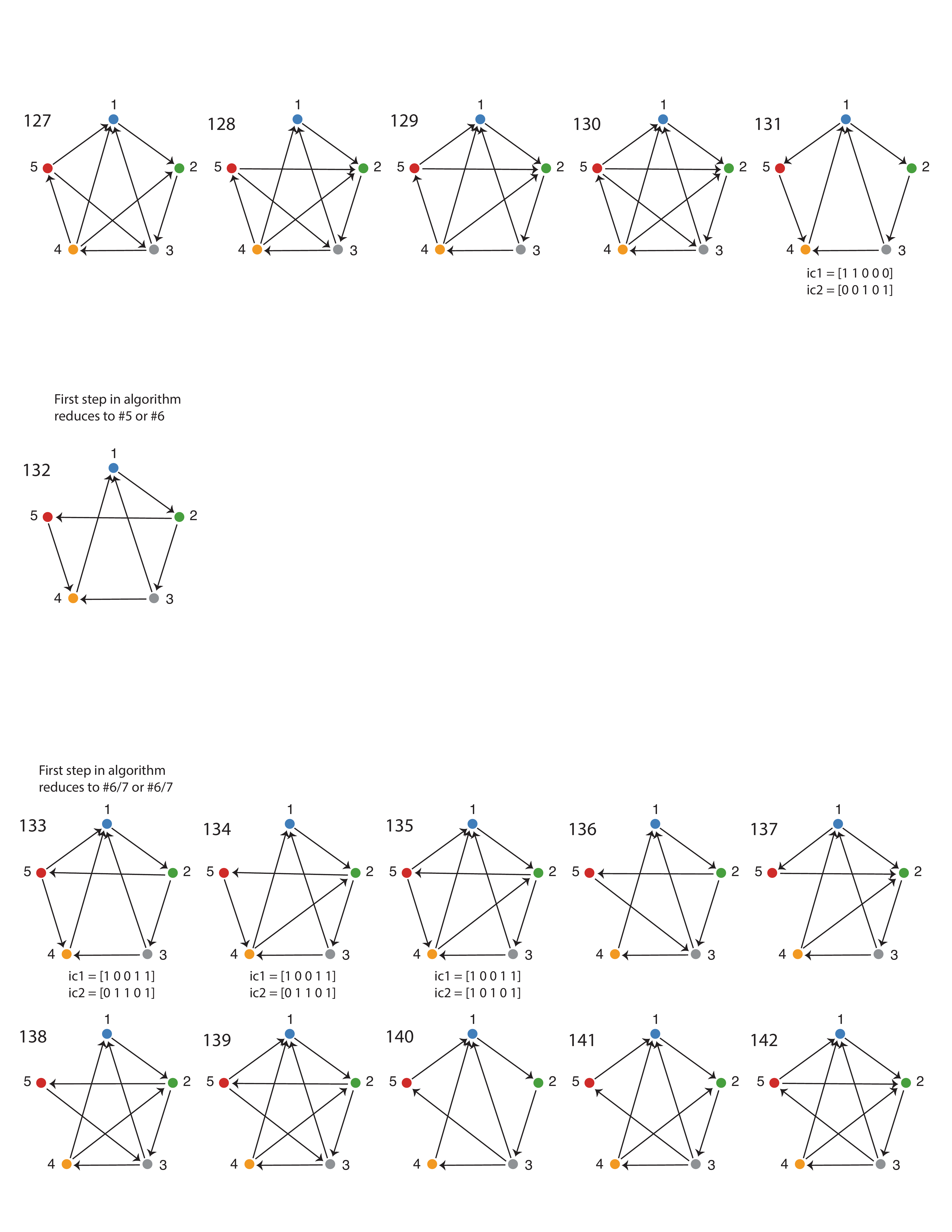}
\end{center}
\newpage
\begin{center}
\includegraphics[width=150mm]{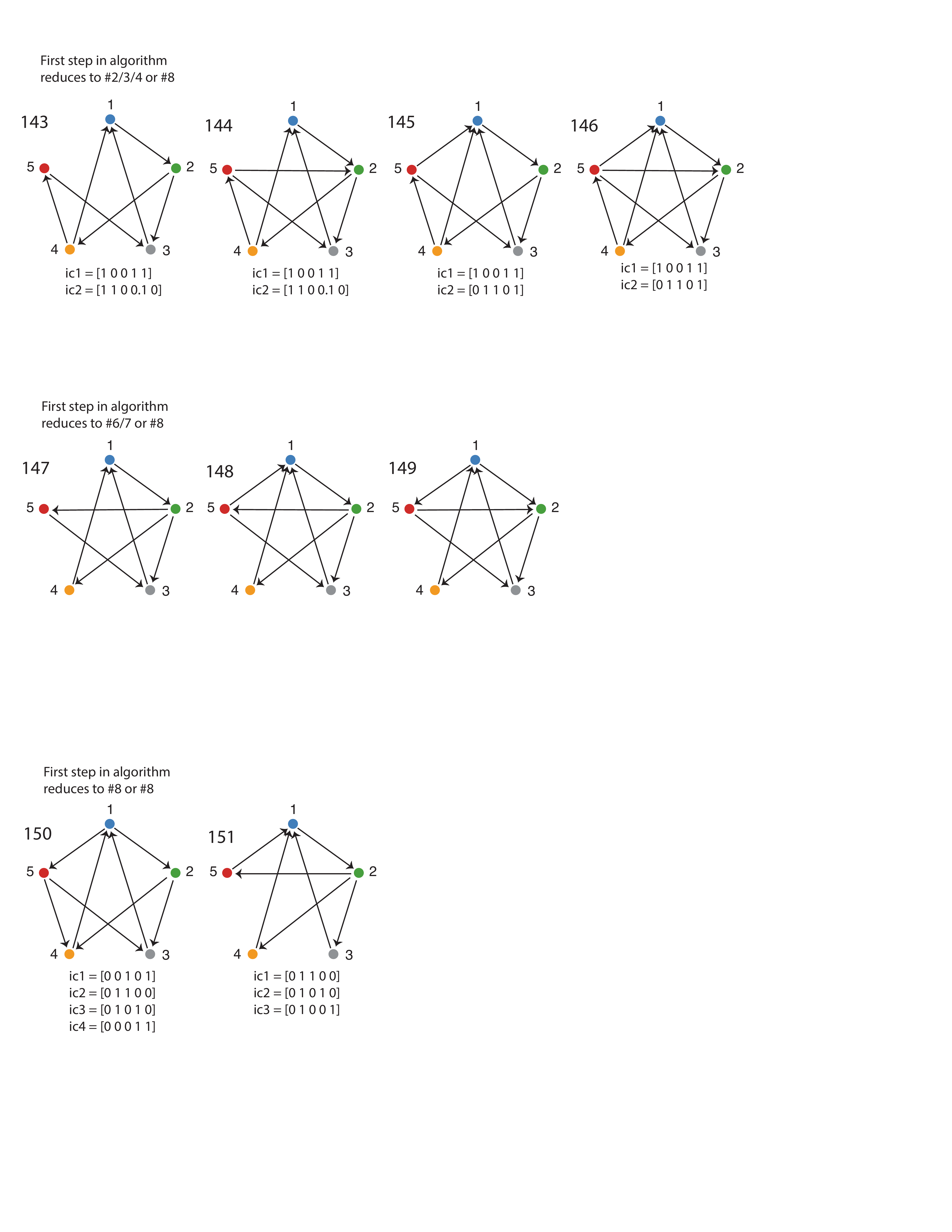}
\end{center}
\newpage
\begin{center}
\includegraphics[width=150mm]{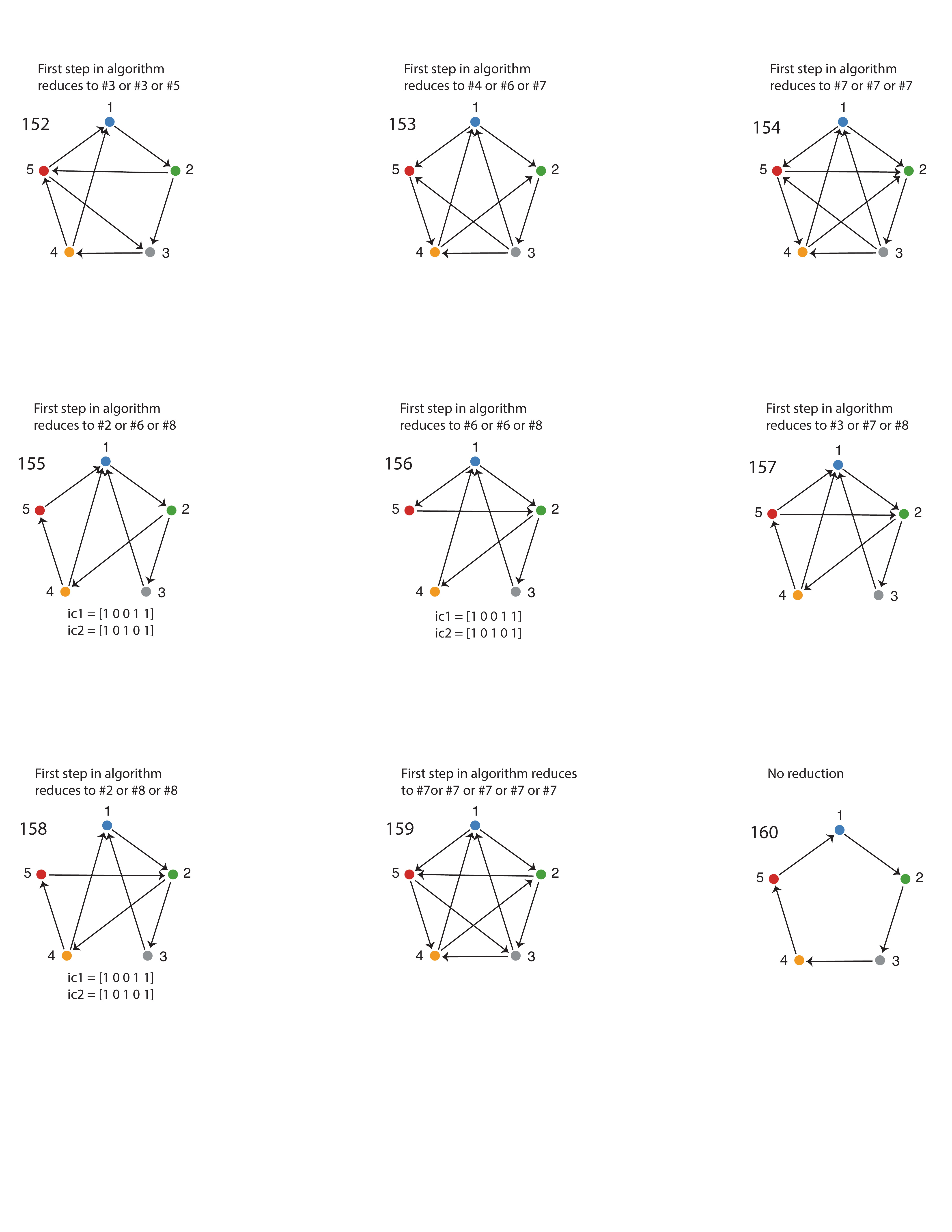}
\end{center}




\bibliographystyle{plain}
\bibliography{bibrefs}

\end{document}